\documentclass[cmp]{svjour}

 \pdfoutput=1

\usepackage{amssymb}
\usepackage{bm}
\usepackage{amsmath}
\usepackage{amsfonts}
\usepackage{mathrsfs}
\usepackage{units}
\usepackage{cite}

\usepackage{xcolor}

\newcommand{\p}{\partial}

\newcommand{\dd}{{\rm d}}

\usepackage{trimclip}

\makeatletter
\newcommand{\twoheaddownarrow}{\mathrel{\mathpalette\twoheaddownarrow@\relax}}
\newcommand{\twoheaddownarrow@}[2]{%
  \begingroup
  \sbox\z@{$\m@th#1\downarrow$}%
  \vphantom{\copy\z@}
  \ooalign{\copy\z@\cr\hidewidth\clipdownarrow@{#1}\hidewidth\cr}%
  \endgroup
}
\newcommand{\clipdownarrow@}[1]{%
  \raisebox{0.25\ht\z@}{\clipbox{-3pt 0pt -3pt {0.5\height}}{\copy\z@}}%
}
\newcommand{\twoheaduparrow}{\mathrel{\mathpalette\twoheaduparrow@\relax}}
\newcommand{\twoheaduparrow@}[2]{%
  \begingroup
  \sbox\z@{$\m@th#1\uparrow$}%
  \vphantom{\copy\z@}
  \ooalign{\copy\z@\cr\hidewidth\clipuparrow@{#1}\hidewidth\cr}%
  \endgroup
}
\newcommand{\clipuparrow@}[1]{%
  \raisebox{0.1\ht\z@}{\clipbox{-3pt {0.65\height} -3pt 0pt}{\copy\z@}}%
}
\makeatother
\begin{document}

\title{Spacetimes as Continuous Posets: \\ Old Foundations and New Developments}


\author{ E. Minguzzi\thanks{ORCID ID: 0000-0002-8293-3802 \email{ettore.minguzzi@unipi.it}}}

\institute{Dipartimento di Matematica, Universit\`a degli Studi di Pisa,  Largo
B. Pontecorvo 5,  I-56127 Pisa, Italy.}

\date{}

\maketitle


\date{}
\maketitle%

\begin{abstract}
\noindent  We introduce the reader to the continuous posets approach to spacetime geometry, a topic pioneered by Martin and Panangaden (MP), supplying the relevant domain-theoretic proofs to ease the transition for readers from Lorentzian geometry and general relativity. Then, we present new results for the spacetime interpretation. Identifying the way-below relation with the chronological relation \(I\), we determine the unique choice of \(\le\) and of causality condition yielding each poset property: continuous posets require \(\le\,=D_p\) and ``past-distinction, supremum compatibility (SC) and future-reflectivity'', bicontinuous posets are causally continuous spacetimes with \(\le\,=D\), and globally hyperbolic posets are precisely globally hyperbolic spacetimes with \(\le\,=J\), a necessity result beyond MP's sufficiency. Causal simplicity is recovered through a chain realizability condition, stable causality through the Hausdorffness of a certain topology \(\tau\), so the causal ladder is fully determined at the abstract level, solving a problem posed by MP.
Being phrased entirely in terms of a poset, the theory is of low-regularity. We introduce a notion of topological Kronheimer--Penrose causal space general enough to encompass the Lorentzian length spaces in the literature, and give weak conditions making it a (bi)continuous poset.
Domain-theoretic constructions then apply directly, e.g., completion schemes yield spacetime boundaries: Lawson's round-ideal completion, the most natural for preserving continuity, is proved to equal the future GKP completion, as is the directed completion recently studied by Gigli et al., for which we remove the (SC) assumption under strong causality. Reversing the logical order of the theory, we characterize those relations that induce a continuous poset having them as way-below relation. Finally, physical considerations on the key role of past-reflectivity in black hole evaporation lead us to suggest that spacetime is, at the fundamental level, a co-continuous poset.
\end{abstract}

\newpage
\setcounter{tocdepth}{3}

\tableofcontents

\section{Introduction}
Twenty years ago Martin and Panangaden (MP) interpreted globally hyperbolic spacetimes $(M,g)$ in the framework of domain theory \cite{martin06} by showing that they are bicontinuous posets $(M,\le)$. The present paper has  four parts and purposes.
\begin{enumerate}
\item Firstly, it is an introduction to continuous poset theory aimed at bridging the gap for readers in Lorentzian geometry interested in understanding the continuous poset interpretation. In fact, while the contribution by Martin and Panangaden was highly regarded, it received perhaps less attention, and fewer follow-up works, than it deserved (see however \cite{keimel09b,ebrahimi14,ebrahimi15,finster21,sharifzadeh19,mazibuko23}). In part this could be attributed to the dry and synthetic presentation of the necessary elements of domain theory. The construction of the way-below relation $\ll$ may seem, at first encounter, somewhat abstract. Of course its study will definitely pay off but its beauty can remain hidden on first approach. In order to ease the passage,   we follow \cite{martin06} but where they just listed the needed definitions and properties from domain theory, we provide complete proofs. One reason to revise this material is also that it is often presented under unnecessarily strong assumptions. While standard sources in domain theory work under a variety of conditions, such as directed complete posets (dcpo), domains, lattices,  we shall only use continuity progressing slowly to weak bicontinuity and bicontinuity  if needed. Hopefully, reading this work will help the reader in more easily building understanding and intuition of the  mathematics. The main source for continuous poset theory will be the book by Gierz et al. \cite{gierz03} (particularly, Secs.\ O-1, O-2, O-5. I-1). It should also be mentioned that in this part standard results are expanded to  include substantial original material including the study of closure and interior in several topologies of notables sets of interest and a discussion of the conditions implying  coincidence of Lawson and interval topologies.

\item Secondly, it is aimed at developing the theory applied to regular spacetimes $(M,g)$ starting from where Martin and Panangaden left it. 

They formulated their whole theory  under the assumption that the poset of interest on spacetime is $(M,J)$ where $J$ is the causal relation, for they restricted the whole of their analysis to this case. The experienced relativist knows that, in fact, there are some other important relations which contain $J$, and so, it would be natural to consider them as candidates for $\le$. Indeed, posterior authors realized that the theory could be applied to different posets, for instance, given the importance of the $K$ relation for spacetime (the smallest closed and transitive relation containing $J$, cf.\ \cite{sorkin96,minguzzi18b})  Neda Ebrahimi \cite{ebrahimi14,ebrahimi15} applied the theory to the poset $(M,K)$, $\le \, =K$. However, having to change the identification of $\le$, it seems more natural to search for the causal relation to identify with $\le$ in order to obtain as way-below relation the chronological order: $\ll\,=I$. This is one of the  problems we consider in this work. It seems extremely natural as the way-below relation is used to assign the poset with a topology, the interval topology, similarly to what one does in Lorentzian geometry with the Alexandrov interval topology. We completely solve this problem.

As it turns out, the answer identifies $\le$ with the $D$-relations $D_p$, $D_f$ and $D=D_p\cap D_f$ previously introduced and studied by the author \cite{minguzzi07b,minguzzi07e,heveling21,minguzzi22}\cite[Sec.\ 4.1]{minguzzi18b}. That is, continuous posets demand the choice
\[
\le\,=D_p=\{(x,y): I^-(x)\subset I^-(y)\}
\]
and are associated with the peculiar combination of causality conditions ``past-distinction, supremum compatibility (SC) and future-reflectivity''. Since the pair $(M,D_p)$ is a poset $D_p$ must be antisymmetric, a condition equivalent to past-distinction.

Dually, co-continuous posets demand the choice
\[
\le\,=D_f=\{(x,y): I^+(x)\supset I^+(y)\}
\]
 and are associated with the peculiar combination of causality conditions ``future-distinction, infimum compatibility (IC) and past-reflectivity''.

Bicontinuous posets are causally continuous spacetimes and $\le\,=D$.

 Finally, the globally hyperbolic posets are precisely the globally hyperbolic spacetimes with $\le =J$, that is, we prove a necessity result beyond the sufficiency result by Martin and Panangaden.

If one is interested in the continuity of the poset $(M,J)$ or $(M,K)$ the causality condition $J=D_p$ (resp.\ $K=D_p$) must be imposed which often simplifies some other imposed conditions. Similarly, for the property of co-continuity or bicontinuity.
\item Thirdly, it aims to  apply the continuous poset theory to spacetimes of low regularity. Understanding spacetimes under low regularity is fundamental for the transition to quantum theories of spacetime and gravity.   The continuous poset theory is a very natural, perhaps optimal, framework of low regularity as the manifold structure is completely removed and the topology and way-below order $\ll$ are deduced from $\le$. If we could frame the causal structure of low regularity spacetimes within this theory we could readily import many results from domain theory to it. We shall see one such application with the notion of completion and spacetime boundary.

    While the theory of continuous posets is well defined, here the uncertainty is more on the spacetime side. It is less clear what a low regularity spacetime should be, and a few proposals have been advanced. In order to include all of them in the analysis we start with a pretty general causal structure which we call {\em topological Kronheimer and Penrose causal space} (topological KP-space for short). It is shared by the Lorentzian length spaces of Kunzinger-S\"amann \cite{kunzinger18,grant18,burtscher25,beran25} and the Lorentzian length spaces of Bykov-Minguzzi-Suhr  \cite{minguzzi22,minguzzi24b,minguzzi26}. This is not entirely surprising as these models differ more in the role and property of the function $d$ than in the causality they induce.

    These topological KP-spaces are purposely left with minimal properties, even topologically. We then proceed adding properties to them in order to recover the correspondences between causality properties and poset continuity properties already found in the regular theory. It will be necessary to impose a topological lower approximation $x\in \overline{I^-(x)}$ to obtain the continuity property. This is very nice since the latter reads $x=\sup \{y: y \ll x\}$, that is, there is correspondence between topological approximation on the spacetime side and order-theoretic approximation on the poset side. It is also  satisfying that  this result does not require  also an upper approximation $x\in \overline{I^+(x)}$.

    With these  results at hand we are able to prove that for both types of Lorentzian length spaces mentioned, lower topological approximation, past-distinction, (SC) and future reflectivity imply the continuity  of the poset $(X,D_p)$ and the identification $\ll\,=I$. Two-sided approximation and causal continuity imply the bicontinuity property.

    Once these properties are settled we can import results from domain theory. A natural application is that of completions. There are a few schemes and fortunately there is one that preserves continuity: the Lawson round-ideal completion. With little effort, given the previously developed correspondence, we show that the round-ideal completion coincides with Geroch, Kronheimer and Penrose (GKP) future completion.

\item Fourthly, it aims at recovering, at the abstract level of posets, the causal ladder of spacetimes, and at clarifying what has to be postulated in order to enter the theory at all.

Global hyperbolicity and causal continuity are visible to the poset, being respectively interval-compactness and bicontinuity, but causal simplicity is not: under bicontinuity (causal continuity) we have necessarily $\le\,=D=\bar J$,  the order is  closed and so the closedness of $J$ does not seem to be easily expressible. We show that the missing rung can be recovered order theoretically through a chain realizability condition, namely the demand that any two related points be joined by a chain that is compact and connected in the interval topology. The relevant results are proved at the level of Nachbin  semiclosed ordered spaces, where they do not involve the approximation relations at all, and they show that for compact intervals chain realizability, density of the order and connectedness of the intervals are one and the same condition. This also shows that global hyperbolicity for the poset has to be improved with respect to the MP definition, by adding a density condition which we show to be equivalent to connectedness (besides compactness) of the order intervals.

Finally, in the whole paper the order $\le$ is the primitive datum and $\ll$ is derived from it, while in the spacetime application one often might want to start from the chronological relation. We therefore reverse the logical order and characterize, in purely relational terms, those relations $R$ on a set for which the induced order returns a continuous poset with $\ll\,=R$. It turns out that five conditions suffice, of which two are the familiar past-distinction and future-reflectivity, so that a set endowed with such a relation is, to all effects, a continuous poset.
\end{enumerate}

\begin{remark}
Approximately twelve days before this paper was posted (arXiv: 2608.\-19502v1), while our work was being finalized, a paper by Gigli et al. \cite{gigli26} appeared (arXiv: 2608.04718v1). In that paper, the authors show that, under past-distinction, two-sided topological approximation, and an (SC) condition, the directed completion \cite{markowsky76,zhao10,gigli25}, when applied to Lorentzian pre-length spaces, coincides with the Geroch–Kronheimer–Penrose future completion.
%

Their work passes through some propositions on the existence of suprema of directed sets that are also to be found in the present work, Prop.\ \ref{mmtb} or part of Prop.\ \ref{vpoj} (they use the notation $\preceq_{I^-}$ for $D_p$). The reason is that, for their completion, they need some of the same ingredients that we need for establishing continuity: although the motivations were different, we passed through similar concepts. In particular, they isolate the following {\em supremum-compatibility condition}: for every directed set $D$ with a supremum
\begin{equation}\label{eq:SC}
I^-(D)=I^-(\sup D).
\tag{SC}
\end{equation}
In domain-theoretic language, (SC) states precisely that the map $X\to \mathcal{P}(X)$, $x\mapsto I^-(x)$, where $\mathcal{P}(X)$ is endowed with the inclusion order, is Scott continuous (see Sec.\ \ref{ccpd}, where we recall equivalent forms of Scott continuity). This property is particularly useful to express results for continuous posets under low causality conditions, and the letters SC can conveniently be read both as supremum compatibility and as Scott continuity.

They show that (SC) is implied by global hyperbolicity or forward  completeness, two properties that are quite strong. We show that strong causality really suffices while if future reflectivity is included, distinction suffices (hence very close to the property of past-distinction that should be assumed anyway to have a poset). We also prove that (SC) cannot be dispensed with through a very peculiar counterexample (Example \ref{suex}).

%
We also worked to clarify the relation between our round-ideal approach  to completion and their directed completion. We proved  that without future reflectivity the round-ideal completion can  still be applied but the round-ideals are not the indecomposable past sets in the sense of GKP.  However, there is an isomorphism between the GKP family of IPs  and the  family of round-ideals which becomes the identity precisely under future reflectivity. It is this  round-ideal framework that connects via a result by Keimel-Lawson \cite[Prop.\ 9.3]{keimel09} to the directed completion.

In the end, we find that imposing (SC) makes several completions coincide with the GKP completion.
This robustness of the GKP future causal boundary might have been  hinted at by work by Harris \cite{harris98}.
\end{remark}

\medskip
The paper is organized as follows. Sec.\ \ref{cpto} is a self-contained introduction to continuous poset theory, closely following Gierz et al.\ \cite{gierz03} but with complete proofs: posets, ideals and the Scott topology; the way-below relation $\ll$ and its basic calculus (Thm.\ \ref{kctt}); continuous and bicontinuous posets and the interpolation property (Thm.\ \ref{oovg}, together with the finite-set strengthening of Thm.\ \ref{tncw}); the Lawson and interval topologies; and, finally, globally hyperbolic posets, for which we show that every upper-bounded directed set converges to its supremum (Prop.\ \ref{mmcl}), strengthening the corresponding result of Martin and Panangaden \cite{martin06}.

A subsection is devoted to the non-decreasing functions, that is to the
order-theoretic version of the theory of time functions. We introduce two  topologies built from the way-below relation alone, $\tau_s$ and $\tau$ of Defs.\ \ref{cctp} and \ref{cctpn}, and show that Peleg's theorem provides a continuous function strictly increasing along $\ll$ (Thm.\ \ref{pelth}), while Levin's theorem, applied to the order in the Lawson topology, recovers the order itself from a countable family of continuous utilities (Thm.\ \ref{levth}). The passage from a semi-time to a time function is obtained by feeding Levin's theorem with the order-theoretic counterpart $\preceq_f$ of the relation $D_f$, which is proved to be closed in $\tau\times \tau$ on {\em every} continuous poset (Thm.\ \ref{prcfcl}, Thm.\ \ref{levthf}). This last result is later used to identify the right notion of stably causal continuous poset.

Sec.\ \ref{qcft} specializes this machinery to a regular spacetime $(M,g)$. We show that $\le\,=D_p$ is the necessary and sufficient choice for $(M,\le)$ to be a continuous poset with way-below relation $I$, precisely under past-distinction, (SC) and future-reflectivity (Thm.\ \ref{kssw}); the corresponding statement for $\le\,=J$ requires in addition the closure of the causal pasts. We then identify causally continuous spacetimes with bicontinuous posets $(M,D)$ (Thm.\ \ref{bxxd}), and, as a corollary, causally simple spacetimes with bicontinuous posets $(M,J)$. The section closes by identifying globally hyperbolic spacetimes with globally hyperbolic posets $(M,J)$, adding a converse to the corresponding result of \cite{martin06}.

Sec.\ \ref{s:ladder} isolates the order-theoretic content of the top of the causal ladder. We introduce the chain relation $\preceq_{c}$ of Def.\ \ref{def:chainrel} and prove, for semiclosed ordered spaces with compact intervals and dense order, that any two related points are joined by a compact connected chain (Thm.\ \ref{thm:ghimpliescs}), a result close to Theorem 4 of Ward \cite{ward54}; Thm.\ \ref{thm:densityconnected} then shows that chain realizability, density and connectedness of the intervals are equivalent conditions, so that in particular the Lawson and interval topologies agree on a globally hyperbolic poset (Thm.\ \ref{thm:ghlawson}), while Example \ref{keim} shows that compactness of the intervals alone does not suffice. Prop.\ \ref{prop:spacetimechain} verifies that chain realizability together with  bicontinuity is, for a spacetime, precisely causal simplicity as it automatically forces $\le\, =J$, so that causally simple continuous posets are also identified.

The section closes with a five-rung causal ladder for continuous posets
(Table \ref{tab:ladder}), whose levels are shown to be all distinct. Global hyperbolicity for posets requires the inclusion of a density property with respect to Martin and Panangaden's proposal. It is also argued that the ladder is complete, as any other level would also introduce a new level in the ladder for regular spacetimes. Besides the three classical top rungs it contains a further one, ``stable causality and future reflectivity'', characterized by continuity of the poset together with Hausdorffness of the topology $\tau$ (Thm.\ \ref{prop:spacesta}); this is precisely the condition under which the results of Sec.\ \ref{ndfn} provide continuous increasing functions. The identification is further confirmed by showing that it implies a non-partial imprisonment property for maximal chains, similarly to Lorentzian geometry (Thm.\ \ref{thm:escape}). By adapting an averaging technique by Burtscher and Garc\'{\i}a-Heveling \cite{burtscher25}, a final result shows that under a countability condition globally hyperbolic posets admit Cauchy increasing functions (Thm.\ \ref{cauchytime}).

Sec.\ \ref{rjgh} extends the  correspondence to low regularity. We introduce the notion of a topological KP-space, general enough to encompass both the Lorentzian length spaces of Kunzinger and S\"amann \cite{kunzinger18} and the Lorentzian metric/length spaces of Bykov, Suhr and the author \cite{minguzzi22,minguzzi24b,minguzzi26}, and show that a topological lower-approximation condition, together with (SC) always forces the chronological relation to sit inside the way-below relation, with equality under future-reflectivity (Thm.\ \ref{iiskkkp}); continuity and bicontinuity of the resulting poset follow.
Secs.\ \ref{mwcf} and \ref{ksll} specialize these results, respectively, to BMS-Lorentzian metric/length spaces and to KS-Lorentzian pre-length spaces. The last subsection identifies, in this general setting, indecomposable past sets, chronological pasts of countable timelike chains or of directed sets, Scott-open ideals, and round-ideals as one and the same family (Thm.\ \ref{ipeqv}).

Sec.\ \ref{cmpl}, finally, turns to completions. After recalling the (plain) ideal completion and the directed completion recently studied by Gigli \cite{gigli25} we turn to the round-ideal completion of Lawson and Keimel--Lawson \cite{lawson97,keimel09}, the continuity-preserving refinement of the ideal completion, and mention that it also embeds the way-below relation itself, not merely the order (Thm.\ \ref{roundidealprop}). The main result of the section identifies the round-ideal completion of $(M,D_p)$ with the future causal boundary of Geroch, Kronheimer and Penrose \cite{geroch72}, the added ideal points being exactly the terminal indecomposable past sets (Thm.\ \ref{ipmatch}).

Sec.\ \ref{wbfound}, finally, reverses the logical order of the whole paper. Taking a relation $R$ on a set $X$ as the only datum and defining $x\le y$ by $R^{-}(x)\subset R^{-}(y)$, we characterize in Thm.\ \ref{wbaxioms} the five conditions (T), (D), (Bd), (fR) and (SC) under which $(X,\le)$ is a continuous poset with way-below relation exactly $R$, and we show that dropping (fR) still leaves a continuous poset, but one in which $R$ is only contained in the way-below relation. Conditions (T) and (D) say that $(X,R)$ is an abstract basis, (Bd) and (fR) are past-distinction and future-reflectivity, and (SC) is the supremum-compatibility condition.
\medskip

Beyond these mathematical developments, there is one key physical suggestion. A study by Lesourd \cite{lesourd19}, building on earlier work by Kodama \cite{kodama79} and Wald \cite{wald84b}, has shown that the causal structure of an evaporating black hole cannot be causally continuous (see also \cite{minguzzi20}). In particular, black hole evaporation undermines global hyperbolicity and, with it, predictability. However, we proved that past reflectivity is sufficient to infer, from the existence of trapped surfaces, the occurrence of future gravitational collapse singularities \cite{minguzzi20}. This salvages the physical content of Penrose's original singularity theorem—which required global hyperbolicity—and thereby establishes compatibility between quantum field theory, which predicts black hole evaporation, and general relativity, under the relatively weak and time-asymmetric condition of past reflectivity.

We had considered whether the emergence of past reflectivity might indicate that a weaker, time-symmetric assumption could yield the same conclusions. Yet, when past reflectivity and future-distinction are paired within the framework of co-continuous (i.e., dual continuous) posets, Thm.\ \ref{ksswdual}, a more intriguing possibility arises:

\begin{quote}
\emph{At the fundamental level, spacetime is a co-continuous poset.}
\end{quote}

Indeed, the structure cannot be that of a bicontinuous poset, since that property implies causal continuity, which is violated in physically realistic gravitational collapse scenarios. Moreover, it can be shown that, in the presence of black hole evaporation, past reflectivity is responsible for the formation of future singularities, while simultaneously causing future reflectivity to be violated (see \cite[Thm.\ 1]{minguzzi20}). Consequently, the only poset framework of remaining interest is that of co-continuous posets. In short, the time-asymmetric condition of past reflectivity may be the key to uncovering the fundamental time-asymmetric deep structure of spacetime, which would then be expressed by co-continuous posets.

This interpretation aligns with hints from cosmology. Under the condition of past reflectivity, the round-ideal boundary that can be attached to the spacetime is the GKP past boundary, rather than the future one. Such a structure suggests that the universe could have had a regular beginning, in stark contrast to its far future. That picture fits perfectly with the expectation, strongly advocated by Roger Penrose, that the entropy at the origin of the universe must have been extremely low, corresponding to a gravitationally ordered configuration, so as to account for the second law of thermodynamics.

If this is indeed the case, it is unfortunate that standard mathematical poset theory develops along lines dual to what physics would require, focusing on the way-below relation rather than the way-above relation. However, since passing from a result to its dual is generally straightforward, this minor inconvenience need not be overstated.

\begin{remark}[Notations and terminologies]
In this work $\le$  will denote an abstract preorder (transitive and reflexive), while $\ll$ will denote the way-below relation associated with it. The causality relations for a spacetime $(M,g)$ will always be denoted with letters, e.g.\ $J, I, K, D$.

Our subset symbol is reflexive, $S\subset S$ (we do not use $\subseteq$).
Relations are composed in the usual functional order: for relations $R,S$ on $X$,
\[
  R\circ S:=\{(x,y): \exists z\in X\ \text{ with }\ (x,z)\in S\ \text{ and }\ (z,y)\in R\},
\]
so that in $R\circ S$ the relation $S$ is applied first.
\end{remark}

\section{Continuous poset theory}  \label{cpto}

\begin{definition}
A preorder $\le$ is a reflexive and transitive relation. A partial order (or order) is an antisymmetric preorder ``$x\le y$ and $y\le x \Rightarrow x=y$''. A preordered set  $(X,\le)$ is a set endowed with a preorder.   A {\em poset} $(X,\le)$ is a set endowed with a partial order.
\end{definition}

We write indifferently $x\le y$ or $y\ge x$.

The reverse (opposite) order is defined by $x\le^{\textrm{op}} y$ if $y\le x$. Given an order-theoretic notion, its dual refers to the notion obtained replacing $\le$ with $\le^{\textrm{op}}$.

\begin{definition}
Let $(X,\le)$ be a preordered space and let $S\subset X$ be a subset. An upper bound  of $S$ is an element $u\in X$ such that $s\le u$ for every $s\in S$.  A supremum of $S$ is an upper bound $v$ such that for every upper bound $u$, $v\le u$. If it exists and is unique it is denoted $\sup S$ or $\bigvee S$ (also called {\em least upper bound}).  The supremum of $S=\{x,y\}$ is also denoted $x\vee y$ and called {\em join}.
\end{definition}

In a poset the supremum of a set is unique when it exists (otherwise $\sup S$ must be read as the family of suprema). Similar notations and definitions hold for the dual concepts. The element $x \wedge y$ is called {\em meet}. On a poset the infimum of $X$ when it exists, i.e.\ the element which is a lower bound for the subset $X$, is called {\em least, smallest or minimum element} and denoted $\bot$ or 0.

\begin{remark} \label{cnkk}
Note that if $S_1\subset S_2$  every supremum $v_2$ for $S_2$ is an upper bound for $S_2$ and so for $S_1$, so if $v_1$ is a supremum for $S_1$, we have $v_1\le v_2$.
\end{remark}

\begin{definition}
Let $(X,\le)$ be a preordered set and let $S\subset X$ be a subset.
We say that $S$ is {\em directed} or {\em upward-directed} if it is non-empty and for every $x,y\in S$ there is $z\in S$ such that $x,y\le z$.

We say that $S$ is {\em filtered} or {\em downward-directed} if it is non-empty and  for every $x,y\in S$ there is $z\in S$ such that $z\le x,y$.
\end{definition}

For any $x\in X$ we write
\[
i(x):=\,\,\uparrow x:=\{y: x\le y\}, \qquad d(x):=\,\,\downarrow x:=\{y: y\le x\}.
\]
and for a subset $S\subset X$
\[
i(S):=\,\,\uparrow S:=\{y:  \ \exists x\in S, \ x\le y\}, \qquad d(S):=\,\,\downarrow S:=\{y: \ \exists x\in S,  y\le x \}.
\]

Observe that $\downarrow^{\textrm{op}} x=\{y: y\le^{\textrm{op}} x\}= \{y: x\le y\}=\,\uparrow x$, and similarly  $\uparrow^{\textrm{op}} x=\{y: x\le^{\textrm{op}}  y\}=\{y: y\le x\}=\, \downarrow x$.

A point $x\in X$ is {\em maximal} if $i(x)=\{x\}$ while it is {\em maximum} if $d(x)=X$.

\begin{definition}
A subset $S$ is called {\em upper} or {\em increasing} if $i(S)\subset S$. A subset $S$ is called {\em lower} or {\em decreasing} if $d(S)\subset S$.
\end{definition}

The complement of an upper set is lower and conversely.

\begin{definition}
A subset $S$ is called {\em ideal}  if it is a directed lower set, and a {\em filter} if it is a filtered upper set.
\end{definition}

\begin{definition}
A poset for which every subset has a supremum and an infimum is a {\em complete lattice}.

A poset in which every directed subset has a supremum is a {\em directed complete poset} or {\em dcpo}.

A poset in which every bounded  subset has a supremum is {\em bounded complete}.

A poset in which every bounded directed subset has a supremum is {\em bounded directed complete} or {\em conditionally directed complete} or {\em cdcpo}.
\end{definition}
Note that dcpos are bounded directed complete. The weaker  properties  will be more important to us.

\begin{remark}
In a dcpo every element $x$ has a maximal element $m$ above it, $x\le m$. Indeed, let $C$ be a maximal chain containing $x$, which exists by Zorn's lemma, then $m:=\sup C$ being an upper bound for every element of $C$ must belong to $C$ (otherwise the chain would not be maximal) and if $z\in X$ is such that  $\sup C\le z$, again by maximality of $C$,  it is in $C$, thus $z\le \sup C$ and hence $z=\sup C$, that is, $i(m)=\{m\}$.
\end{remark}

\subsection{Scott topology}

The Scott topology, defined next, encodes the order $\le$ as a
topological structure; it is introduced here because it is needed for
the proof of Proposition~\ref{cvxz}, which in turn prepares the ground
for the characterization of the way-below relation and continuity.

\begin{definition} \label{kooi}
A subset $U$ of a poset $(X,\le)$ is {\em Scott-open} if
\begin{itemize}
\item[(i)] $U$ is an upper set,
\item[(ii)] Inaccessibility by directed suprema: for every directed subset $S$  with a supremum, $\sup S \in U \Rightarrow S\cap U\ne \emptyset$.
\end{itemize}
The collection of all Scott open sets $\mathscr{S}$ on $X$ is called the {\em Scott topology}.
\end{definition}

\begin{proof}[Proof that it is a topology]
The empty set satisfies (i)--(ii) vacuously.
The full space $X$ satisfies them because both conclusions hold
(directed sets are non-empty).

Finite intersections and arbitrary unions of upper sets are upper, so it is sufficient to consider property (ii) under finite intersections and arbitrary unions.


\smallskip\noindent\textit{Finite intersections.}
Let $U=\bigcap_i U_i$ where each $U_i$ satisfies (i)--(ii), and let
$S$ be directed with $\sup S\in U$.
Then $\sup S\in U_i$ for every $i$, so there exists $s_i\in S\cap U_i$.
Since $S$ is directed, there is $s\in S$ with $s_i\le s$ for every $i$;
as each $U_i$ is upper, $s\in U_i$ for every $i$, hence $s\in U\cap S$.


\smallskip\noindent\textit{Arbitrary unions.}
Let $U=\bigcup_\alpha U_\alpha$ where each $U_\alpha$ satisfies
(i)--(ii), and let $S$ be directed with $\sup S\in U$.
Then $\sup S\in U_\alpha$ for some $\alpha$, so there exists
$s\in U_\alpha\cap S\subset U\cap S$.
\qed \end{proof}

The {\em dual Scott topology} $\mathscr{S}^{\textrm{op}}$ is the Scott topology of $\le^{\textrm{op}}$. Its open sets are more easily identified as the lower sets which are inaccessible by directed infima.

The next result establishes that the information in $\le$ is the same as that in the Scott topology (one direction is clear since $\mathscr{S}$  is derived from $\le$).

\begin{proposition} \label{cvxz}
On a poset $(X,\le)$ we have for every $y\in X$, $\operatorname{Cl}_{\mathscr{S}}(y)=\,\downarrow y$; in particular $\downarrow y$ is Scott-closed.
\end{proposition}

\begin{proof}
First, we show $\operatorname{Cl}_{\mathscr{S}}(y) \subset {\downarrow} y$.
It suffices to show that $X \setminus {\downarrow} y$ is Scott-open. To start with, it is the complement of a lower set hence upper.

Now let $D$ be directed with $\sup D \not\le y$, i.e.\ $\sup D\in X \setminus {\downarrow} y$. Since $y$ is not an upper bound of $D$, there is $d
\in D$, $d\not\leq y$
proving inaccessibility by directed suprema of $X \setminus {\downarrow} y$.
Hence $X \setminus {\downarrow} y$ is Scott-open, so ${\downarrow} y$ is Scott-closed. Since $\{y\} \subset {\downarrow} y$, we get $\operatorname{Cl}_{\mathscr{S}}(y) \subset {\downarrow} y$.

Conversely, we prove ${\downarrow} y \subset \operatorname{Cl}_{\mathscr{S}}(y)$.
Take \(x \le y\). Let \(U\) be any Scott-open set with \(x \in U\).
Because \(U\) is an upper set and \(x \le y\), we have \(y \in U\).
Thus every Scott-open neighbourhood of \(x\) contains \(y\), so \(x \in \operatorname{Cl}_{\mathscr{S}}(y)\).
Therefore, ${\downarrow} y \subset\operatorname{Cl}_{\mathscr{S}}(y)$.
\qed \end{proof}

For the notion of specialization order see \cite[Def.\ O-5.2.]{gierz03}.

\begin{corollary}
On a poset $(X,\le)$ we have that $\le$ coincides with the specialization order of the Scott topology: $x\le y$ iff  $x\in \operatorname{Cl}_{\mathscr{S}}(y)$. The Scott topology is $T_0$.
\end{corollary}

\begin{proof}
The equivalence is clear because   $x\in \operatorname{Cl}_{\mathscr{S}}(y)$ reads $x\in \downarrow y$, that is, $x \le y$. Suppose that $x$ and $y$ have the same Scott open neighborhoods,  then $x\in \operatorname{Cl}_{\mathscr{S}}(y)$ and $y\in \operatorname{Cl}_{\mathscr{S}}(x)$ which gives $x\le y \le x$ and hence $x=y$ by the poset assumption.
 \qed \end{proof}

A different characterization of the Scott topology is \cite{lawson91}

\begin{proposition}
A set is Scott-closed iff it is a lower set which is closed with respect to taking suprema of directed subsets (i.e.\ if the supremum exists it belongs to the set).
\end{proposition}

It is a characterization because, if closed sets are defined via this property, we are actually defining a family of open sets (Scott's) which do constitute a topology by the proof after Def.\ \ref{kooi}.

\begin{proof}
Let $A$ be Scott-closed (hence lower) and let $U=X\backslash A$. Let $D\subset A$ be a directed subset. Suppose $\sup D\notin A$, then $\sup D\in U$ and since $U$ is Scott-open, there is $d\in D\cap U$, which contradicts $D\subset A$. The contradiction proves that $\sup D\in A$, so $A$ is closed  with respect to taking suprema of directed subsets.

Conversely, if $A$ is lower and closed with respect to taking suprema of directed subsets, let $U=X\backslash A$ (so upper) and let $D$ be a directed subset such that $\sup D \in U$. It cannot be $D\subset A$ otherwise $\sup D\in A$, thus there is some $d\in U\cap D$.
\qed \end{proof}

The following is an elementary fact observed e.g.\ in \cite[Lemma 7.1]{keimel09}
\begin{proposition}
If $D$ is directed and $\sup D$ exists, then $D$, viewed as a net, converges to $\sup D$ in $\mathscr{S}$.
\end{proposition}

\begin{proof}
Index the net by $D$ itself, ordered as in $X$, and put $x_d=d$; since $D$ is
directed, $(x_d)_{d\in D}$ is a net. Let $U\in\mathscr{S}(X)$ with $\sup D\in U$. As $U$ is inaccessible by existing directed suprema, there is $d_0\in D\cap U$. If $d\in D$ and $d\ge d_0$, then $d\in U$ because $U$ is an upper set. Hence the net is eventually in every Scott-open neighbourhood of $\sup D$, i.e.\ it converges to $\sup D$ in $\mathscr{S}(X)$. \qed
\end{proof}

\subsection{The way-below relation}

\begin{definition}[Way-below relation \cite{gierz03}(Def.\ I-1.1)]
On a preordered set $(X,\le)$ we say that $x$ is {\em way-below} $y$, in symbols, $x\ll y$, if for every directed subset $D$ with a supremum, $y\le \sup D$ implies that there is $d\in D$ such that $x\le d$.
\end{definition}

The way-below relation has several interpretations in computer science (finite approximation) and topology (relative compactness) as discussed in \cite{gierz03}.

\begin{definition}[Way-above relation]
On a preordered set $(X,\le)$ we say that $y$ is {\em way-above} $x$, in symbols $x \,\,\overline{\ll}\,\, y$, if for every filtered subset $F$ with an infimum, $\inf F\le x$ implies that there is $f\in F$ such that $f\le y$.
\end{definition}
In the literature $\overline{\ll}$ is often denoted $\ll_d$. We employ the overline notation as a reminder that this is the `above' version.

Let $\varphi$ be the map which sends the partial order into the way-below relation, then $(\varphi(\le^{\textrm{op}}))^{\textrm{op}}$ is the way-above relation. We shall develop most of the theory for the way-below relation, the version for the way-above relation being obtained by duality.

The following result is important \cite[Prop.\ I-1.2]{gierz03}, for it is remarkable that researchers in computer science came up with a structure very close to Kronheimer and Penrose's causal spaces \cite{kronheimer67} (properties (i) and (ii)) at around the same time.
\begin{theorem} \label{kctt}
In a poset $(X,\le)$ the following statements hold for all $u, x,
y, z,$ $v \in X$:
\begin{itemize}
\item[(i)] $x\ll y$ implies $x\le y$,
\item[(ii)]  $u\le x\ll y\le v$ implies $u\ll v$,
\item[(iii)] $x \ll z$ and $y \ll z$ imply $x \vee y \ll z$ whenever $x \vee y$ exists,
\item[(iv)] $0\ll x$ whenever $X$ has a smallest element $0$.
\end{itemize}
\end{theorem}
Identifying the graph of the relation with the relation itself, we can write (i) as $\ll \subset \le$ and (ii) as $\ll \circ \le \cup \le \circ \ll \subset \ll$, where $\circ$ denotes composition of relations.
Conditions (i) and (ii) immediately imply that $\ll$ is transitive, while (i) alone implies that $\ll$ is antisymmetric. However, $\ll$ is not necessarily irreflexive; that is, there may exist $x\in X$ such that $x\ll x$. Points with this property are said to be {\em compact} (this terminology has nothing to do with a topology on $X$; rather, it comes from a topological model of a poset that inspired the definition of $\ll$ \cite{gierz03}).
\begin{proof}
(i). Take $D=\{y\}$ then $y$ is an upper bound of $D$ and if $u$ is an upper bound of $D$, $y\le u$ thus $y=\sup D$. Now, $x\ll y$ implies that there is $d\in D$ such that $x\le d$, that is $x\le y$.

(ii). Since $x\le y$ we have $u \le v$.
Let $D$ be a directed set such that $\sup D$ exists  and satisfies $v\le \sup D $. Then, as $y\le v$, $y\le \sup D$, and since $x\ll y$, there is $d\in D$ such that $x\le d$ so $u\le d$. In conclusion, for every   directed set such that $\sup D$ exists  and satisfies $v\le \sup D$ there is $d\in D$ such that $u\le d$, that is, $u\ll v$.

(iii). Let $D$  be a directed set such that $\sup D$ exists  and satisfies $z\le \sup D$. Then there are $d_1,d_2\in D$ such that $x\le d_1$ and $y\le d_2$. Since $D$ is directed, there is $d\in D$, $x, y \le d$, i.e.\ $d$ is an upper bound for $\{x,y\}$ so $x\vee y \le d$, hence $x\vee y \ll z$.

(iv). Suppose that a smallest element 0 exists, this means $0\le u$ for every $u$. Let $D$   be a directed set such that $\sup D$ exists  and satisfies $x\le \sup D$, then there is $d\in D$ (in fact any element would do), such that $0\le d$, thus $0\ll x$.
\qed \end{proof}

In analogy with $\le$ we write
\[
\twoheaduparrow x:=\{y: x\ll y\}, \qquad \twoheaddownarrow x:=\{y: y\ll x\}.
\]
One must be cautious because both sets are characterized via directed subsets, not filtered ones, at least
without an additional condition (bicontinuity) to be considered below.

For the way-above relation we have analogously
\begin{itemize}
\item[(i')] $x\,\, \overline{\ll}\, \,y$ implies $x\le y$,
\item[(ii')]  $u\le x\, \,\overline{\ll} \,\, y\le v$ implies $u\, \, \overline{\ll} \, \,v$,
\item[(iii')] $z \,\, \overline{\ll} \,\, x$ and $z\, \,\overline{\ll} \,\,y$ imply $z\, \,\overline{\ll}\,\, x \wedge y $ whenever $x \wedge y$ exists,
\item[(iv')] $x\,\, \overline{\ll} \,\, 1$ whenever $X$ has a largest element $1$.
\end{itemize}
and
 we write
\[
\bar{\twoheaduparrow}\, x:=\{y: x\,\, \overline{\ll} \,\, y\}, \qquad \bar \twoheaddownarrow \, x:=\{y: y\,\, \overline{\ll}\,\, x\}.
\]

\begin{proposition} \label{cnqz}
In a poset $(X,\le)$, we have $x\le y \Rightarrow \ \twoheaddownarrow x \subset\, \twoheaddownarrow y$.
\end{proposition}


\begin{proof}
 Let $z\ll x$ and $D$ be a directed set such that $\sup D$ exists and $y\le \sup D$, then, by $x\le y$, $x\le \sup D$ so, by $z\ll x$,  there is $d\in D$ such that $z\le d$, thus $z \ll y$.
\qed \end{proof}

The following result is \cite[Prop.\ I-1.5(i)]{gierz03}

\begin{proposition} \label{mmtg}
In a poset $(X,\le)$ the condition $x \ll y$ is equivalent to: $x \in I$ for every ideal $I$ such that $y \le \sup I$.
\end{proposition}

\begin{proof}
Suppose $x\ll y$ and let $I$ be an ideal (hence a directed set) such that $y \le \sup I$. By $x\ll y$, we have $x\le d$ for some $d\in I$, so $x\in I$ as $I$ is lower.

Suppose $x \in I$ for every ideal $I$ such that $y \le \sup I$. Let  $D$ be a directed subset with $y \le \sup D$. Then $I = \,\downarrow D$ is an ideal, and $y \le \sup D = \sup I$ (the families of upper bounds for $D$ and $I$ coincide).  Then $x \in I$ by  the assumption, i.e.,
there is a $d \in D$ such that $x \le d$. Hence $x \ll y$.
 \qed \end{proof}

The following result is \cite[Prop.\ I-1.5(ii)]{gierz03}

\begin{proposition} \label{wwgf}
In a poset $(X,\le)$ for $x\in X$, suppose there is a directed set $D \subset \, \twoheaddownarrow x$ such that $\sup D=x$. Then the same is true for $D=\,\,\twoheaddownarrow x$: the set  $\twoheaddownarrow x$ is directed (hence non-empty) and $\sup\twoheaddownarrow x=x$.
\end{proposition}

\begin{proof}
If $y,z\in \, \twoheaddownarrow x$, namely $y\ll x$, $z\ll x$, then there are $d_y,d_z\in D$ such that $y\le d_y$, $z\le d_z$, and since $D$ is directed, there is $d\in D$  such that $y,z\le d\ll x$ (recall $D\subset\,  \twoheaddownarrow x$). Thus $\twoheaddownarrow x$ is directed.

Clearly, $x$ is an upper bound for any element of $\twoheaddownarrow x$. Let $v$ be an upper bound for $\twoheaddownarrow x$ then it is an upper bound for $D$, so $\sup D\le v$, that is  $x\le v$, so $x$ is the supremum for $\twoheaddownarrow x$.
\qed \end{proof}


\subsection{Continuous posets and future reflectivity}

The following notion should be distinguished from that of basis of a poset topology (e.g.\ Scott's or Lawson's)

\begin{definition}
A basis $B$ for a poset $(X,\le)$ is a subset $B\subset X$ such that for each $x \in X$, the subset $B\,\cap \twoheaddownarrow x$ is directed and $\sup (B\,\cap \twoheaddownarrow x)=x$.
\end{definition}

\begin{definition}
A poset $(X,\le)$ is {\em continuous} if it satisfies the {\em axiom of approximation from below}:
for every $x\in X$, $\twoheaddownarrow x$ is directed and $\sup \twoheaddownarrow x=x$.
\end{definition}
In other words, a poset $(X,\le)$ is continuous iff $X$ is a  basis.

\begin{proposition}
A poset is continuous iff it admits a basis.
\end{proposition}

\begin{proof}
If the poset is continuous $B=X$ is a basis. Conversely, if $B$ is a basis then for each $x\in X$, $D_x:=B\,\cap \twoheaddownarrow x$ is directed and such that $\sup D_x=x$. By Prop.\ \ref{wwgf} $\twoheaddownarrow x$ is directed and $\sup \twoheaddownarrow x=x$.
\qed \end{proof}

A poset is {\em $\omega$-continuous} if it admits a countable basis $B$.

\begin{definition}
A poset $(X,\le)$ is {\em dual continuous} or   {\em co-continuous}  if  $(X,\le^{\textrm{op}})$ is continuous. Equivalently,
it satisfies the {\em axiom of approximation from above}:
for every $x\in X$, $\bar{\twoheaduparrow}\,  x$ is filtered, has infimum and $\inf \bar{\twoheaduparrow} \, x=x$.
\end{definition}

\begin{definition}
A poset is {\em weakly bicontinuous} if it is both continuous and co-continuous. It is  {\em weakly $\omega$-bicontinuous} if it admits a countable basis and a countable co-basis.
\end{definition}
Note that \cite{keimel09b,ebrahimi14,ebrahimi15} would call these spaces just {\em bicontinuous}. Following Martin and Panangaden we keep the term {\em bicontinuous} for a stronger notion (which additionally demands $\ll\,=\, \overline{\ll}$) which is more central in relativity theory.

A continuous dcpo is a {\em domain}.

In a continuous poset $\twoheaddownarrow x$ is an ideal for every $x\in X$, called {\em primitive}.

We have a converse of Prop.\ \ref{cnqz}.
\begin{proposition} \label{cnzz}
In a continuous poset $(X,\le)$, we have $x\le y \Leftrightarrow \ \twoheaddownarrow x \subset \,\twoheaddownarrow y$.
\end{proposition}

\begin{proof}
We need only to prove  the $\Leftarrow$ direction. By Remark \ref{cnkk} as $\twoheaddownarrow x \subset \, \twoheaddownarrow y$  and the suprema exist, $x=\sup  \twoheaddownarrow x \le \sup \twoheaddownarrow y=y$.
\qed \end{proof}

\begin{remark}
The way-below relation is defined from $\le$. The previous result establishes that under continuity, the relation $\le$ can be recovered from $\ll$ and so they contain the same information.
\end{remark}

The following simple result  will be important for us.

\begin{proposition} \label{jsse}
Every continuous poset $(X,\le)$ is {\em future reflecting} in the sense that: for every $x,y\in X$,  $\twoheaddownarrow x \subset \,\twoheaddownarrow y$ $\Rightarrow\,\,  \twoheaduparrow x \supset\,\, \twoheaduparrow y$.
\end{proposition}

\begin{proof}
From Prop.\ \ref{cnzz}  $\twoheaddownarrow x \subset \, \twoheaddownarrow y$ implies $x\le y$. Let $z\in  \,\twoheaduparrow y$, namely $y\ll z$. From $x\le y \ll z$ we get $x\ll z$, that is, $z\in\,  \twoheaduparrow x$, hence, by the arbitrariness of $z$, $\twoheaduparrow x \supset\,\, \twoheaduparrow y$.
\qed \end{proof}

The following non-trivial result \cite[Thm.\ I-1.9]{gierz03}  provides an  equivalent strengthened form of the way-below relation

\begin{theorem} \label{tncv}
Let $(X,\le)$ be a continuous poset. If $x \ll y$ and $y \le \sup D$ for a directed set $D$, then there is $d\in D$ such that $x\ll d$. Additionally, if $x\ne y$, $d$ can be found so that $x\ne d$.
\end{theorem}

Note that the definition of way-below relation  ensures $x\le d$.
\begin{proof}
Let $I = \bigcup \{\twoheaddownarrow d: d \in D\}$. Observe that $I$ is lower and it is directed because if $i_1,i_2\in I$ there are $d_1,d_2\in D$ such that $i_1\ll d_1$, $i_2\ll d_2$. As $D$ is directed there is $d\in D$ such that $d_1,d_2\le d$, so $i_1,i_2\ll d$, i.e.\ $i_1,i_2\in \twoheaddownarrow d$. But, by continuity, $\twoheaddownarrow d$ is directed thus there is $i\in  \twoheaddownarrow d\subset I$ such that $i_1,i_2\le i$. In conclusion, $I$ is an ideal.

Note that every upper bound for $D$ is also an upper bound  for $I$. Conversely, every upper bound $v$ for $I$ is an upper bound for $\twoheaddownarrow d$, for any $d\in D$, so $d=\sup  \twoheaddownarrow d \le v$, hence every upper bound for $I$ is an upper bound for $D$, so $\sup D=\sup I$, in particular $\sup I$ exists.
Now, $x\ll y$ and $y \le \sup D=\sup I$ implies $x \in I$ by Prop.\ \ref{mmtg}, which means that $x \ll d_1$ for some $d_1 \in D$. This concludes the proof of the first statement.

If $x\ne y$ we continue the argument as follows.
Since $\le$ is antisymmetric, $y\nleq x$, now $x$ cannot be an upper bound for $D$, otherwise $y\le \sup D\le x$, a contradiction. Thus there is $d_2\in D$ such that $d_2 \nleq x$. Since $D$ is directed there is $d\in D$ such that $d_1,d_2\le d$, so $x\ll d_1\le d$ which implies $x\ll d$. Moreover,  $d_2\le d$ and  $d_2 \nleq x$ imply that $x\ne d$.
\qed \end{proof}

The following is an even stronger form
\begin{theorem} \label{tncw}
Let $(X,\le)$ be a continuous poset. Let $B$ be a finite subset. If $B \ll y$ (meaning $b_i\ll y$ for every $b_i\in B$) and $y \le \sup D$ for a directed set $D$, then there is $d\in D$ such that $B\ll d$. Additionally, if $y \notin B$, $d$ can be found so that $d\notin B$.
\end{theorem}

\begin{proof}
Let $I = \bigcup \{\twoheaddownarrow d: d \in D\}$. The proof that it is an ideal, that $\sup I$ exists and $\sup D=\sup I$ is as in Thm.\ \ref{tncv}.

Now, $b_i\ll y$ and $y \le \sup D=\sup I$ imply $b_i \in I$ by Prop.\ \ref{mmtg}, which means that $b_i \ll d_i$ for some $d_i \in D$. As $D$ is directed, there is $d\in D$, such that $d_i\le d$ for every $i$, so $b_i\ll d$ for every $i$. This concludes the proof of the first statement.

If $b_i \ne y$ for every $i$, we continue the argument as follows.
Since $\le$ is antisymmetric, $y \nleq b_i$ for every $i$. So $b_i$ cannot be an upper bound for $D$, otherwise $y\le \sup D\le b_i$, a contradiction. Thus there is $\tilde d_i\in D$ such that $\tilde d_i \nleq b_i$. Since $D$ is directed there is $d\in D$ such that for every $i$, $d_i,\tilde d_i\le d$, so $b_i\ll d_i\le  d$ which implies $B\ll d$. Moreover,  $\tilde d_i\le d$ and  $\tilde d_i \nleq b_i$ imply that $b_i\ne d$ for every $i$.
\qed \end{proof}

An important  consequence of Thm.\ \ref{tncv} is \cite[Thm.\ I-1.9, I-1.19]{gierz03},

\begin{theorem} \label{oovg}
Let $(X,\le)$ be a continuous poset. Then $\ll$ is interpolating: for every pair $x,z\in X$ such that $x\ll z$ there is $y\in X$ such that $x\ll y \ll z$. Moreover, (strong interpolation) if $x\ne z$ then $y$ can be found so that  $y\ne x$.
\end{theorem}

\begin{proof}
Let $D= \twoheaddownarrow z$, then $x\ll z\le z= \sup D$ thus by  Theorem \ref{tncv} there is $d\in D$ such that $x\ll d$. Moreover, if $x\ne z$ the element $d$ can be found so that $d\ne x$. Set $y:=d$.
\qed \end{proof}

An important  consequence of Thm.\ \ref{tncw} is

\begin{theorem} \label{nxre}
Let $(X,\le)$ be a continuous poset and let $B$ be a finite subset.
If $B\ll z$ then there is $y\in X$ such that $B\ll y \ll z$. Moreover, (strong interpolation) if $z\notin B$ then $y$ can be found so that  $y\notin B$.
\end{theorem}

\begin{proof}
Let $D= \,\twoheaddownarrow z$, then $B\ll z\le z= \sup D$ thus by  Theorem \ref{tncw} there is $d\in D$ such that $B\ll d$. Moreover, if $z\notin B$ the element $d$ can be found so that $d\notin B$. Set $y:=d$.
\qed \end{proof}

\begin{remark}
Observe that we do not claim that in a continuous poset $\ll$ is irreflexive.
If $0\in X$ exists, by Thm.\ \ref{kctt}(iv), $0\ll x$ for every $x\in X$, in particular $0\ll 0$.
\end{remark}

We are ready to prove a convenient characterization of the Scott topology. It is well known though usually formulated for dcpos \cite[Prop.\ II-1.10]{gierz03}

\begin{theorem} \label{nncf}
In a continuous poset each set of the form $\twoheaduparrow x$ is Scott-open. The sets of the form $\twoheaduparrow x$ provide a basis for the Scott topology. In particular, for every Scott-open set $U$ we have $U=\cup_{u\in U}  \twoheaduparrow u$.
\end{theorem}

Analogously, the dual Scott topology $\mathscr{S}^{\textrm{op}}$ is generated by sets of the form $\bar \twoheaddownarrow\, v$.

\begin{proof}
Each set $\twoheaduparrow x$ is upper and if $D$ is a directed set with a supremum $v\in \,\twoheaduparrow x$ this means $x\ll v$, and so by Theorem \ref{tncv} there is $d\in D$ such that $x \ll d$ which means $D\, \cap\twoheaduparrow x  \ne \emptyset$, that is  $\twoheaduparrow x$ is Scott open.

Let $U$ belong to the Scott topology, and let $u\in U$ then as $u=\sup \twoheaddownarrow u$, by the definition of Scott open set, there is an element $u_2\in U\cap \twoheaddownarrow u$, that is, $u_2\ll u$, so $u\in  \, \twoheaduparrow u_2$. Note that as $U$ is upper, $\twoheaduparrow u_2\subset U$. So we proved $U=\cup_{u\in U}  \twoheaduparrow u$, thus every Scott open set is indeed a union of sets of the form $\twoheaduparrow x$.
\qed \end{proof}

\begin{theorem}[Scott interior of a principal upper set]\label{thm:scott}
Let $X$ be a continuous poset. Then for every $x\in X$,
\begin{equation}
  \twoheaduparrow x=\operatorname{Int}_{\mathscr{S}}(\uparrow x).
\end{equation}
\end{theorem}

\begin{proof}
\emph{($\subset$).}
We already know that $\twoheaduparrow x$ is Scott open.
Since  $\twoheaduparrow x\subset\uparrow x$  every
$y\in \twoheaduparrow x$ lies in a Scott open subset of $\uparrow x$, hence in $\operatorname{Int}_{\mathscr{S}}(\uparrow x)$.

\emph{($\supset$).} Let $y\in\operatorname{Int}_{\mathscr{S}}(\uparrow x)$ and pick $U$ Scott open with
$y\in U\subset\uparrow x$. Let $S$ be directed with a supremum and $y\le\sup S$. As $U$
is an upper set, $\sup S\in U$; as $U$ is Scott open, there is $s\in S\cap U$; as
$U\subset\uparrow x$, we get $x\le s$. This is exactly $x\ll y$.
\qed \end{proof}

\begin{remark}
The inclusion $\twoheaduparrow x \supset \operatorname{Int}_{\mathscr{S}}(\uparrow x)$ used no hypothesis on $X$ and holds in
every poset. Continuity is needed only for the converse.
\end{remark}

\begin{definition}
In a poset $(X,\le)$ the {\em biScott topology} or {\em interval topology}\footnote{Actually in domain theory the {\em interval topology} is a different concept. It is the join of the lower topology, whose subbasis of closed sets is $\{\uparrow a, \ a\in X\}$, and the upper topology, whose subbasis of closed sets is $\{\downarrow b, \ b\in X\}$. The term {\em interval topology} is more suggestive of some equivalences that we shall explore with the Alexandrov topology in relativity.} $\mathscr{I}$ is the join of the Scott and dual Scott topologies: $\mathscr{I}=bi \mathscr{S}=\mathscr{S}\vee \mathscr{S}^{\textrm{op}}$.
\end{definition}
It is a general fact that the join of two topologies has as a basis the intersections of the basis elements of the two joined topologies, thus  a consequence of Thm.\ \ref{nncf} is

\begin{proposition} \label{mmft}
In a weakly bicontinuous poset a basis for the interval topology is provided by the sets of the form
\[
(a,b):=\{x\in X: \ a\ll x \,\, \overline{\ll}\, \,b\}.
\]
\end{proposition}
Alternatively, a direct proof is as follows

\begin{proof}
Every $x$ lies in such a set: continuity gives $a$ with $a\ll x$ and dual continuity gives $b$ with
$x\, \overline{\ll}\, b$, so $x\in (a,b)$. Suppose $x\in (a,b)\cap (c,d)$. Since $\twoheaddownarrow x$ is
directed and $a,c\in \,\twoheaddownarrow  x$, there is $e\in \, \twoheaddownarrow  x$ with $a,c\le e\ll x$. Since $\bar \twoheaduparrow \,x$ is
filtered and $b,d\in \bar \twoheaduparrow  \, x$, there is $f\in\, \bar \twoheaduparrow \,x$ with $x\, \overline{\ll}\, f\le b,d$. Then
$x\in (e,f)$, and for $z\in (e,f)$ we have $a\le e\ll z$ and $z\, \overline{\ll} \, f\le b$, so  $z\in(a,b)$; similarly $z\in (c,d)$. Hence
$(e,f)\subset (a,b)\cap (c,d)$.
\qed \end{proof}

The very definition of biScott topology clarifies that it is self-dual. For a weakly bicontinuous poset this can also be seen as follows. The basis open set reads
\[
a \ \varphi(\le) \ x \ (\varphi(\le^{\textrm{op}}))^{\textrm{op}} \ b ,
\]
or equivalently
\[
b  \ \varphi(\le^{\textrm{op}}) \ x \  \varphi(\le)^{\textrm{op}}  \ a \quad \Leftrightarrow \quad  b  \ \varphi(\le^{\textrm{op}}) \ x \ \varphi((\le^{\textrm{op}})^{\textrm{op}})^{\textrm{op}} \ a ,
\]
thus renaming $a'=b$, $b'=a$, $\le'=\le^{\textrm{op}}$ we get precisely the same basis open sets for $\le^{\textrm{op}}$ in place of $\le$.
\begin{remark} \label{mvvr}
Note that in a weakly bicontinuous poset for each $x\in X$, $\twoheaddownarrow x$ and $\bar \twoheaduparrow\, x$ are non-empty (since they are directed and filtered subsets respectively), thus each point $x$ belongs to an interval $(a,b)$ for some $a,b\in X$.
\end{remark}

\begin{proposition} \label{mmft2}
In a weakly $\omega$-bicontinuous poset the interval topology is second countable.
\end{proposition}

\begin{proof}
Let $B$ be a basis and $C$ a co-basis. Let $x\in (y,z)$, since $\sup B\,\cap \twoheaddownarrow x= x$, we can find  $b\in B$ such that $y \le b \ll x$, and dually $c\in C$ such that $x\,\,\overline{\ll} \,\,c \le z$. Then $x\in (b,c)\subset (y,z)$ where the sets of the form $(b,c)$, with $b\in B$, $c\in C$ form a countable basis.
\qed \end{proof}

\begin{theorem}[Interval interior of a principal upper set]\label{thm:pointed}
Let $X$ be weakly bicontinuous. Then for every $x\in X$,
\begin{equation}
  \operatorname{Int}_{\mathscr{I}}(\uparrow x)=\, \twoheaduparrow x .
\end{equation}
\end{theorem}

Applying this theorem to  $X^{\mathrm{op}}$, whose interval topology coincides with
that of $X$, yields $\operatorname{Int}_{\mathscr{I}}(\downarrow x)=\bar \twoheaddownarrow \, x$.

\begin{proof}
\emph{($\supset$).} Let $x\ll y$. By dual continuity choose $b$ with $y\,\,\overline{\ll}\, \,b$. Then
$y\in(x,b)$, and every $z\in(x,b)$ satisfies $x\ll z$, so
$(x,b)\subset\,\twoheaduparrow x\subset\, \uparrow x$. Hence
$y\in\operatorname{Int}_{\mathscr{I}}(\uparrow x)$.

\emph{($\subset$).} Let $y\in\operatorname{Int}_{\mathscr{I}}(\uparrow x)$. By Prop.\ \ref{mmft} there is a basic
set $(a,b)$ with $y\in(a,b)\subset\,\uparrow x$; thus $a\ll y\,\,\overline{\ll} \, \, b$. By
interpolation choose $a'$ with $a\ll a'\ll y$. Then $a'\le y\,\, \overline{\ll} \,\, b$ gives $a'\,\overline{\ll} \, b$, so $a'\in (a,b) \subset\,\uparrow x$, i.e.\ $x\le a'$. From
$x\le a'\ll y$ we conclude $x\ll y$.
\qed \end{proof}

A subset of a topological space is {\em regularly open} if it is the interior of its closure.

\begin{proposition}\label{prop:regopen}
Let $X$ be weakly bicontinuous. Then for every $x\in X$,
\[
  \operatorname{Int}_{\mathscr{I}}\big(\operatorname{Cl}_{\mathscr{I}}(\twoheaduparrow x)\big)=\, \twoheaduparrow x .
\]
That is, $\twoheaduparrow x$ is regularly open.
\end{proposition}

\begin{proof}
$\twoheaduparrow x$ is open so $\twoheaduparrow x\subset\operatorname{Int}(\operatorname{Cl}(\twoheaduparrow x))$. Conversely
$\uparrow x$ is closed and contains $\twoheaduparrow x$, so
$\operatorname{Cl}(\twoheaduparrow x)\subset\,\uparrow x$ and therefore $\operatorname{Int}(\operatorname{Cl}(\twoheaduparrow x))\subset\operatorname{Int}(\uparrow x)=\,\twoheaduparrow x$, by
Theorem~\ref{thm:pointed}.
\qed \end{proof}

\begin{proposition} \label{clddw}
In a continuous poset $\operatorname{Cl}_{\mathscr{S}}(\twoheaddownarrow x)=\,\downarrow x$  and  $\operatorname{Cl}_{\mathscr{I}}(\twoheaddownarrow x)=\,\downarrow x$ for every $x$, and dually in a co-continuous poset
$\operatorname{Cl}_{\mathscr{S}^{\textrm{op}}}(\bar\twoheaduparrow \, x)=\,\uparrow x$ and $\operatorname{Cl}_{\mathscr{I}}(\bar\twoheaduparrow \, x)=\,\uparrow x$.
\end{proposition}

\begin{proof}
The proof is the same for the two topologies $\mathscr{S}$ and $\mathscr{I}$.
If $y\le x$ then $\twoheaddownarrow y\subset \, \twoheaddownarrow x$, and $y=\sup \twoheaddownarrow y$ is the limit of the directed set $\twoheaddownarrow y$, whence $y\in \operatorname{Cl}(\twoheaddownarrow x)$. Thus $\downarrow x\subset \operatorname{Cl}(\twoheaddownarrow x)\subset \operatorname{Cl}(\downarrow x)=\,\downarrow x$, the last set being closed by Prop.\ \ref{cvxz}. The dual statement follows by applying this to $X^{\textrm{op}}$, the interval topology being self-dual.
\qed \end{proof}

\begin{proposition} \label{mmcl}
Let  $(X,\le)$ be a weakly bicontinuous poset. Let $D$ be a directed subset whose supremum exists. Regarding $D$ as a net,  $(x_d)_{d \in D}$ (where $x_d = d$) it converges to $\sup D$ in the interval topology.
\end{proposition}

\begin{proof}
Let $v=\sup D$, and let $(a,b)\ni v$. This means $a\ll v \, \,\overline{\ll} \, \, b$. Since $v$ is an upper bound, for every $d\in D$, $d\le v$, so $d \, \, \overline{\ll} \,\, b$. By Thm.\ \ref{tncv} and $a\ll v$ there is $d_0\in D$ such that $a\ll d_0$, thus for every $d$ such that $d_0\le d$, $a\ll d \, \, \overline{\ll} \,\, b$, which proves convergence in the interval topology.
\qed \end{proof}

%
%
%
%



\subsection{Scott continuity of maps} \label{ccpd}

Throughout, $S$, $T$ and $X$ denote posets. No completeness assumption is made
unless stated.

\begin{definition}[{Scott continuity \cite[Sec.~5]{keimel09}}]\label{def:scottcont}
A map $f\colon S\to T$ between posets is \emph{Scott-continuous} if it is
order-preserving and preserves all existing directed suprema, i.e.\ for every
directed $D\subset S$ such that $\sup D$ exists, $\sup f(D)$ exists and
$f(\sup D)=\sup f(D)$.
\end{definition}

For dcpos this is exactly \cite[Def.~II-2.2]{gierz03}.

It is called {\em monotone convergence property} in \cite[Def.\ 2.2]{gigli25} \cite[Def.\ 2.4]{gigli26}

\begin{theorem}[{cf.\ \cite[Prop.~II-2.1]{gierz03}, \cite[Lemma~5.3]{keimel09}}]\label{thm:posets}
For a map $f\colon S\to T$ between arbitrary posets the following are
equivalent:
\begin{enumerate}
  \item[(1)] $f$ is continuous for the Scott topologies $\mathscr{S}(S)$ and $\mathscr{S}(T)$;
  \item[(2)] $f$ is Scott-continuous in the sense of Definition~\ref{def:scottcont}.
\end{enumerate}
\end{theorem}

Keimel and Lawson note that the proof in \cite{gierz03} is given for dcpos
but works in general \cite[proof of Lemma~5.3]{keimel09}. One step needs a small
modification, because in a general poset one cannot presuppose the existence
of $\sup f(D)$; it has to be derived. We give the full argument.

\begin{proof}

(1)$\Rightarrow$(2). \emph{Monotonicity.} Suppose $x\le y$ but
$f(x)\not\le f(y)$. By Prop.~\ref{cvxz}, $V=T\setminus\downarrow f(y)$ is
Scott-open and contains $f(x)$, so $f^{-1}(V)$ is a Scott-open, hence upper,
set containing $x$. Thus $y\in f^{-1}(V)$, i.e.\ $f(y)\not\le f(y)$, a
contradiction.

\emph{Existing directed suprema.} Let $D\subset S$ be directed with
$x=\sup D$. By monotonicity, $f(x)$ is an upper bound of $f(D)$. Let $t$ be
\emph{any} upper bound of $f(D)$ and suppose $f(x)\not\le t$. Then
$f^{-1}(T\setminus\downarrow t)$ is a Scott-open set containing $\sup D$, so it
contains some $d\in D$; that is, $f(d)\not\le t$, contradicting the choice of
$t$. Hence $f(x)\le t$ for every upper bound $t$ of $f(D)$, which says
precisely that $\sup f(D)$ exists and equals $f(\sup D)$.

(2)$\Rightarrow$(1). Let $A\subset T$ be Scott-closed. Since $f$ is
monotone, $f^{-1}(A)$ is a lower set. Let $D\subset f^{-1}(A)$ be directed
with $\sup D$ existing. Then $f(D)\subset A$ is directed and, by (2),
$\sup f(D)$ exists and equals $f(\sup D)$. As $A$ is Scott-closed,
$f(\sup D)\in A$, i.e.\ $\sup D\in f^{-1}(A)$. Hence $f^{-1}(A)$ is
Scott-closed. \qed
\end{proof}


Dually, $f$ is continuous for $\mathscr{S}^{\textrm{op}}(S)$ and $\mathscr{S}^{\textrm{op}}(T)$ iff $f$ is order-preserving and preserves all existing filtered infima.

\begin{remark}[The net condition]\label{rem:nets}
By Thm.~\ref{nncf} for a continuous poset a net $(x_j)_{j\in J}$ converges to $x$ in $\mathscr{S}(X)$ iff for every $y\ll x$ we have $y\ll x_j$ eventually, so Scott continuity can be expressed as follows
\begin{enumerate}
  \item[(3)] for every net $(x_j)$ in $S$ Scott-converging to $x$, the net
        $(f(x_j))$ Scott-converges to $f(x)$,
\end{enumerate}
which is equivalent to (1) by general topology.
\end{remark}

For continuous posets we have also the following additional characterization
\begin{theorem}\label{thm:cont}
Let $S$ and $T$ be continuous posets and $f\colon S\to T$ a map. The
conditions (1), (2) and (3)  are equivalent to each of the following:
\begin{enumerate}
  \item[(4)] for all $x\in S$ and $y\in T$: $y\ll f(x)$ iff $y\ll f(w)$ for some $w\ll x$;
  \item[(5)] for all $x\in S$: $\sup f(\twoheaddownarrow x)$ exists and
        $f(x)=\sup f(\twoheaddownarrow x)$.
\end{enumerate}
\end{theorem}

\begin{proof}

The equivalence of (1), (2) and (3) was shown above. We follow the scheme
of \cite{gierz03}: (2)$\Rightarrow$(5)$\Rightarrow$(4)$\Rightarrow$(1).

(2)$\Rightarrow$(5). Since $S$ is continuous, $\twoheaddownarrow x$ is directed with supremum $x$; apply (2).

(5)$\Rightarrow$(4). First, $f$ is monotone: if $x\le y$ then
$\twoheaddownarrow x\subset\,\twoheaddownarrow y$, so $f(x)=\sup f(\twoheaddownarrow x)\le\sup f(\twoheaddownarrow y)=f(y)$.
Hence $f(\twoheaddownarrow x)$ is directed with supremum $f(x)$. If $y\ll f(x)$, Thm.\ \ref{tncv}, applied in $T$, gives $w\ll x$ with
$y\ll f(w)$. Conversely, if $y\ll f(w)$ with $w\ll x$, then $w\le x$, so
$f(w)\le f(x)$ and $y\ll f(x)$.
%

(4)$\Rightarrow$(1). We first show that (4) implies that $f$ is monotone. Let
$a\le b$ in $S$ and let $v\ll f(a)$. By the direct implication in (4) at the
point $a$, there is $w'\ll a$ with $v\ll f(w')$. Since $w'\ll a\le b$, we have $w'\ll b$, so the converse implication in (4) at the point $b$ gives
$v\ll f(b)$. Hence $\twoheaddownarrow f(a)\subset\,\twoheaddownarrow f(b)$, and continuity of $T$ gives $f(a)=\sup\twoheaddownarrow f(a)\le\sup\twoheaddownarrow f(b)=f(b)$.

Now let $U\subset T$ be Scott-open and $x\in f^{-1}(U)$. By
Thm.~\ref{nncf}, applied in $T$, there is $y\in U$ with $y\ll f(x)$, and by
(4) there is $w\ll x$ with $y\ll f(w)$; in particular $y\le f(w)$. We claim
$\twoheaduparrow w\subset f^{-1}(U)$. Indeed, if $w\ll z$, then $w\le z$, and
monotonicity of $f$ gives $y\le f(w)\le f(z)$. Since $y\in U$ and $U$ is an
upper set, $f(z)\in U$. We conclude that $\twoheaduparrow w$ is a Scott-open neighbourhood of $x$ contained in $f^{-1}(U)$, so $f^{-1}(U)$ is
Scott-open. \qed
\end{proof}

\subsection{Lawson topology}

\begin{definition}
On a continuous poset $(X,\le)$ the {\em Lawson topology} $\mathscr{L}$ is  that whose basis is given by sets of the form $\twoheaduparrow x\backslash \uparrow B$ where $B$ is finite.\footnote{In the language of domain theory, it is the join of the Scott and lower topologies.}
\end{definition}

Thus the Lawson topology is finer than the Scott topology (take $B=\emptyset$). Let $B=\{x\}$, since $X$ is Scott open $X=\cup_{z\in X} \twoheaduparrow z$, thus $X\backslash \uparrow B=\cup_{z\in X} (\twoheaduparrow z\backslash  \uparrow B)$ which is Lawson-open. So $\uparrow x$ is Lawson-closed and by Prop.\ \ref{cvxz} $\downarrow x$ is Lawson-closed. So $(X, \mathscr{L}, \le)$ is a semiclosed ordered space \cite{nachbin65}. We shall see in a moment that a stronger result holds.

\begin{proof}[Proof that it is a topology with the mentioned basis]
We need to prove that if a point belongs to a finite intersection of sets of this form then the point belongs to a set of this form included in the finite intersection.

If $u \in \, \twoheaduparrow x_i \backslash \uparrow B_i$ for every $i=1,\ldots, k$, then $u\in \cap_i  \twoheaduparrow x_i$, which is the finite intersection of Scott-open sets hence Scott-open, and by the characterization of the Scott topology (Thm.\ \ref{nncf}), there is $y\in \cap_i  \twoheaduparrow x_i$ such that $u\in \, \twoheaduparrow y\subset \cap_i  \twoheaduparrow x_i$. But $u\notin \,\uparrow B_i$ for every $i$, thus $u\notin \,\cup_i \uparrow B_i=\,\uparrow  B$ with $B=\cup_i B_i$, thus $u\in\, \twoheaduparrow y\backslash \uparrow B\subset \cap_i  (\twoheaduparrow x_i \backslash \uparrow B)\subset  \cap_i  (\twoheaduparrow x_i \backslash \uparrow B_i)$.
\qed \end{proof}

The following is an improvement over \cite[Thm.\ III-1.10.]{gierz03}

\begin{proposition} \label{cjll} Let $(X,\le)$ be a continuous poset.
The graph of $\le$ is closed in the product Lawson topology on $X\times X$, that is,  $(X,\mathscr{L},\le)$ is a closed ordered space. In particular, the Lawson topology is Hausdorff.
\end{proposition}

\begin{proof}
Let $x \nleq y$. There is an element $u \ll x$ with $u \nleq y$ (otherwise, $y$ would be an upper bound for $\twoheaddownarrow x$ and  by continuity $x=\sup \twoheaddownarrow x\le y$, a contradiction). Then $\twoheaduparrow u$ is a Scott (hence, Lawson)-open neighborhood of $x$  and $X\backslash \uparrow u$ is a  Lawson-open neighborhood of $y$. We have shown that $(x,y)\in \, \twoheaduparrow u\times  X\backslash \uparrow u$. Let us prove that this product Lawson-open set does not intersect the graph of $\le$. Suppose $(a,b)\in \,\twoheaduparrow u\times  X\backslash \uparrow u$ then $u\ll a$ and $u\nleq b$. So, we cannot have $a\le b$ otherwise $u\ll a\le b \Rightarrow u\le b$, which contradicts $u\nleq b$. We proved that the complement of the graph of $\le$ is open in the product Lawson topology. Now, since $\le$ is a partial order $\le\cap \le^{\textrm{T}}=\Delta$, the diagonal of $X\times X$, which thus is closed in the product Lawson topology (a condition equivalent to Hausdorffness).
\qed \end{proof}


\begin{proposition} \label{pldf}
In a weakly bicontinuous poset  the Lawson topology is included in the interval topology, thus $\mathscr{S} \subset \mathscr{L}\subset \mathscr{I}$.
\end{proposition}
The Lawson topology for $\le^{\textrm{op}}$ is also included in the interval topology and so the interval topology coincides with the join of the two Lawson topologies and can be called biLawson topology.


\begin{proof}
We need to prove that $\twoheaduparrow x\backslash \uparrow B$, where $B$ is finite, is open in the interval topology. We already know that $\twoheaduparrow x$ is open in the interval topology, thus we have only to prove that $\uparrow b$ is closed as $\uparrow B$ is a finite union of sets of this form. But  $\uparrow b$ is closed for the Scott topology of $\le^{\textrm{op}}$ by Prop.\ \ref{cvxz}, i.e.\ the dual Scott topology, thus it is also closed for the interval topology.
\qed \end{proof}

\begin{corollary} \label{cnpp}
Let $(X,\le)$ be a weakly bicontinuous poset.
The graph of $\le$ is closed in the product interval topology on $X\times X$, that is,  $(X, \mathscr{I},\le)$ is a closed ordered space. In particular, the interval topology is Hausdorff.
\end{corollary}

\begin{proof}
Immediate because the interval topology is finer than the Lawson topology.
\qed \end{proof}

\begin{theorem}\label{thm:lawpointed}
Let $X$ be a continuous poset. Then $\operatorname{Int}_{\mathscr{L}}(\uparrow x)=\, \twoheaduparrow x$ for every $x\in X$.
\end{theorem}

We obtained the same result for the  Scott topology  (Thm.\ \ref{thm:scott}) while for the interval topology it required the stronger weak bicontinuity property, cf.\ Thm.\ \ref{thm:pointed}.

From this result it  also follows $\operatorname{Int}_{\mathscr{L}}(\operatorname{Cl}_{\mathscr{L}}(\twoheaduparrow x))=\twoheaduparrow x$, indeed $\twoheaduparrow x$ is $\mathscr{L}$-open and $\uparrow x$ is $\mathscr{L}$-closed, so $\operatorname{Cl}_{\mathscr{L}}(\twoheaduparrow x)\subset\uparrow x$ and
$\operatorname{Int}_{\mathscr{L}}(\operatorname{Cl}_{\mathscr{L}}(\twoheaduparrow x))\subset\operatorname{Int}_{\mathscr{L}}(\uparrow x)=\twoheaduparrow x$.

\begin{proof}
\emph{($\supset$).} $\twoheaduparrow x$ is Scott open, hence $\mathscr{L}$-open, and $\twoheaduparrow x\subset\uparrow x$.

\emph{($\subset$).} Let $y\in\operatorname{Int}_{\mathscr{L}}(\uparrow x)$ and take a basic set with
$y\in\, \twoheaduparrow a\backslash \uparrow{F}\subset\, \uparrow x$, $F$ finite. By Thm.\ \ref{oovg} choose $a'$ with
$a\ll a'\ll y$. Then $a'\in\,\twoheaduparrow a$. Moreover $a'\notin\,\uparrow F$: if $f\le a'$ for some $f\in F$ then
$f\le a'\le y$, contradicting $y\notin\,\uparrow F$. Hence $a'\in\,\twoheaduparrow a\setminus\uparrow F\subset\,\uparrow x$, so
$x\le a'\ll y$ and therefore $x\ll y$.
\qed \end{proof}

\begin{theorem}\label{thm:lawconverse}
Let $X$ be a continuous poset. The following are equivalent:
\begin{enumerate}
  \item[(1)] $\operatorname{Cl}_{\mathscr{L}}\big(\twoheaduparrow x\big)=\, \uparrow x$ for every $x\in X$;
  \item[(2)] $x\in\operatorname{Cl}_{\mathscr{L}}(\twoheaduparrow x)$ for every $x\in X$;
  \item[(3)] $\twoheaduparrow x\setminus\uparrow F\ne\emptyset$ for every $x\in X$ and every finite $F\subset X$
    with $x\notin\, \uparrow F$.
\end{enumerate}
\end{theorem}

If they hold, by (3) with $F=\emptyset$,  we have $\twoheaduparrow x\ne \emptyset$ for every $x\in X$.

\begin{proof}
(1) $\Rightarrow$ (2) is immediate.

(2) $\Rightarrow$ (1). We have $\operatorname{Cl}_{\mathscr{L}}(\twoheaduparrow x)\subset\,\uparrow x$. Conversely let $x\le y$ and let $U$ be a $\mathscr{L}$-neighbourhood
of $y$. By (2) at $y$ there is $z\in U\cap\twoheaduparrow y$, and $x\le y\ll z$ gives $x\ll z$, so
$z\in U\cap\twoheaduparrow x$. Hence $y\in\operatorname{Cl}_{\mathscr{L}}(\twoheaduparrow x)$.

(2) $\Leftrightarrow$ (3). The basic $\mathscr{L}$-neighbourhoods of $x$ are the sets
$\twoheaduparrow a\setminus\uparrow F$ with $a\ll x$ and $x\notin\,\uparrow F$. Such a set meets $\twoheaduparrow x$ iff there is
$z$ with $x\ll z$, $a\ll z$ and $z\notin\,\uparrow F$; and $a\ll z$ is automatic, since $a\ll x\le z$.
So the neighbourhood meets $\twoheaduparrow x$ iff $\twoheaduparrow x\setminus\uparrow F\ne\emptyset$, and the choice of
$a$ is irrelevant.

\qed \end{proof}

For $a,b\in X$ the sets  $\uparrow a$,   $\downarrow b$ and
\begin{equation}
[a,b]:=\{x: a\le x \le b \}= \, \uparrow a \, \cap \downarrow b
\end{equation}
are closed in the Lawson topology. By Prop.\ \ref{pldf} in a weakly bicontinuous poset they are also closed in the interval topology.


\begin{proposition} \label{nncr}
A weakly bicontinuous poset is locally compact in the interval topology iff every point $x\in X$ belongs to some interval $(a,b)$ such that $[a,b]$ is compact in the interval topology.
\end{proposition}

Recall that under Hausdorffness the existence of one compact neighborhood is equivalent to the existence of a basis of compact neighborhoods \cite[Thm.\ 18.2]{willard70}.

\begin{proof}
If it is locally compact then $x\in O$ with $O$ relatively compact. Then there is $(c,d)\ni x $ such that $(c,d)\subset O$ hence relatively compact. Interpolating, $x\in (a,b)\subset (c,d)$ thus $[a,b]\subset (c,d)\subset O$, hence $[a,b]$ is compact. The other direction is clear.
\qed \end{proof}

\begin{definition}
A weakly bicontinuous poset is said to be {\em interval-compact} if the intervals $[a,b]$ are compact in the interval topology.
\end{definition}

The following is \cite[Thm.\ 6.1]{martin06}, but  the assumption of bicontinuity is here weakened to weak bicontinuity, see also \cite{keimel09b}.

\begin{theorem}[Interval-compactness implies conditional bicompleteness]\label{oftr} $\empty$\\
  Let $(X,\le)$ be  weakly bicontinuous and interval-compact.  Then every directed subset of $X$ that is bounded above has a supremum, and every filtered subset bounded below has an infimum.
\end{theorem}

We recall that by Prop.\ \ref{mmcl} if $D$ is a directed subset and $\sup D$ exists, then the net $(x_d)_{d \in D}$ (where $x_d = d$) converges to $\sup D$ in the     interval topology.

\begin{proof}
  We treat the directed case; the filtered case is dual.  Let $D
  \subset X$ be directed and bounded above by $u$, and fix $d_0 \in
  D$.  Replacing $D$ by $\{d \in D : d_0 \leq
  d\}$ (which is cofinal and has the same upper bounds) we may assume
  $D \subset [d_0,u]$.
   Since  $[d_0,u]$ is $bi\mathscr{S}$-compact the net  $D$ has a cluster point $m_0$.   For each $d \in D$, the tail $\{d' \in D : d \leq d'\}$
  is a net in the $bi\mathscr{S}$-closed set $\uparrow d$, so $m_0 \in \, \uparrow d$, i.e.\
  $d\leq m_0$, that is, $m_0$ is an upper bound of $D$.  For any upper bound $z$ of $D$, the set $\downarrow{z}$
  is  Scott-closed, hence $bi\mathscr{S}$-closed, and contains every tail of
  $D$, so $m_0 \in \,\downarrow{z}$, i.e.\ $m_0 \leq z$.  Thus $m_0 =
  \sup D$.
\qed \end{proof}

\begin{proof}[Second proof.]
Let $\mathscr{I}$ be the interval topology, then $(X,\mathscr{I}, \le)$ is a closed ordered space by Cor.\ \ref{cnpp}. We proved the statement of the theorem for these more general spaces in the recent paper \cite{minguzzi26c}.
\qed \end{proof}

\begin{proposition}\label{prop:compactintervals}
Let $X$ be weakly bicontinuous and suppose $[a,b]$ is $\mathscr{I}$-compact. Then $\mathscr{L}$ and
$\mathscr{I}$ induce the same topology on $[a,b]$. In particular,  $\mathscr{L}$ and
$\mathscr{I}$ induce the same topology  on every   $\mathscr{I}$-compact set.
\end{proposition}

\begin{proof}
The identity map $([a,b],\mathscr{I})\to([a,b],\mathscr{L})$ is a continuous bijection by Proposition~\ref{pldf}. Its domain is compact and its codomain is Hausdorff, because the Lawson topology of a continuous poset is Hausdorff. A continuous bijection from a compact space to a Hausdorff space is a homeomorphism.

Let $K$ be a $\mathscr{I}$-compact subset. For every $x\in K$ we can find an interval $(a_x,b_x)\ni x$. Passing to a finite covering, this selects points, $a_i,b_i$, $i=1,\ldots, k$, such that $K\subset \cup_i (a_i,b_i)$. Let $I$ be a $\mathscr{I}$-open set, then $I\cap K= \cup_i ([a_i,b_i]\cap K\cap I)= \cup_i ([a_i,b_i]\cap K\cap L_i)$ where $L_i$ is a $\mathscr{L}$-open set such that $L_i\cap  [a_i,b_i]=I\cap [a_i,b_i]$. Let $L=\cup_i L_i$, then $I\cap K=L\cap K$, thus the $\mathscr{I}$ topology induced on $K$ is coarser than the $\mathscr{L}$ topology induced on $K$. The other direction is clear since $\mathscr{I}$ is finer than $\mathscr{L}$
\qed \end{proof}

\begin{example} \label{keim}
Keimel \cite{keimel09b} states that in an interval-compact weakly bicontinuous poset the Lawson topology coincides with the biScott topology. The proof is incorrect in the very last inclusion step and the statement is actually false. A counterexample can be constructed letting $X=\mathcal{P}_{\textrm{fin}}(\mathbb{N})$ be the collection of finite subsets of $\mathbb{N}$ ordered by inclusion. Then $(X,\le)$ is bicontinuous with $\mathscr{I}$-compact intervals and $\ll\,=\,\le\,= \overline{\ll}$, the biScott topology is discrete but the singleton $\{\varnothing\}$ is not  Lawson-open. Thus the Lawson topology is strictly contained in the biScott topology.
\end{example}

\subsubsection{Non-decreasing functions} \label{ndfn}

In relativity one is interested in the continuous real functions which increase along the causal or the chronological relation. It was observed in
\cite{minguzzi09c} that this is the same problem as that of representing a preference by a utility in mathematical economics, and that the theorems of Peleg \cite{peleg70} and Levin \cite{levin83} can be imported into causality theory (see Theorem \ref{mcsd}). It is natural to ask for the order-theoretic version of those results. This subsection provides it. Nothing here uses  dual continuity.

\begin{definition}
Let $(X,\le)$ be a poset. A function $u:X\to \mathbb{R}$ is {\em non-decreasing} or {\em isotone} if $x\le y$ implies $u(x)\le u(y)$, and it is {\em strictly increasing along $\ll$} if $x\ll y$ and $x\ne y$ imply $u(x)<u(y)$.
\end{definition}

Peleg's theorem applies to strict partial orders (transitive and irreflexive relations), so we first observe that the way-below relation always has a canonical strict part, and that no quotient is needed in order to obtain
it. Let $\Delta$ denote the diagonal and set $\ll_s\, :=\, \ll\backslash \Delta$, that is $x\ll_s y$ iff $x \ll y$ and $x\ne y$. We write
$\twoheaddownarrow_s x:=\{z\in X: z\ll_s x\}$.

\begin{lemma} \label{strll}
In every poset, $x\ll y$ and $y \ll x$ imply $x=y$. Consequently the elements at which $\ll$ fails to be irreflexive are exactly the compact ones, and $\ll_s$ is a strict partial order.
\end{lemma}

\begin{proof}
From $\ll\,\subset\,\le$ and antisymmetry we get $x=y$. Hence irreflexivity of $\ll$ can fail only at $x$ with $x\ll x$. The relation $\ll_s$ is irreflexive by construction; it is transitive because $x\ll_s y \ll_s z$ gives $x\ll z$, while $x=z$ would give $x\ll y\ll x$ and
so $x=y$, excluded.
\qed \end{proof}

Thus, were one to quotient $X$ by the equivalence generated by ``$x\ll y$ and $y\ll x$'', in the spirit of the identification of chronology violating classes, one would obtain nothing new: antisymmetry has already collapsed those classes to points, and what survives of the phenomenon is the possible presence of compact elements.

\begin{definition} \label{cctp}
On a continuous poset $\tau_s$ is the topology generated by
the Scott topology $\mathscr{S}$ together with the sets $\twoheaddownarrow_s b$, $b\in X$.
\end{definition}

More precisely, with ``generated'' we mean that the sets $\twoheaduparrow a$ and $\twoheaddownarrow_s b$ form a subbasis of $\tau_s$,
a basis being given by the finite intersections $\twoheaduparrow a\,\cap \twoheaddownarrow_s b_1 \cap \dots \cap \twoheaddownarrow_s b_n$. The pairwise sets $\twoheaduparrow a\,\cap \twoheaddownarrow_s b$ do not in general suffice, for that would require $\twoheaduparrow x$ to
be filtered, that is dual continuity; they do suffice, and give the interval basis $(a,b)$, when $X$ is bicontinuous.

By construction $\mathscr{S}\subset \tau_s$ and the sets $\twoheaddownarrow_s b$ are $\tau_s$-open, which is what Peleg's theorem demands.
If $\ll$ is irreflexive and $X$ is bicontinuous then $\twoheaddownarrow_s b= \bar{\twoheaddownarrow} \,b$ and $\tau_s$ is the interval topology
$\mathscr{I}$; on a spacetime with $\ll\,=I$ it is the Alexandrov topology of relativity, so that $\tau_s$ is coarser than the manifold topology and a $\tau_s$-continuous function is in particular continuous.

For the causal ladder it is convenient to have also the non-strict version of this topology, which admits compact elements without any special treatment.

\begin{definition} \label{cctpn}
On a continuous poset $\tau$ is the topology generated by the Scott topology $\mathscr{S}$ together with the sets $\twoheaddownarrow b$, $b\in X$.
\end{definition}

The remarks made for $\tau_s$ apply verbatim: the sets $\twoheaduparrow a$ and $\twoheaddownarrow b$ form a subbasis, a basis being given by their finite intersections $\twoheaduparrow a\,\cap \twoheaddownarrow b_1\cap \dots \cap \twoheaddownarrow b_n$, and on a spacetime with $\ll\,=I$ the topology $\tau$ is the Alexandrov topology. If $X$ is bicontinuous then $\twoheaddownarrow b=\bar{\twoheaddownarrow}\, b$ and $\tau=\mathscr{I}$, this time with no hypothesis on the compact elements; since $\le$ is $\mathscr{I}$-closed by Cor.\ \ref{cnpp},
the topology $\tau$ is then Hausdorff.

The two topologies coincide in the absence of compact elements. In general $\{b\}$ is $\tau$-open at a compact $b$, so such elements are isolated; whether they are also closed, and hence clopen, depends on the poset.

\begin{lemma} \label{isolcomp}
Let $X$ be a continuous poset and let $b\in X$ be compact. Then
$\twoheaduparrow b\,\cap \twoheaddownarrow b=\{b\}$, hence $\{b\}$ is $\tau$-open.
\end{lemma}

\begin{proof}
As $b\ll b$ we have $b$ in the intersection, while $z$ in the intersection satisfies $b\le z\le b$, so $z=b$. Both sets $\twoheaduparrow b, \twoheaddownarrow b$ are $\tau$-open.
\qed \end{proof}
\begin{proposition} \label{omsc}
Let $X$ be an $\omega$-continuous poset with countable basis $B$. Then the sets
\[
\twoheaduparrow a\,\cap \twoheaddownarrow b_1\cap \dots \cap \twoheaddownarrow b_n, \qquad a,b_1,\dots,b_n\in B,
\]
form a basis of $\tau$. In particular $\tau$ is second countable.
\end{proposition}

\begin{proof}
Let $x\in \,\twoheaduparrow a\,\cap \twoheaddownarrow c_1\cap \dots \cap \twoheaddownarrow c_n$, where $a,c_i\in X$. Since $\twoheaddownarrow x\cap B$ is directed with supremum $x$ and $a\ll x$, there is $a'\in B$ with $a\le a'\ll x$, hence $x\in \,\twoheaduparrow a'\subset\, \twoheaduparrow a$. For each $i$, interpolation applied to $x\ll c_i$ gives $z$ with $x\ll z\ll c_i$, and the directedness of $\twoheaddownarrow c_i\cap B$ with supremum $c_i$ gives $b_i\in B$ with $z\le b_i\ll c_i$; then $x\ll z\le b_i$ yields $x\ll b_i$, so that $x\in\, \twoheaddownarrow b_i\subset\, \twoheaddownarrow c_i$. Intersecting, $x$ belongs to a set of the displayed form contained in the given one.
\qed \end{proof}

\begin{lemma}[Spaciousness] \label{spac}
Let $(X,\le)$ be a continuous poset. If $x\ll_s y$ then
$\overline{\twoheaddownarrow_s x}\,\subset\, \twoheaddownarrow_s y$, the closure being taken in $\tau_s$.
\end{lemma}

\begin{proof}
The set $\downarrow x$ is Scott closed, hence $\tau_s$-closed, and it contains $\twoheaddownarrow_s x$, so $\overline{\twoheaddownarrow_s x}\subset \, \downarrow x$. Let $z\le x$. From $z\le x \ll y$ we get $z\ll y$. Moreover $z\ne y$, for $z=y$ would give $y\le x$ and, with $x\ll y\,\subset\,\le$, the equality $x=y$, which is excluded. Hence $z\in \,\twoheaddownarrow_s y$.
\qed \end{proof}

Notice that spaciousness, which in the spacetime setting is equivalent to future reflectivity \cite[Lemma 1]{minguzzi09c}, is here automatic. We shall see that  future reflectivity secures the identity $\ll\,=I$, cf.\ Thm.\ \ref{kssw}.

\begin{definition}
The relation $\ll_s$ is {\em separable} if there is a countable $B\subset X$ such that $x\ll_s y$ implies the existence of $b\in B$ with $x\ll_s b \ll_s y$.
\end{definition}

\begin{proposition} \label{seppr}
An $\omega$-continuous poset without compact elements has $\ll_s\,=\,\ll$ separable, a countable basis serving as the set $B$.
\end{proposition}

\begin{proof}
Let $x\ll y$ and let $B$ be a countable basis. By interpolation there is $z$ with $x\ll z\ll y$. Since $\twoheaddownarrow y\cap B$ is directed with supremum $y$ and $z\ll y$, there is $b\in B$ with $z\le b$ and $b \ll y$. Then $x\ll z\le b$ gives $x\ll b$. Finally $b\ne y$, since $b=y$ would give $y\ll y$, and $b\ne x$, since $b=x$ would give $x\ll x$.
\qed \end{proof}


\begin{theorem} \label{pelth}
Let $(X,\le)$ be a continuous poset with $\ll_s$ separable. Then there exists a $\tau_s$-continuous function $u:X\to \mathbb{R}$ which is strictly increasing along $\ll$. If moreover $X$ has no compact elements, then $u$ is also non-decreasing, so that $x\le y\Rightarrow u(x)\le u(y)$ and $x\ll_s y\Rightarrow u(x)<u(y)$.
\end{theorem}

\begin{proof}
We apply Peleg's theorem \cite[Thm.\ 3.1]{peleg70} to the topological space $(X,\tau_s)$ and to the strict partial order $\prec$ defined by $x\prec y$ iff $x \ll_s y$, for which $L(y)=\,\twoheaddownarrow_s y$. Condition (2.1) there, the openness of $L(y)$, holds by the definition of $\tau_s$; condition (2.3), separability, is the hypothesis; condition (2.4), spaciousness, is Lemma \ref{spac}. Peleg's theorem provides a $\tau_s$-continuous $u$ with $x\prec y \Rightarrow u(x)<u(y)$, which is the first claim.

For the second claim let $x\le y$ and let $a\ll_s x$. Then $a\ll y$, and $a\ne y$ since $a=y$ would give $y\ll x\le y$ and hence $x=y=a$, contradicting $a\ne x$; thus $u(a)<u(y)$. It remains to observe that $x$ belongs to the $\tau_s$-closure of $\twoheaddownarrow_s x$. Indeed,
a basic $\tau_s$-neighbourhood of $x$ has the form $V=\twoheaduparrow a\cap \twoheaddownarrow_s b_1\cap \dots \cap \twoheaddownarrow_s b_n$ with
$a\ll x$ and $x \ll_s b_i$; by interpolation there is $z$ with $a\ll z\ll x$, and $z\ne x$ because $X$ has no compact elements, while $z\ll x \ll_s b_i$ gives $z\ll_s b_i$ as in Lemma \ref{spac}; hence $z\in V\cap \twoheaddownarrow_s x$. Continuity of $u$ then gives
$u(x)\le \sup u(\twoheaddownarrow_s x)\le u(y)$.
\qed \end{proof}

The Peleg route provides one function. The Levin route provides, under stronger topological hypotheses, the whole family, and with it a representation of the order. Here the relevant topology is the Lawson one, and we recall that a continuous poset is always a closed
ordered space $(X,\mathscr{L},\le)$ , Prop.\ \ref{cjll}.

%

For a partial order a {\em utility} is a real function $u$ such that  $x\le y \Rightarrow u(x)\le u(y)$ and $x<y \Rightarrow u(x)<u(y)$.

\begin{theorem} \label{levth}
Let $(X,\le)$ be a continuous poset whose Lawson topology is locally compact (or, more weakly, a $k_\omega$-space \cite[Prop.\ 3.1]{minguzzi11c}) and second countable. Then there exists a continuous utility and actually a countable family of continuous utilities $\mathcal{U}$  such that
\begin{equation} \label{multiu}
x\le y \quad \Longleftrightarrow \quad u(x)\le u(y) \ \textrm{ for every } u \in \mathcal{U}.
\end{equation}
In other words, the order of a continuous poset is recoverable from its continuous utilities.
\end{theorem}

Due to Cor.\ \ref{cnpp} an analogous result holds for the interval topology on a weakly bicontinuous poset.
\begin{proof}
The relation $\le$ is a closed preorder on the second countable, locally
compact, Hausdorff space $(X,\mathscr{L})$, so Levin's theorem \cite{levin83}, in the form recalled in \cite[Cor.\ 1.4]{minguzzi11c}, provides a continuous utility and the multi-utility
representation.
\qed \end{proof}

The Peleg route applies to the way-below relation, while Thm.\ \ref{levth} applies to the order $\le$ but in the Lawson topology. In order to  obtain the
abstract counterpart of the passage from a semi-time to a time function, one has to feed Levin's theorem with a different relation closed with respect to $\tau$ (it should be a topology that does not impose closure of $\uparrow x$, as Lawson does, otherwise one would climb the causal ladder more than  desired: one can show that in the application to $(M,g)$ past-reflectivity would be gained leading to causal continuity and the stronger assumption of continuity for the poset). The standard causality argument shows that a future distinguishing, future reflecting spacetime has $D_f$ closed and transitive \cite[Prop.\ 105,107]{minguzzi18b} so $K=D_f$ and then applies Levin's theorem to $D_f$; the relation $D_f$ has an order-theoretic counterpart, and the corresponding closedness holds with no hypothesis whatsoever.

\begin{definition} \label{prcf}
On a continuous poset let
\[
x\preceq_f y \quad :\Longleftrightarrow \quad \twoheaduparrow y \subset\, \twoheaduparrow x .
\]
\end{definition}

This is the mirror of the characterization of $\le$ in a continuous poset, which reads $\twoheaddownarrow x\subset \,\twoheaddownarrow y$. The relation $\preceq_f$ is a preorder and, in every poset, $\le\, \subset\, \preceq_f$, for $x\le y \ll w$ gives $x\ll w$. On a spacetime it
is $D_f$, and the inclusion just noted is future reflectivity (this property will be secured by the identification $\ll\,=I$, cf.\ Thm.\ \ref{kssw}).

\begin{theorem} \label{prcfcl}
On every continuous poset the relation $\preceq_f$ is closed in $\tau\times \tau$. That is $(X,\tau, \preceq_f)$ is a closed preordered space. If  $\preceq_f$ is antisymmetric, namely  $\twoheaduparrow x =\, \twoheaduparrow y \Rightarrow x=y$, then $\tau$ is Hausdorff.
\end{theorem}

\begin{proof}
Let $(x,y)$ belong to the closure of $\preceq_f$ and let $w$ be such that $y\ll w$; we prove $x\ll w$. By interpolation there is $v$ with $y\ll v\ll w$, so that $\twoheaddownarrow v$ is a $\tau$-open set containing $y$.

Let $U$ be any $\tau$-neighbourhood of $x$. Since $U\times \twoheaddownarrow v$ is a neighbourhood of $(x,y)$ it meets $\preceq_f$, so there are $x'\in U$ and $y'\in \,\twoheaddownarrow v$ with $\twoheaduparrow y'\subset \,\twoheaduparrow  x'$. From $y'\ll v$ we get $v\in \,\twoheaduparrow y'\subset\, \twoheaduparrow x'$, that is $x'\ll v$ and
$x'\in \,\twoheaddownarrow v$. Hence every neighbourhood of $x$ meets $\twoheaddownarrow v$, that is $x\in \overline{\twoheaddownarrow v}$.

Now $\downarrow v$ is Scott closed, hence $\tau$-closed, and it contains $\twoheaddownarrow v$; therefore $x\le v$, and $x\le v\ll w$ gives $x\ll w$.

The last statement follows from $\Delta=\preceq_f\cap \preceq_f^T$, and closure of the diagonal is equivalent to the Hausdorff property.
\qed \end{proof}


Recall that $\preceq_f$ is a preorder so by $x\prec_f y$  we mean $x\preceq_f y$ and $y \not\preceq_f x$.

\begin{lemma} \label{tofd}
If $\tau$ is $T_1$ and $x<y$ then $\twoheaduparrow x\ne\,\, \twoheaduparrow y$, that is $x\prec_f y$.
\end{lemma}

\begin{proof}
Assume $\twoheaduparrow x=\, \twoheaduparrow y$ with $x<y$. Every $\twoheaduparrow a$ containing $x$ contains $y$, since $a\ll x\le y$; and every $\twoheaddownarrow b$ containing $x$ contains $y$, since $x\ll b$ gives $b\in \,\twoheaduparrow x=\,\twoheaduparrow y$ and hence $y\ll b$. So every basic neighbourhood of $x$ contains $y$, against the $T_1$ property.
\qed \end{proof}

\begin{proposition} \label{prop:t1char}
Let $(X,\le)$ be a continuous poset. Then $\tau$ is $T_1$ if and only if  ``$x\le y$ with $\twoheaduparrow x=\,\twoheaduparrow y \Rightarrow x=y$''.
\end{proposition}

Note that the stronger property ``$\twoheaduparrow x=\,\twoheaduparrow y \Rightarrow x=y$'' implies the Hausdorff condition by Thm.\ \ref{prcfcl} (the converse can be proved if the sets $\twoheaduparrow x$ are non-empty and filtered; we omit the proof).

\begin{proof}
Fix $x\ne y$. No basic $\tau$-neighbourhood of $x$ avoids $y$ precisely when $a\ll y$ for every $a\ll x$, and $y\ll b$ for every $b$ with $x\ll b$; that is, precisely when $\twoheaddownarrow x\subset \,\twoheaddownarrow y$ and $\twoheaduparrow x\subset\, \twoheaduparrow y$.

Assume this happens. From $\twoheaddownarrow x\subset\,\twoheaddownarrow y$ every $a\ll x$ satisfies $a\le y$, so $y$ is an upper bound of $\twoheaddownarrow x$ and $x=\sup \twoheaddownarrow x\le y$, whence $x<y$; and $x\le y$ gives $\twoheaduparrow y\subset\, \twoheaduparrow x$, so $\twoheaduparrow x=\,\twoheaduparrow y$. Conversely if $x<y$ and
$\twoheaduparrow x=\,\twoheaduparrow y$, then $a\ll x\le y$ gives $a\ll y$ for every $a\ll x$, and $b\in \,\twoheaduparrow x=\,\twoheaduparrow y$ gives $y\ll b$ for every $b\gg x$.

Therefore $\tau$ fails to be $T_1$ exactly when such a pair $x<y$ exists, a failure of $T_1$ at the pair $(x,y)$ being the statement that $x$ lies in the closure of $\{y\}$.
\qed \end{proof}

We shall see that for a spacetime $(M,g)$ which can be regarded as a continuous poset with way below relation  $I$ the property selected by the previous result is equivalent to distinction,  Cor.\ \ref{bcot}. Then by \cite[Prop.\ 105,107]{minguzzi18b} the causality property is promoted to stable causality and hence strong causality, thus to the $T_2$ property for $\tau$. We conclude that in the regular spacetime case, under poset continuity, the $T_1$ and $T_2$ property for $\tau$ are equivalent. This is not true for a general continuous poset and it is the $T_2$ property that ensures the existence of increasing functions.

\begin{theorem} \label{levthf}
Let $X$ be a continuous poset whose topology $\tau$ is Hausdorff, locally compact (or, more weakly, a $k_\omega$-space) and second countable. Then there is a $\tau$-continuous $u:X\to \mathbb{R}$ with
\[
x< y \quad \Longrightarrow \quad u(x)<u(y),
\]
and a suitable countable family of $\tau$-continuous
$\preceq_f$-utilities, $\mathcal{U}_f$, such that
\[
x \preceq_f y \quad \Longleftrightarrow \quad u(x)\le u(y) \ \textrm{ for every } u\in \mathcal{U}_f .
\]
\end{theorem}

Observe that if there is a  $\preceq_f$-increasing function, i.e.\ $x\preceq_f y \Rightarrow u(x)\le u(y)$ with equality case only under $x=y$, then $\preceq_f$ is antisymmetric and so $\tau$ is Hausdorff by Thm.\ \ref{prcfcl}.

\begin{proof}
By Thm.\ \ref{prcfcl} the relation $\preceq_f$ is a closed preorder on a second countable, locally compact Hausdorff space, so Levin's theorem \cite{levin83}, in the form recalled in \cite[Cor.\ 1.4]{minguzzi11c}, applies and gives a continuous utility together with the multi-utility representation. If $x<y$ then $x\preceq_f y$, and $x\prec_f y$  by
Lemma \ref{tofd}, whence $u(x)<u(y)$.
\qed \end{proof}

\begin{remark} \label{fdst}
The relation $\preceq_f$ is antisymmetric precisely when
\begin{equation} \label{fdist}
\twoheaduparrow x=\,\,\twoheaduparrow y \ \Rightarrow\ x=y ,
\end{equation}
the mirror of the injectivity of $\twoheaddownarrow$ which holds in every continuous poset, and the abstract form of future distinction. Condition \eqref{fdist} is {\em not} implied by the Hausdorffness of $\tau$. For a counterexample let $X$ consist of two disjoint copies of $\mathbb{R}$, with their usual orders, together with two further elements $x$ and $y$, the first above every element of the first copy, the second above every element of the second copy, and with no other relation. This is a continuous poset, $\tau$ is Hausdorff, the sets $\twoheaduparrow (t,1)$ and $\twoheaduparrow (t,2)$ separating $x$ from $y$, and yet $\twoheaduparrow x=\twoheaduparrow y=\emptyset$. Condition \eqref{fdist} does hold whenever $\tau$ is Hausdorff and every $\twoheaduparrow x$ is non-empty and filtered, which is the case on a spacetime.

Let us stress that Thm.\ \ref{levthf} does not need \eqref{fdist}: Levin's theorem applies to preorders, and the strictness of the conclusion comes from Lemma \ref{tofd}, which uses only the $T_1$ property. What \eqref{fdist} adds is that $\preceq_f$ be an order, so that the multi-utility representation becomes a representation of a partial order.
\end{remark}

\begin{remark}
Anticipating an application of the notion of bicontinuity, one has $\preceq_f\,=\,\le$ if and only if $\twoheaduparrow y\subset \, \twoheaduparrow x$ implies $x\le y$, which holds under bicontinuity, since there $\twoheaduparrow x$ is filtered with infimum $x$ and hence $\twoheaduparrow y \subset\, \twoheaduparrow x$ gives
$x=\inf \twoheaduparrow x \le \inf \twoheaduparrow y=y$.
\end{remark}

\begin{remark}
The results dualize without effort: applying them to $(X,\le^{\textrm{op}})$ one obtains, for a co-continuous poset, a continuous function strictly decreasing along the way-above relation and the corresponding recovery of the order. This is in agreement with the
spacetime statement, Thm.\ 5 of \cite{minguzzi09c}, in which the hypothesis is chronology together with future {\em or} past reflectivity: under the dictionary of Sec.\ \ref{qcft} these two alternatives are precisely the continuous and the co-continuous case.
\end{remark}

\subsubsection{Obtaining equality of Lawson and interval topologies} \label{coinc}

In Section \ref{qcft} we shall prove that causally continuous spacetimes provide examples of weakly bicontinuous posets such that the way below relation coincides with the chronological relation $I$, and where $\le$ is then forced to be the closure of the causal relation. By using the exponential map we shall prove that the Lawson topology coincides, in this example, with the interval topology, cf.\ Thm.\ \ref{kxqp}. Actually, working in a purely order theoretic setting, it turns out that strong properties such as interval-compactness are not sufficient to induce the equality. This raises the question of what abstract properties, not relying on the manifold structure,  ensure the coincidence of the two topologies. This subsection is devoted to answering this question.

The equivalence over compact subsets was obtained in Prop.\ \ref{prop:compactintervals}.
Passing from this local statement to a global one is non-trivial.
We isolate the condition that makes it
work. Throughout, for $b\in X$ put
\begin{equation}
  M_b:=X\setminus\bar \twoheaddownarrow\, b=\{z\in X: z\,\not\!\!\overline{\ll} \,\, b\},
\end{equation}
which is an $\mathscr{I}$-closed  upper set.

\begin{definition}[(forward) Domination]\label{def:domination}
$X$ satisfies condition \textup{(D)} if for all $a,b\in X$ with $a \,\, \overline{\ll} \,\,b$,
\begin{itemize}
\item[(D)] \qquad \qquad \qquad \qquad \qquad $M_b\,\cap\uparrow a\ \subset\ \uparrow\big(M_b\cap [a,b]\big)$.
\end{itemize}
\end{definition}

For the notion of  $\mathscr{I}$-local compactness see Prop.\ \ref{nncr}.
\begin{theorem}\label{thm:lawequalsint}
Let $X$ be weakly bicontinuous, $\mathscr{I}$-locally compact, and suppose \textup{(D)} holds. Then $\mathscr{L}=\mathscr{I}$.
\end{theorem}

\begin{proof}
By Prop.\ \ref{pldf} it suffices to show that $\mathscr{L}$ is finer than $\mathscr{I}$ and, since a basis for the latter topology has the form $\twoheaduparrow a \cap \bar \twoheaddownarrow\,  b$, for that it suffices to show that $\bar \twoheaddownarrow \, b$ is an $\mathscr{L}$-neighbourhood of each of its points, so we fix $x\,\, \overline{\ll}\,\, b$.

By local compactness there are $a_0,b_0$ with $x\in (a_0,b_0)$ and
$[a_0,b_0]$ compact; put $a:=a_0$, so $a\ll x\,\, \overline{\ll}\,\, b_0$. By dual continuity $\bar \twoheaduparrow x$ is filtered, so choose
$b'$ with $x \,\, \overline{\ll}\,\, b'$, $b'\le b_0$ and $b'\le b$. Then $[a,b']\subset [a,b_0]$ is $\mathscr{I}$-closed inside a compact set, hence $\mathscr{I}$-compact. Finally, $a\le x \,\, \overline{\ll}\,\, b'$ gives $a \,\, \overline{\ll}\,\, b'$ so (D) applies to the pair $a,b'$.

Put $K:=M_{b'}\cap [a,b']$, which is compact because $M_{b'}$ is $\mathscr{I}$-closed. No $w\in K$ satisfies $w\le x$: otherwise $w\le x\,\, \overline{\ll}\,\, b'$ would give $w\,\, \overline{\ll}\,\, b'$, contradicting $w\in M_{b'}$. Hence, for $w\in K$, not every $f\ll w$ satisfies $f\le x$, since that would force $w=\sup\twoheaddownarrow w\le x$; choose $f_w\ll w$ with $f_w\not\le x$. The Scott open sets $\twoheaduparrow{f_w}$ are $\mathscr{I}$-open and cover $K$, so compactness yields a finite $F=\{f_1,\dots,f_n\}$ with $K\subset\,\uparrow F$ and $x\notin\,\uparrow F$.

 Let $z\in\,\twoheaduparrow a$ with $z\notin  \bar \twoheaddownarrow \,b'$, i.e.\ $z\in M_{b'}$. Since $a\ll z$ gives $a\le z$, we have $z\in M_{b'}\,\cap\uparrow a\subset\,\uparrow K\subset\,\uparrow F$ by (D). Therefore
\[
  x\ \in\ \twoheaduparrow a\setminus\uparrow F\ \subset\ \bar \twoheaddownarrow\, b'\ \subset\ \bar \twoheaddownarrow \,b .
\]
\qed \end{proof}

\begin{remark}\label{rem:whatisused}
Compactness of a full interval is not needed: the proof uses only that
$K=M_{b'}\cap [a,b']$ is compact, and $K$ is the part of the $\mathscr{I}$-boundary of $\downarrow{b'}$ lying above $a$. The minimal hypothesis is therefore that for each $x$ and each $b$ with $x\,\,\overline{\ll}\,\, b$ there exist $a\ll x$ and $b'$ with $x\,\,\overline{\ll}\,\, b'\le b$ such that $M_{b'}\cap [a,b']$ is
compact. Local compactness is the convenient sufficient form.
\end{remark}

We now give two sufficient conditions for (D), which is why it deserves a name: the two natural hypotheses one might impose both act only through (D) to give the equality of the topologies. A third condition implying (D) will be introduced in Sec.\ \ref{chart}.

\begin{theorem}[Connected intervals]\label{thm:connectedD}
Let $X$ be weakly bicontinuous and suppose $[a,b]$ is $\mathscr{I}$-connected for all $a\le b$. Then \textup{(D)} holds.
\end{theorem}

\begin{proof}
Let $a\,\, \overline{\ll}\,\, b$ and $z\in M_b\,\,\cap\uparrow a$. If $z\le b$ then $z\in K:=M_b\cap [a,b]$ and we are done, so assume $z\not\le b$ and put $C:= [a,z]$, connected by hypothesis. From $a\,\, \overline{\ll}\,\, b$
we get $a\in\bar \twoheaddownarrow \, b=\operatorname{Int}_{\mathscr{I}}(\downarrow b)$ by Thm.\ \ref{thm:pointed}, while $z$ lies in the open set $X\setminus \downarrow b$. If $C$ missed $\p_{\mathscr{I}}(\downarrow b)$, the disjoint open sets
$\operatorname{Int}_{\mathscr{I}}(\downarrow b)$ and $X\setminus\downarrow b$ would cover $C$ and both meet it, contradicting connectedness. So there is $w\in C\cap\p_{\mathscr{I}}(\downarrow b)$. As $\downarrow b$ is closed, $w\le b$; and $a\le w$, so $w\in [a,b]$; and $w\notin\operatorname{Int}_{\mathscr{I}}(\downarrow b)=\bar \twoheaddownarrow \, b$, so $w\in M_b$. Thus $w\in K$ and $w\le z$.
\qed \end{proof}

Theorem~\ref{thm:connectedD} combines with Theorem~\ref{thm:lawequalsint} to give:

\begin{theorem}\label{thm:lawequalsint2}
Let $X$ be weakly bicontinuous, $\mathscr{I}$-locally compact and suppose $[a,b]$ is $\mathscr{I}$-connected for all $a\le b$. Then $\mathscr{L}=\mathscr{I}$.
\end{theorem}

\begin{definition}\label{def:nojumps}
$X$ satisfies
\begin{itemize}
\item[(no jumps: (NJ))] if for all $a<b$ there is a directed
$D\subset [a,b]\setminus\{b\}$ with $\sup D=b$;
\item[(density)] if $a<b$ implies
$a<x<b$ for some $x$.
\end{itemize}
\end{definition}

\begin{proposition}\label{prop:daggerD}
Let $X$ be weakly bicontinuous and conditionally bicomplete, and suppose (NJ) holds. Then \textup{(D)} holds.
\end{proposition}

\begin{proof}
Let $a\,\, \overline{\ll}\,\, b$ and $y\in M_b\,\cap\uparrow a$, and set $S:=M_b\cap [a, y]$, which contains $y$.

\emph{$S$ is chain-complete from below.} Let $C\subset S$ be a chain. It is filtered and bounded below by $a$, so $\inf C$ exists by conditional bicompleteness, and $a\le\inf C\le y$. Moreover $\bar \twoheaddownarrow\, b=\operatorname{Int}_{\mathscr{S}^{\textrm{op}}}(\downarrow b)$ (see Thm.\ \ref{thm:scott}) is $\mathscr{S}^{\textrm{op}}$-open, hence inaccessible by filtered infima, so $\inf C\in\bar \twoheaddownarrow\, b$ would force some element of $C$ into $\bar \twoheaddownarrow\, b$; as $C\subset M_b$ this is
impossible, hence $\inf C\in S$.

By Zorn's lemma $S$ has a minimal element $m$, which is minimal in $M_b\,\cap\uparrow a$: if $m'\in M_b$ satisfied $a\le m'<m$, then $m'\le m\le y$ would put $m'$ in $S$.

Since $a\,\, \overline{\ll}\,\, b$ we have $a\in\bar \twoheaddownarrow\,  b$, so $a\notin M_b$ and hence $a<m$. By (NJ) there is a
directed $D\subset [a,m]\setminus\{m\}$ with $\sup D=m$. Each $d\in D$ satisfies $a\le d<m$, so $d\notin M_b$ by minimality, i.e.\ $d\,\, \overline{\ll}\,\, b$ and in particular $d\le b$. Thus $b$
is an upper bound of $D$ and $m=\sup D\le b$. Therefore $m\in M_b\cap [a,b]$ and $m\le y$.
\qed \end{proof}


Remember that interval-compactness implies conditional bicompleteness by Theorem \ref{oftr}, and local compactness by Remark \ref{mvvr} and Prop.\ \ref{nncr}.

\begin{proposition}[Comparison of the hypotheses]\label{prop:hypcomparison}
Let $X$ be weakly bicontinuous and conditionally bicomplete.
Density and (NJ) are equivalent, and both are implied by the condition that all  intervals $[a,b]$ are $\mathscr{I}$-connected.
\end{proposition}
 All  three conditions act only through \textup{(D)} to determine the coincidence of $\mathscr{L}$ and $\mathscr{I}$.
\begin{proof}
(NJ) implies density (this direction  holds in any poset). Let $a<m$ and take $D$ as in (NJ). If every $d\in D$ equalled $a$ then $\sup D=a\ne m$;
so some $d$ has $a<d<m$.

Density implies (NJ). Let $a<m$ and let $C$ be a chain in $[a,m]\setminus\{m\}$ containing $a$ and maximal among such chains, by Zorn. It is directed and bounded above by $m$, so $s:=\sup C$ exists by
conditional bicompleteness and $a\le s\le m$. If $s<m$, density gives $s<x<m$, and $C\cup\{s,x\}$ would be a strictly larger chain in $[a,m]\setminus\{m\}$. Hence $s=m$ and $C$ is directed, lies in $[a,m]\setminus\{m\}$ and has supremum $m$, as (NJ) requires.

$\mathscr{I}$-connected implies both. If $a<m$ and $[a,m]=\{a,m\}$, this subspace is Hausdorff, hence discrete, hence
disconnected.
\qed \end{proof}

Connectedness of the intervals is therefore the strongest of the three conditions in general. However, if the intervals are not merely connected but compact, the three conditions coincide: by Theorem \ref{thm:densityconnected}, density then implies that
every interval is connected. In that case the hypotheses of Proposition
\ref{prop:daggerD} are redundant, since (D) already follows from Theorem
\ref{thm:connectedD}, and Proposition \ref{prop:daggerD} retains its own domain only under the weaker assumption of local compactness, where the argument of Theorem \ref{thm:densityconnected} is unavailable because $[a,b]$ need not be compact for the pair $a,b$ at hand.

\begin{corollary}\label{cor:ghlaw}
Let $X$ be weakly bicontinuous, $\mathscr{I}$-locally compact, conditionally bicomplete and such that (NJ) (equiv.\ density) holds. Then $\mathscr{L}=\mathscr{I}$.
\end{corollary}

\begin{proof}
Combine Theorem~\ref{thm:lawequalsint} with
Proposition~\ref{prop:daggerD}.
\qed \end{proof}

\subsection{Bicontinuous posets}

\begin{definition}
A poset $(X,\le)$ is {\em bicontinuous}\footnote{Notice that in \cite{keimel09b,ebrahimi14,ebrahimi15} these posets are called {\em jointly bicontinuous}. We follow Martin and Panangaden's terminology \cite{martin06}.}
if both $(X,\le)$ and $(X,\le^{\textrm{op}})$  are continuous and the way-below relation of the latter is the dual of the way-below relation of the former. 
Equivalently, it is weakly bicontinuous and $\ll\,=\overline{\ll}$.
\end{definition}

\begin{remark} \label{reqt}
Thus, for a bicontinuous poset  the sets
\begin{equation}
(a,b):= \{x: a\ll x \ll b \}= \, \twoheaduparrow a \, \cap \twoheaddownarrow b
\end{equation}
form a basis for the   interval topology.
\end{remark}
Observe that most results have been obtained under a weak bicontinuity assumption. However, in a spacetime interpretation in which the way-below relation is interpreted as chronology the bicontinuity assumption seems more natural unless one wants to work with two chronological relations.

The following is \cite[Prop.\ 6.2]{martin06}

\begin{proposition} \label{pwdt}
In a bicontinuous poset $\mathscr{I}$-separability is equivalent to $\mathscr{I}$-second countability. Under this condition and interval-compactness $(X,\mathscr{I})$ is Polish (completely metrizable).
\end{proposition}

\begin{proof}
Let $C$ be countable and $\mathscr{I}$-dense and let $x\in (a,b)$ then, by density there are $c_1\in (a,x)\cap C$ and $c_2\in (x,b)\cap C$, thus $x\in (c_1,c_2)\subset (a,b)$.
A locally compact, second-countable Hausdorff topological space is completely metrizable.
\qed \end{proof}

Interestingly, the $\mathscr{I}$-openness of $\ll$ is obtained under bicontinuity.

\begin{theorem}\label{intlemma}
Let $(X,\le)$ be a bicontinuous poset and give $X\times X$ the product of the interval topologies.
Then $\ll$ is $\mathscr{I}\times \mathscr{I}$-open, and
\begin{equation}
  \ll\ =\ \operatorname{Int}_{\mathscr{I}\times \mathscr{I}}(\le),\qquad \operatorname{Cl}_{\mathscr{I}\times \mathscr{I}}(\ll)\ =\ \le .
\end{equation}
\end{theorem}

\begin{proof}
\emph{Openness.} Let $x\ll y$ and interpolate twice, $x\ll u\ll v\ll y$. If $x'\ll u$ and $v\ll y'$
then $x'\ll u\ll v\ll y'$ so $x'\ll y'$. Hence
$\twoheaddownarrow u\times\twoheaduparrow v$ is an open neighbourhood of $(x,y)$ contained in $\ll$, the sets $\twoheaddownarrow u$ and
$\twoheaduparrow  v$ being open in the interval topology.

\emph{$\operatorname{Int}(\le)=\ \ll$.} Let $(x,y)$ be interior to $\le$ and choose a basic
neighbourhood $(a,b)\times(c,e)\subset\ \le$, so that $a\ll x\ll b$ and $c\ll y\ll e$. Interpolation gives $z$ such that $c\ll z\ll y\ll e$, in particular $z\in (c,e)$.
As $x\in(a,b)$ and
$(a,b)\times(c,e)\subset\ \le$, we get $x\le z$, and therefore $x\le z\ll y$ gives $x\ll y$. The reverse inclusion is the openness just proved together with $\ll \ \subset \ \le$.

\emph{$\operatorname{Cl}(\ll)=\ \le$.} Since, by Cor.\ \ref{cnpp}, $\le$ is closed   and contains $\ll$, we
have $\operatorname{Cl}(\ll)\subset\ \le$. Conversely let $x\le y$ and let $(a,b)\times(c,e)$ be a basic
neighbourhood of $(x,y)$. Interpolating $a\ll x'\ll x$ we get $x'\in(a,b)$, because $x'\ll x\ll b$
gives $x'\ll b$; and $x'\ll x\le y$ gives $x'\ll y$. Hence
$(x',y)\in\big((a,b)\times(c,e)\big)\,\cap\ll$, so $(x,y)\in \operatorname{Cl}(\ll)$.
\qed \end{proof}

In Thm.\ \ref{thm:pointed} we have established that in a weakly bicontinuous poset $\operatorname{Int}_{\mathscr{I}}(\uparrow x)=\, \twoheaduparrow x$. We now investigate whether $\operatorname{Cl}_{\mathscr{I}}(\twoheaduparrow x)=\, \uparrow x$, namely if $\uparrow x$ is regularly closed.

\begin{lemma}[Density criterion]\label{lem:density}
Let $X$ be weakly bicontinuous and $x\in X$. Then $x\in\operatorname{Cl}_{\mathscr{I}}(\twoheaduparrow x)$ if and only if for
every $b$ with $x\,\, \overline{\ll}\,\, b$ there exists $z$ with $x\ll z\,\, \overline{\ll}\,\, b$. In relational form, the condition
``$x\in \operatorname{Cl}_{\mathscr{I}}(\twoheaduparrow x)$ for all $x$'' says
\[
  \overline{\ll}\ \subset\ \overline{\ll}\ \circ\ll .
\]
\end{lemma}

\begin{proof}
By Prop.\ \ref{mmft} the basic neighbourhoods of $x$ are the sets $(a,b)$ with
$a\ll x\,\, \overline{\ll} \,\, b$.
Assume the stated condition and let $(a,b)$ be such a neighbourhood. Since $x\,\, \overline{\ll} \,\, b$ there
is $z$ with $x\ll z \,\, \overline{\ll} \,\, b$; from $a\ll x\le z$ we get $a\ll z$, so
$z\in (a,b) \, \cap\twoheaduparrow x$. Thus, under the stated condition, by the arbitrariness of $(a,b)$, $x\in \operatorname{Cl}_{\mathscr{I}}(\twoheaduparrow x)$.

Conversely, assume $x\in \operatorname{Cl}_{\mathscr{I}}(\twoheaduparrow x)$ and let $x\,\, \overline{\ll} \,\, b$. By continuity pick $a\ll x$, so
$(a,b)$ is a neighbourhood of $x$ and therefore contains some $z\in\,\twoheaduparrow x$. Then $x\ll z$
and $z \,\, \overline{\ll} \,\, b$, as required.
\qed \end{proof}

A subset of a topological space is {\em regularly closed} if it is the closure of its interior.

\begin{theorem}\label{thm:pointedconverse}
Let $X$ be weakly bicontinuous. The following are equivalent:
\begin{enumerate}
  \item[(1)] $\operatorname{Cl}_{\mathscr{I}}\big(\twoheaduparrow x\big)= \, \uparrow x$ for every $x\in X$ (i.e.\ regular closure of $\uparrow x$);
  \item[(2)] $x\in\operatorname{Cl}_{\mathscr{I}}(\twoheaduparrow x)$ for every $x\in X$;
  \item[(3)] $\overline{\ll}\ \subset\ \overline{\ll}\ \circ\ll$.
\end{enumerate}
\end{theorem}

By duality, we have also the equivalence between  (1') $\operatorname{Cl}_{\mathscr{I}}(\bar\twoheaddownarrow x)= \, \downarrow x$ for every $x\in X$; (2')   $x\in\operatorname{Cl}_{\mathscr{I}}(\bar \twoheaddownarrow x)$ for every $x\in X$; (3') $\ll \, \,\subset \, \overline{\ll}\ \circ\ll $.

\begin{proof}

(1) $\Rightarrow$ (2) is immediate since $x\in \,\uparrow x$, and (2) $\Leftrightarrow$ (3) is
Lemma~\ref{lem:density}.

(2) $\Rightarrow$ (1). The inclusion $\operatorname{Cl}_{\mathscr{I}}(\twoheaduparrow x)\subset \uparrow x$ holds because $\uparrow x$ is closed and contains $\twoheaduparrow x$. For the reverse, let $x\le y$ and let $U$ be any $\mathscr{I}$-neighbourhood of
$y$. By (2) applied to $y$ we have $y\in\operatorname{Cl}_{\mathscr{I}}(\twoheaduparrow  y)$, so there is $z\in U\cap\twoheaduparrow y$, i.e.\
$y\ll z$. Then $x\le y\ll z$ gives $x\ll z$, so
$z\in U\cap\twoheaduparrow x$. As $U$ was arbitrary, $y\in\operatorname{Cl}_{\mathscr{I}}(\twoheaduparrow x)$.
\qed \end{proof}

\begin{lemma}\label{lem:compositecontained}
In any poset, $\overline{\ll}\ \circ\ll \ \subset\ \ll\,\cap\ \overline{\ll}$.
\end{lemma}

\begin{proof}
If $x\ll z\,\,\overline{\ll}\,\, y$ then $x\ll z\le y$ gives $x\ll y$, and $x\le z\,\,\overline{\ll}\,\, y$ gives $x\,\,\overline{\ll}\,\, y$.
\qed \end{proof}

We arrive at a topological characterization of bicontinuity.

\begin{theorem}\label{thm:regclosed}
Let $X$ be weakly bicontinuous. The following are equivalent
\begin{itemize}
\item[(i)]  $X$ is bicontinuous;
\item[(ii)] $\uparrow x$ and $\downarrow x$ are regularly closed;
\item[(iii)]  $\operatorname{Cl}_{\mathscr{I}}(\twoheaduparrow x)=\, \uparrow x$
  and
$\operatorname{Cl}_{\mathscr{I}}(\bar \twoheaddownarrow \,x)=\, \downarrow x$
  for all $x$;
\item[(iv)]   $x\in \operatorname{Cl}_{\mathscr{I}}(\twoheaduparrow x)$ and $x\in
\operatorname{Cl}_{\mathscr{I}}(\bar \twoheaddownarrow \,x)$
  for all $x$;
\end{itemize}
\end{theorem}

\begin{proof}
Properties (ii), (iii), (iv) are equivalent by Theorems   \ref{thm:pointed} and \ref{thm:pointedconverse}.

(iii) $\Rightarrow$ (i). By Theorem~\ref{thm:pointedconverse} the two hypotheses say $\overline{\ll}\,\subset\,\overline{\ll}\,\,\circ\ll$ and
$\ll\,\,\subset\,\overline{\ll}\,\,\circ\ll$. Together with Lemma~\ref{lem:compositecontained} this yields
\[
  \ll\, \cup\ \overline{\ll}\ \subset\ \overline{\ll}\,\,\circ\ll\ \, \subset\ \ll \cap \  \overline{\ll},
\]
forcing $\ll\,=\,\overline{\ll}$, i.e.\ bicontinuity.

(i) $\Rightarrow$ (iii).  If $X$ is bicontinuous then $\ll\,\,=\,\overline{\ll}$, so $\overline{\ll}\,\,\circ\ll\,=\,\ll\circ\ll\,=\,\ll\,=\,\overline{\ll}$ and  Theorem~\ref{thm:pointedconverse}  gives the result.
\qed \end{proof}

\begin{example}\label{ex:unitinterval}
Let $X=[0,1]$ with its usual order. Then
\[
  x\ll y\iff x<y\ \text{ or }\ x=y=0,\qquad
  x\, \,\overline{\ll}\,\, y\iff x<y\ \text{ or }\ x=y=1 .
\]
Indeed $\twoheaddownarrow y=[0,y)$ for $y>0$ and $\twoheaddownarrow 0=\{0\}$, directed with supremum
$y$; dually $\bar \twoheaduparrow \, x=(x,1]$ for $x<1$ and $\bar \twoheaduparrow 1=\{1\}$, filtered with
infimum $x$. So $X$ is weakly bicontinuous. It is \emph{not} bicontinuous, since $1\,\,\overline{\ll}\, \,1$ while
$1\not\ll 1$ (and $0\ll 0$ while $0\,\not{\!\!\overline{\ll}} \ 0$). The sets $(a,b)$ generate the usual
topology of $[0,1]$.

Here $\twoheaduparrow 1=\emptyset$, so by Theorem~\ref{thm:pointed} $\operatorname{Int}_{\mathscr{I}}(\uparrow 1)=\emptyset$ and
\[
  \operatorname{Cl}_{\mathscr{I}}\big(\operatorname{Int}_{\mathscr{I}}(\uparrow 1)\big)=\emptyset\ \ne\ \{1\}=\uparrow 1 .
\]
Dually $\bar \twoheaddownarrow \, 0=\emptyset$ and $\operatorname{Cl}_{\mathscr{I}}(\operatorname{Int}_{\mathscr{I}}(\downarrow 0))=\emptyset\ne\{0\}=\downarrow 0$. So both
pointed converses fail, a failure predicted by
Theorem~\ref{thm:regclosed}.
\end{example}

We have also the following curious mixed result

\begin{proposition}\label{prop:mixed}
Let $X$ be a continuous poset in which $\twoheaduparrow x$ is filtered with $\inf \twoheaduparrow x=x$ for every
$x$ --- in particular, let $X$ be bicontinuous. Then for every $x\in X$,
\[
  \operatorname{Cl}_{\mathscr{S}^{\textrm{op}}}\big(\twoheaduparrow x\big)=\, \uparrow x .
\]
\end{proposition}
We recall that by Theorem~\ref{thm:scott}, in a continuous poset $\operatorname{Int}_{\mathscr{S}}(\uparrow x)=\twoheaduparrow x$, so this result uses both the Scott and the dual-Scott topologies.
\begin{proof}

Recall that the $\mathscr{S}^{\textrm{op}}$-closed sets are exactly the upper sets stable under infima of filtered
subsets, and that this family is closed under arbitrary intersections; hence
$\operatorname{Cl}_{\mathscr{S}^{\textrm{op}}}(A)$ is the smallest upper set containing $A$ and stable under filtered infima.

\emph{($\subset$).} The set $\uparrow x$ is
$\mathscr{S}^{\textrm{op}}$-closed, and $\twoheaduparrow x\subset\, \uparrow x$, hence
$\operatorname{Cl}_{\mathscr{S}^{\textrm{op}}}(\twoheaduparrow  x)\subset\uparrow x$.

\emph{($\supset$).} Write $C=\operatorname{Cl}_{\mathscr{S}^{\textrm{op}}}(\twoheaduparrow x)$. By hypothesis $\twoheaduparrow x$ is filtered and
non-empty with $\inf \twoheaduparrow x=x$. Since $\twoheaduparrow x\subset C$ and $C$ is stable under infima of
filtered subsets, $x=\inf \twoheaduparrow x\in C$. Since $C$ is an upper set, $\uparrow x\subset C$.

If $X$ is bicontinuous then $\ll\,=\,\overline{\ll}$ gives $\twoheaduparrow x=\bar \twoheaduparrow x$, which is filtered with infimum
$x$ by dual continuity, so the hypothesis is satisfied.
\qed \end{proof}

The following definition comes from \cite[Sec.\ 6]{martin06}

\begin{definition}
A {\em globally hyperbolic poset} is a bicontinuous interval-compact poset. 
\end{definition}

By Thm.\ \ref{oftr} globally hyperbolic posets are conditionally bicomplete. By Remark \ref{mvvr} and Prop.\ \ref{nncr} they have a locally compact interval topology.



\begin{theorem}
Let $X$ be  a globally hyperbolic poset such that (a)  intervals $[a,b]$ are connected, or (b) density: for each $a<b$ there is $x$ such that $a<x<b$, then $\mathscr{L}=\mathscr{I}$.
\end{theorem}

This theorem will be improved with Thm.\ \ref{thm:ghlawson}.

\begin{proof}
For case (a), interval-compactness implies local compactness (Remark \ref{mvvr} and Prop.\ \ref{nncr}); then use Thm.\ \ref{thm:lawequalsint2}.
For case (b), by interval-compactness and Thm.\ \ref{oftr} the space is conditionally bicomplete, then use Cor.\ \ref{cor:ghlaw}.
\qed \end{proof}

\section{From causality of $(M,g)$ to continuity of $(M,D_p)$} \label{qcft}

 Throughout this work, all manifolds are assumed Hausdorff and second countable, and hence paracompact. A spacetime $(M, g)$ is a connected, time-oriented Lorentzian manifold with metric $g$ of signature $(-,+,\ldots,+)$, assumed $C^2$ (although $C^{1,1}$ would be sufficient). When we speak of a curve $\gamma$, we may refer either to the map $\gamma: I \to M$ or to its image. We refer to \cite{minguzzi18b} for most causality-theoretic results. The spacetime results presented here generalize to the Lorentz-Finsler case, as the entire causality theory translates to that setting almost word for word.

Given the spacetime $(M,g)$, in order to use the result of the previous section, we need to decide a convenient choice for $\le$. In causality theory
there are usually several interesting reflexive relations, which under some conditions get identified or become closed in the manifold topology, e.g.\ $J,D,K, J_S$, but just one natural open relation
 - the chronological relation $I$. Since  the basis for the Scott topology is given by the sets of the form $\twoheaduparrow x$, it is natural to identify $\twoheaduparrow x$ with $I^+(x)$ and so the way-below relation with the chronological relation:
\begin{equation}
x\ll y \Leftrightarrow (x,y)\in I.
\end{equation}
Observe that with such identification $\ll$ is irreflexive which is a property not necessarily shared by the way-below relation, in general.

Having made this choice we look for the choice of $\le$ that induces it.
Since we want to work with posets that are continuous, Prop.\ \ref{cnzz} imposes
\begin{equation} \label{cbbb}
x\le y \Leftrightarrow  I^-(x)\subset I^-(y)
\end{equation}
This is the relation $D_p$ we introduced in \cite{minguzzi07b}\cite[Sec.\ 4.1]{minguzzi18b}
\[
D_p=\{(x,y):  x \in \overline{J^-(y)}\}=\{(x,y):  I^-(x)\subset I^-(y)\}
\]
The relation $D_p$ is reflexive and transitive and its antisymmetry  is equivalent to past-distinction \cite[Thm.\ 3.3]{minguzzi07b}. Moreover, it satisfies  $D_p\circ I\subset I$ and in fact it is the largest relation that satisfies this property \cite[Lemmas 2.7, 2.8]{minguzzi07e}. Observe that under past-distinction $(M,D_p,I)$ is not a causal space in the sense of Kronheimer and Penrose because $I \circ D_p\not\subset I$, in general. Unfortunately, this contradicts the property (ii) of Theorem \ref{kctt}.

Interestingly, recalling the relation
\[
D_f=\{(x,y):  y \in \overline{J^+(x)}\}=\{(x,y):  I^+(x)\supset I^+(y)\}
\]
and $D:=D_p\cap D_f$ we have that $D$ is reflexive, transitive and it is antisymmetric under weak distinction \cite[Sec.\ 4.1]{minguzzi18b}. Moreover,
under weak distinction $(M,D,I)$ is a Kronheimer-Penrose causal space as $D\circ I \cup I \circ D\subset I$, and in fact $D$ is the largest relation  with this property \cite[Lemmas 2.7, 2.8]{minguzzi07e}.

The most natural solution of the previous difficulty with $D_p$ seems to be demanding the condition $D=D_p$, namely {\em future reflectivity} \cite[Def.\ 4.6]{minguzzi18b}. It admits several equivalent characterizations, the simplest being
\[
I^-(p) \subset I^-(q) \Rightarrow I^+(q) \subset I^+(p).
\]
Another interesting form is $D_f=\bar J$.
In general $D \subset D_p$ and so past-distinction implies weak distinction. Under future reflectivity, the two conditions coincide.


Let us investigate whether these heuristic conditions are indeed necessary.
\begin{proposition} \label{oovt}
The condition $I\circ D_p\subset I$ is equivalent to future reflectivity.
\end{proposition}

\begin{proof}
We already mentioned that future reflectivity implies $D=D_p$ and so  $I\circ D_p\subset I$. Suppose $I\circ D_p\subset I$. Let $(x,y)\in D_p$ namely $I^-(x)\subset I^-(y)$, and let $z\in I^+(y)$, then $z\in I^+(x)$. If $w\in I^+(y)$ we can take $z$ so close to $y$ that $z\in I^-(w)$ thus $w\in I^+(x)$. As $w$ is arbitrary $I^+(y)\subset I^+(x)$. We proved that  $I\circ D_p\subset I$ implies, for every pair $x,y$ with  $I^-(x)\subset I^-(y)$, that $I^+(y)\subset I^+(x)$, namely future reflectivity.
\qed \end{proof}

Identifying relations with their graphs we have
\begin{corollary} \label{ncgf}
On $(M,g)$ a continuous poset $(M,\le)$ with $\ll\,=I$ has necessarily $\le =D_p$, moreover, past-distinction and future reflectivity hold.
\end{corollary}

\begin{proof}
The necessity of $\le =D_p$ is given by Prop.\ \ref{cnzz} and Eq.\ (\ref{cbbb}). Past-distinction is the consequence of the antisymmetry of $\le$. Future reflectivity follows
from Prop.\  \ref{jsse}.
\qed \end{proof}

These are only necessary conditions. The  results of Sections \ref{cmqq}, \ref{kebb} are meant to prove that, under (SC), they are also sufficient.

Meanwhile let us observe that these spacetimes are not badly behaved. We recall that $t:M\to \mathbb{R}$ is a {\em semi-time function} if it is continuous  and $(x,y)\in I \Rightarrow t(x)<t(y)$. By continuity and $\bar I=\bar J$, a semi-time function satisfies also  $(x,y)\in \bar J \Rightarrow t(x)\le t(y)$. A {\em time function} is a continuous function $t: M \to \mathbb{R}$ such that $(x,y)\in J\backslash \Delta \Rightarrow t(x)<t(y)$.

\begin{theorem} \label{mcsd}
Let $(M,g)$ be past-distinguishing and future reflecting, then it admits a semi-time function. If it is distinguishing then it admits a time function.
\end{theorem}

\begin{proof}
It follows from the results ``A chronological, future or past reflective spacetime admits a semi-time function'' and  ``A distinguishing, future or past reflective spacetime admits a time function'' proved in \cite[Thm.\ 5]{minguzzi09c} and discussion thereafter.
\qed \end{proof}

\subsection{Suprema of directed sets and timelike curves} \label{cmqq}

\begin{proposition} \label{mmtr}
On the past-distinguishing spacetime $(M,g)$, consider the poset $(M,D_p)$.
Let $S$ be a subset, then its upper bounds $u$ are characterized by the inclusion $S\subset D_p^{-}(u)$, or equivalently, by the inclusion $I^-(S)\subset I^-(u)$. Thus, the upper bounds of $S$ coincide with those of $I^-(S)$,  $\sup S$ exists iff $\sup I^-(S)$ exists and in this case
 \begin{equation}
\sup S=\sup  I^-(S).
\end{equation}
\end{proposition}

The result implies that any supremum is the supremum of a chronologically past open set.
\begin{proof}
If $u$ is an upper bound for $S$, we have for every $s\in S$, $(s,u)\in D_p$ or $s\in D_p^{-}(u)$, thus $S\subset D_p^{-}(u)$. Thus for every $s\in S$, $I^-(s)\subset I^-(u)$ hence $I^-(S)\subset I^-(u)$. Conversely, if $u$ is such that $I^-(S)\subset I^-(u)$ then for every $s\in S$,  $I^-(s)\subset I^-(u)$, which reads  $s\in D_p^{-}(u)$ and so $S\subset D_p^{-}(u)$.

The last statement follows from the fact that the upper bounds for $S$ coincide with those of $I^-(S)$.
\qed \end{proof}

\begin{proposition} \label{jjit}
On the past-distinguishing spacetime $(M,g)$, consider the poset $(M,D_p)$.
Let $S$ be a directed subset, then $I^-(S)$ is directed and there is a countable chain $C$ such that $I^-(S)=I^-(C)\supset C$.
\end{proposition}

\begin{proof}
Let $x,y\in I^-(S)$, then there are $p,q\in S$, such that $x\in I^-(p)$, $y\in I^-(q)$. Since $S$ is directed, there is $r\in S$ such that $I^-(p),I^-(q)\subset I^-(r)$, thus $x, y \in I^-(r)$, or $r\in I^+(x)\cap I^+(y)$. As $I$ is open in the manifold topology, there is $z\in I^+(x)\cap I^+(y)\cap I^-(r)$, thus $z\in I^{-}(r)\subset I^-(S)$, and $I^-(x)\subset I^-(z)$, $I^-(y)\subset I^-(z)$, that is $(x,z)\in D_p$, $(y,z)\in D_p$, which proves that $I^-(S)$ is directed.

The manifold topology is separable, so let $\{x_k\}$ be a countable dense subset of $I^-(S)$. Set $z_1=x_1$, then, using directedness of $I^-(S)$ choose $z_{i+1}\in I^-(S)$ so that $I^-(z_{i+1})\supset I^-(z_i)\cup I^-(x_i)$, then $C:=\{z_i\}\subset I^-(S)$ is a chain.
Let $x\in I^-(S)$, as $I^-(S)$ is open, $I^+(x) \cap I^-(S)$ is non-empty so it contains some $x_i$, hence $x\in I^{-}(x_i)\subset I^-(z_{i+1}) \subset I^-(C)$.
\qed \end{proof}


\begin{definition} \label{def:sc}
A spacetime $(M,g)$ satisfies the \emph{supremum-compatibility condition} {\em (SC)} if for
every $D_p$-directed subset $S\subset M$ admitting a supremum one has
\begin{equation}
I^-(S)=I^-(\sup S).
\end{equation}
\end{definition}
The condition, expressing the Scott continuity of the map $x\mapsto I^-(x)$, was given this name in the recent \cite[Sec.\ 3.2]{gigli26} so we adopt  their terminology. Observe that the direction $\subset$ is trivial by the definition of supremum.

We shall also denote with (IC) the dual condition: for  every $D_f$-filtered subset $F\subset M$ admitting an infimum one has $I^+(F)=I^+(\inf F)$.

\begin{proposition} \label{vpoi}
On the past-distinguishing spacetime $(M,g)$, consider the poset $(M,D_p)$.
Let $S$ be a directed subset, then $I^-(S)$ is directed and there is a timelike curve $\gamma$ such that $I^-(S)=I^-(\gamma)\supset \gamma$.
The upper bounds of $S$ are given by the points in the set\footnote{The interior of this set is the so called {\em common future} of $\gamma$ (resp.\ $I^-(S)$),  a concept introduced by Hawking and Sachs \cite{hawking74}\cite[Sec.\ 4.1]{minguzzi18b}.}
\[
US:=\cap_{y\in\gamma} I^+(y)=\cap_{y\in I^-(S)} I^+(y).
\]
If the timelike curve $\gamma$ has a future endpoint, which we denote $p$,   $\sup S$ exists,  $\sup S=p$ and
\begin{equation}
I^-(S)=I^-(\sup S)=I^-(p).
\end{equation}
Under the further causality conditions {\em strong causality} or {\em past reflectivity} we have a converse:  if $\sup S$ exists then $\gamma$ has a future endpoint and so the previous identities apply, in particular (SC) holds.

Finally, assume $I^-(S)=I^-(\eta)$ for some timelike curve $\eta$. If
 $\eta$ has a future endpoint then $S$ has a supremum, which we denote $q$, they coincide and $I^-(S)=I^-(q)$. Under the further causality conditions {\em strong causality} or {\em past reflectivity} we have a converse:  if $S$ has supremum then $\eta$ has endpoint and the previous identity applies.
\end{proposition}

\begin{proof}
Let $C=\{z_i\}$ be the countable chain of Prop.\ \ref{jjit}.
For every $i$ let $z_{ji}\in I^-(z_i)$ be points such that $z_{j+1,i} \in I^+(z_{ji})$, $\lim_{j\to \infty} z_{ji} =z_i$. Let $y_1=z_{11}$. Observe that $I^-(z_{i+1})$ contains all $z_{ji}$ as it contains $I^-(z_i)$ for every $i$. We choose $y_{i+1}=z_{r(i),i+1}$ where $r(i)$ is chosen so large that $I^-(y_{i+1})$ contains $y_i$ and $z_{jk}$ for $j\le i$, $k\le i$. Note that $y_{i+1}$ exists because $I^+(y_i)\cap \big(\cap_{j, k\le i} I^+(z_{jk})\big)$ is an open neighborhood of $z_{i+1}$. This gives a sequence of chronologically related points $y_i\in I^-(y_{i+1})$ for all $i$,
hence a chain $C_2$ such that $I^-(C_2)$ includes all points $z_{ji}$, as $z_{ji}\in I^-(y_{m+1})$ for $m=\max(j,i)$.

Clearly, $C_2 \subset I^-(C)$ so $I^-(C_2)\subset I^-(S)$. Conversely, if $p\in I^-(C)$, $p\in I^-(z_{i})$ for some $z_{i}\in C$, thus $z_i\in I^+(p)$ so we can find $j$ such that $z_{ji}\in I^+(p)$, thus $p\in I^-(z_{ji})\subset I^-(C_2)$. Thus $C_2$ is timelike and $C_2\subset I^-(C_2)=I^-(C)=I^-(S)$.
The points of $C_2$ can now be connected by a timelike curve $\gamma$ so that $\gamma \subset I^-(C_2)=I^-(\gamma)$.

Observe that $I^-(\gamma)=I^-(S)$ and the upper bounds of the latter set coincide with those of $S$. The condition that $u$ is an upper bound for $I^-(\gamma)$  is $I^-(u)\supset I^-(\gamma)$ or equivalently $I^-(u)\supset \gamma$, which is equivalent to $u\in \cap_{y\in\gamma} I^+(y)$. Note that if $z\in I^-(w)$, $w\in \gamma$ then $I^+(z)\supset I^+(w)$, so  $u\in \cap_{y\in I^-(\gamma)} I^+(y)=\cap_{y\in \gamma} I^+(y)$.

Suppose that $\gamma$ has a future endpoint $p$. First, $p$ is a ($D_p$-)upper bound of $S$ because $I^-(S)=I^-(\gamma)\subset I^-(p)$ (the last inclusion is less trivial than it might seem and is related to the properties of convex neighborhoods \cite[Thm.\ 6]{minguzzi13d}). Note that $p$ cannot belong to the image of $\gamma$ otherwise by $p\in \gamma \subset I^-(\gamma)\subset I^-(p)$ there would be a violation of chronology.
Let $q$ be an upper bound for $S$, so for $I^-(S)$ and so for $\gamma$. Let $y_k$ be a sequence on $\gamma$, $y_k\in I^-(y_{k+1})$, $y_k\to p$. We have $y_k \in I^-(\gamma) \subset I^-(q)$, hence $p\in \overline{I^-(q)}$, or equivalently, $I^-(p)\subset I^-(q)$, that is $(p,q)\in D_p$. This proves that $p=\sup S$.
We already know $I^-(\gamma)\subset  I^-(p)$.
Conversely, if $z\in I^-(p)$ then $I^+(z)$ is an open neighborhood of $p$, so it contains some $y_k$, $z\in I^-(y_k)\subset I^-(\gamma)$. So $I^-(p)=I^-(\gamma)$.

Observe that by \cite[Props. 105, 107]{minguzzi18b} past-distinction  and past reflectivity give stable and hence strong causality, so we can prove the statement just for {\em strong causality} as it implies that for {\em past reflectivity}.

Suppose $(M,g)$ is strongly causal, $\sup S$ exists and suppose, by contradiction, that $\gamma$ has no future endpoint. Let $y_k$ be a sequence on $\gamma$, $y_k\in I^-(y_{k+1})$, cofinal in the directed set $\gamma$. By past-distinction  $\gamma$ is not totally future imprisoned so $y_k$ can be chosen so that it admits no converging subsequence (this is used to claim inextendibility of $\sigma^q$ below).
Let $q=\sup S$, then it is an upper bound so $y_k\in I^-(y_{k+1})\subset I^-(q)$. Let $\sigma_k$ be a causal curve connecting $y_k$ to $q$. By the limit curve theorem there is a past inextendible continuous causal curve $\sigma^q$ ending at $q$, to which a subsequence, denote again $\sigma_k$, converges.  Thus for every $r\in \sigma^q\backslash\{q\}$, $r\in \overline{I^+(y_k)}$. If $r\in I^-(q)$ we can find $s\in I^-(q)\cap I^+(r)$, hence $y_k\in I^-(s)$ for every $k$, thus $\gamma \subset I^-(s)$ and $I^-(\gamma)\subset I^-(s)$. This would give that $s$ is an upper bound of $I^-(\gamma)$ and $(s,q)\in I\subset D_p$, but $q$ being a supremum gives $(q,s)\in D_p$ a contradiction with the anti-symmetry of $D_p$, i.e.\ with past-distinction. Thus $(r,q)\notin I$ (by the arbitrariness of $r$, $\sigma^q$ is an achronal null geodesic). Let $r'\in I^+(r)$,  then  $r\in I^-(r')$ which implies $y_k\in I^-(r')$ and so $\gamma\subset I^-(r')$, or equivalently $I^-(\gamma)\subset I^-(r')$. This proves that $r'$ is an upper bound. Since $\sup S$ exists $(q, r')\in D_p$, that is $I^-(q)\subset I^-(r')$. Let $q'_i\to q$, $q'_{i+1} \in I^+(q'_i)\cap I^-(q)$ and $q''_i\to q$, $q''_{i+1}\in I^-(q''_i)\cap I^+(q)$. For each $i$, $r\in I^-(q''_i)$, thus we can find some $r_i\in I^-(q''_i)\cap I^+(r)$. In fact we can choose it so that $r_i\to r$. But we have shown above that $I^-(r_i)\supset I^-(q)\ni q'_j$ for every $j$.
We conclude that strong causality is violated at $q$ as following a timelike curve it is possible to start from $q'_j\in I^-(q)$, arbitrarily close to $q$, pass through $r_i$ in an arbitrarily small neighborhood of $r$, and reach $q''_i \in  I^+(q)$, arbitrarily close to $q$.


Finally, observe that we started from $I^-(S)=I^-(C)$ to construct a timelike curve $\gamma$ with the property $I^-(S)=I^-(\gamma)$. If we already knew that $I^-(S)=I^-(\eta)$ then we could have developed the last part of the proof with $\eta$ in place of $\gamma$, which proves that if $\eta$ has endpoint then $S$ has supremum and they coincide, and under strong causality that if $S$ has supremum then $\eta$ has endpoint.
\qed \end{proof}

\subsection{The way-below relation of $D_p$ is the chronological relation}
\label{kebb}

Let us denote by
\begin{equation}
C^+(x):=\operatorname{Int}\bigcap_{r\in I^-(x)} I^+(r)
\end{equation}
the \emph{common future} of $x$ \cite{hawking74}\cite[Sec.\ 4.1]{minguzzi18b}. We also write $(x,y)\in C$ for $y\in C^+(x)$. Two facts from the review \cite{minguzzi18b} will be used: by \cite[Prop.\ 4.5]{minguzzi18b} one has $C^+(x)=\operatorname{Int}D_p^+(x)$, and by \cite[Def.\ 4.6]{minguzzi18b} the spacetime is future reflecting precisely when $C^+(x)=I^+(x)$ for every $x$. The first is immediate in our notation, since $I^-(I^-(x))=I^-(x)$ gives
\begin{equation} \label{comfut}
\bigcap_{r\in I^-(x)} I^+(r)=\{u\in M:\ I^-(x)\subset I^-(u)\}=D_p^+(x)=\,\uparrow x ,
\end{equation}
so that $C^+(x)=\operatorname{Int}\uparrow x$: the common future of $x$ is the interior of the principal upper set of $x$ in the poset $(M,D_p)$.

\begin{proposition}
The common future relation $C$ is transitive.
\end{proposition}

\begin{proof}
If $(x,y)\in C$ and $(y,z)\in C$, there is a neighborhood $V\ni z$ such that $V\subset D_p^+(y)$, so, since $C\subset D_p$ and $D_p$ is transitive $V\subset D_p^+(x)$ and so $z\in \textrm{Int} D_p^+(x)$, namely $(x,z)\in C$.
\end{proof}

The way-below relation is described by the common future, and it is condition {\em (SC)} that makes the description exact.

\begin{theorem} \label{comfutthm}
Let $(M,g)$ be a past-distinguishing spacetime and consider the poset $(M,D_p)$. Then
\[
\twoheaduparrow x\ \subset\ C^+(x)\qquad \text{for every } x\in M ,
\]
and if $(M,g)$ satisfies {\em (SC)} the inclusion is an equality, that is $\ll\,\,=C$.
\end{theorem}

\begin{proof}
\emph{The inclusion.} Let $x\ll y$. Let $\eta$ be a timelike curve with future endpoint $y$, $y\notin \eta$, and apply the definition of the way-below relation to the directed set $S:=\eta$: by the part of Prop.\ \ref{vpoi} concerning curves with a future endpoint, $\sup S=y$, and $(y,\sup S)\in \Delta\subset D_p$, so there is $s\in \eta$ with $(x,s)\in
D_p$. Since $s\in \eta$ and $y$ is the future endpoint of $\eta$, $I^+(s)$ is an open neighbourhood of $y$, and every $u\in I^+(s)$ satisfies $(s,u)\in I\subset D_p$, hence $(x,u)\in D_p$ by transitivity. Thus $I^+(s)\subset\, \uparrow x$ and $y\in \operatorname{Int}\uparrow x=C^+(x)$ by \eqref{comfut}.

\emph{The converse under {\em (SC)}.} Let $y\in C^+(x)$ and let $S$ be a directed subset with supremum $p$ such that $(y,p)\in D_p$. The latter reads $I^-(y)\subset I^-(p)$, that is $y\in \overline{I^-(p)}$; since $C^+(x)$ is an open neighbourhood of $y$, there is $w\in C^+(x)\cap I^-(p)$. From $w\in C^+(x)\subset \,D_p^+(x)$ we get $(x,w)\in D_p$, and from $w\in I^-(p)=I^-(S)$, which is {\em (SC)}, we get $(w,s)\in I\subset D_p$ for some $s\in S$. Transitivity gives $(x,s)\in D_p$. As $S$ was arbitrary, $x\ll y$.
\qed \end{proof}

The theorem splits the identification $\ll\,=I$ into its two independent halves.

\begin{corollary} \label{iiskk}
Let $(M,g)$ be a past-distinguishing spacetime and consider the poset $(M,D_p)$.
\begin{enumerate}
  \item[(i)] If $(M,g)$ satisfies {\em (SC)} then $I\subset\,\ll$.
  \item[(ii)] If $(M,g)$ is future-reflecting then $\ll\,\subset I$.
  \item[(iii)] Under {\em (SC)}, one has $\ll\,=I$ if and only if $(M,g)$ is future   reflecting.
  \item[(iv)] Under future reflectivity, one has $\ll\,=I$ if and only if {\em (SC)} holds.
\end{enumerate}
\end{corollary}

\begin{proof}
(i) The set $I^+(x)$ is open and contained in $D_p^+(x)$, since $I\subset D_p$, so $I^+(x)\subset C^+(x)=\,\twoheaduparrow x$ by \eqref{comfut} and Thm.\ \ref{comfutthm}.

(ii) By Thm.\ \ref{comfutthm} and \cite[Def.\ 4.6]{minguzzi18b}, $\twoheaduparrow x\subset C^+(x)=I^+(x)$.

(iii) By Thm.\ \ref{comfutthm} the identification $\ll\,=I$ reads $C^+(x)=I^+(x)$ for every $x$, which is future reflectivity by \cite[Def.\ 4.6]{minguzzi18b}.

(iv) Sufficiency of {\em (SC)} follows from (i) and (ii). For the necessity, let $S$ be directed with $p=\sup S$ and let $x\in I^-(p)$. From $I=\,\ll$ we can find $x\ll x'\ll p$; since $x'\ll p$, there is $s\in S$ with $(x',s)\in D_p$. From $(x,x')\in I$, $(x',s)\in D_p$ and $D_p \circ I \subset I$,  we get $(x,s)\in I$, that is $x\in I^-(S)$; the inclusion $I^-(S)\subset I^-(p)$ is immediate.
\qed \end{proof}

\begin{theorem} \label{contreg}
Let $(M,g)$ be a past-distinguishing spacetime satisfying {\em (SC)}. Then $(M,D_p)$ is a continuous poset and its way-below relation is the common-future relation,
\[
\ll\ =\ C,\qquad \text{that is}\qquad \twoheaduparrow x=C^+(x)=\operatorname{Int}D_p^+(x),
\]
and $I\subset\,\ll\,\subset D_p$.
\end{theorem}

\begin{proof}
\emph{$I^-(x)$ is directed with supremum $x$.} Let $y,z\in I^-(x)$; then $I^+(y)\cap I^+(z)$ is an open neighbourhood of $x$, so we can find $w\in I^-(x)\cap I^+(y)\cap I^+(z)$, and $(y,w),(z,w)\in I\subset D_p$ with $w\in I^-(x)$, which gives directedness. Let $\eta$ be a timelike curve with future endpoint $x$, $x\notin \eta$; then $I^-(x)=I^-(\eta)$, and by $I^-(I^-(x))=I^-(x)$ and Prop.\ \ref{vpoi} applied with $S=I^-(x)$ we get $\sup I^-(x)=\sup I^-(\eta)=x$.

\emph{Continuity.} By Cor.\ \ref{iiskk}(i), which uses {\em (SC)} and no reflectivity, we have $I\subset\,\ll$, that is $I^-(x)\subset\, \twoheaddownarrow x$ for every $x$. Proposition \ref{wwgf} applied with $D_x=I^-(x)$ gives that $(M,D_p)$ is continuous.

\emph{Identification of $\ll$.} This is Thm.\ \ref{comfutthm}, whose second half uses {\em (SC)} alone. The inclusions $I\subset\,\ll\,\subset D_p$ follow from \eqref{comfut} and Lemma \ref{comfut}, $I^+(x)$ being open and contained in $D_p^+(x)$.
\qed \end{proof}

Future reflectivity is thus not needed for the poset to be continuous: it is needed only in order to identify the way-below relation with the chronological relation, and by \cite[Def.\ 4.6]{minguzzi18b} it is precisely the condition under which the common future of a point reduces to its chronological future. In fact $C=I\circ D_p$, and 
the way-below relation $C$ is  a fixed point of the map $R\mapsto R\circ D_p$, see Thm.\ \ref{fixedkp} and Cor.\ \ref{wbfixed}.

\begin{theorem} \label{kssw}
Let $(M,g)$ be a past-distinguishing future-reflecting spacetime satisfying {\em (SC)}. Then $(M,D_p)$ is a continuous poset and $I$ is the way-below relation of $D_p$. Conversely, if $(M,\le)$ is a continuous poset with $\ll\,=I$, then necessarily $\le\,=D_p$ and $(M,g)$ is
past-distinguishing, future-reflecting and satisfies {\em (SC)}. Thus the three conditions are necessary and sufficient.
\end{theorem}

Under future reflectivity $D=D_p$, thus in the above theorem we can replace $D_p$ with $D$.

\begin{proof}
For the first statement, by Thm.\ \ref{contreg} the poset $(M,D_p)$ is continuous with $\ll\,=C$, and by \cite[Def.\ 4.6]{minguzzi18b} future reflectivity gives $C^+(x)=I^+(x)$ for every $x$, that is $\ll\,=I$.

The necessity of $\le\,=D_p$, of past-distinction and of future reflectivity is Cor.\ \ref{ncgf}, which does not use {\em (SC)}; the necessity of {\em (SC)} is the last statement of Cor.\ \ref{iiskk}.
\qed \end{proof}
%
%
%
%
%

\begin{theorem} \label{mmwr}
On the spacetime $(M,g)$ consider $(M,D_p)$. The property (SC) is implied by any of the following causality conditions (all implying past-distinction, hence  $(M,D_p)$ is a poset):
\begin{itemize}
\item[(a)] strong causality,
\item[(b)] past-distinction and past reflectivity,
\item[(c)] distinction and future reflectivity.
\end{itemize}
\end{theorem}

None of these three conditions is necessary: the dual of Malament's example \cite[Fig.\ 15]{minguzzi18b} is past-distinguishing, future-reflecting, neither distinguishing nor past-reflecting, hence not strongly causal, and yet it satisfies {\em (SC)}.

The property of {\em chronocompleteness} \cite{braun26,ohanyan26} ``every chronologically increasing sequence $x_n$ in $M$, $x_n \ll x_{n+1}$, which is upper bounded, in the sense that  there is $x^+\in M$  with $x_n \ll x^+$ for every $n \in \mathbb{N}$, converges to some point of $M$''  and the variant ``there is no pair $(\gamma, p)$, where $\gamma$ is a future-inextendible timelike curve and $p\in M$,  such that $I^-(\gamma)\subset I^-(p)$'' are equivalent to global hyperbolicity as proved in \cite{minguzzi26c} and so imply (SC) via (a).

\begin{proof}
The statements with (a) or (b) are proved  in Prop.\ \ref{vpoi}. The statement with (c) follows again from \cite[Props. 105, 107]{minguzzi18b}  as future distinction  and future reflectivity give stable and hence strong causality.
\end{proof}

\begin{corollary}
In a distinguishing future reflecting spacetime $(M,g)$ the poset $(M,D_p)$ is continuous and $I=\,\,\ll$.
\end{corollary}

A natural question is whether (SC) can be dropped from the characterization of continuous posets of Thm.\ \ref{kssw}. The answer is negative as the next example shows.

\begin{example} \label{suex}
Let $M=\mathbb{R}_t\times S^1_x\times  \mathbb{R}_y\backslash (A\cup B)$ where $S^1$ is identified with $[0,4]$ with the endpoints identified (Fig.\ \ref{mcp}).  Consider the metric $g= \dd x \dd t + f(t) \dd x^2 + \dd y^2$ where $f\ge 0$ is smooth, equal to zero for $t\le 0$ and positive with $f'>0$ for $t>0$. The time orientation is given by the causal vector $\p_t$.
We set $A=\{t\le 0, \ x=1\}$ and $B=\{t\le h(y)\le 0, \ x=3\}$ where $h(y)=-a y^2-b y^4$, $a,b>0$. Observe $t$ is a semi-time function and the metric is locally Minkowski for $t\le 0$. The locus $t=0,y=0$ is really the union of two achronal lightlike geodesic segments $\eta_1,\eta_2$, one for $1<x<3$ and the other on the complement to $1 \le x \le 3$.

The constant $a$ is chosen in such a way that the graph of $h$ osculates the intersection with $x=3$ of the past light cone from $q=(0,2,0)$. In fact, since this intersection is known to be a parabola the graph of $-ay^2$ coincides with such intersection. The set $B$ is called the {\em filter} and comes from an idea introduced in \cite{minguzzi07e}. Basically, past inextendible timelike curves ending at $q$ can pass through the filter near the achronal lightlike geodesic $\eta_2$, hence having points in $x>3$, thanks to $b>0$. However, similar curves ending at $q'=(0,2+\epsilon,0)$ with $\epsilon\in (0,1)$ do not cross the filter passing arbitrarily close to $t=0$ or $\eta_2$.

\begin{figure}[ht]
\centering
\includegraphics[width=11cm]{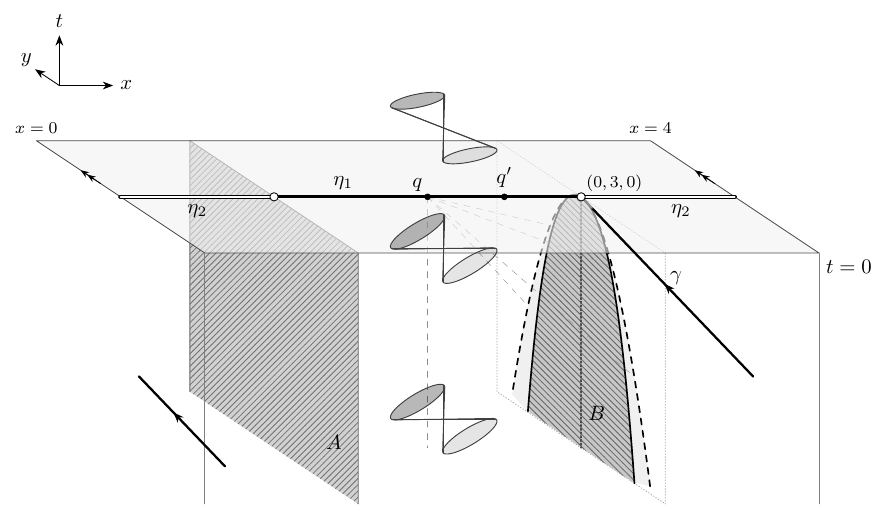}
\caption{A past-distinguishing future reflecting  spacetime in which (SC) does not hold. The hypersurface $t=0$ is null. The future inextendible timelike curve $\gamma$ has $D_p$-supremum $q$. All the points in $\eta_1$ and $\eta_2$ have the same chronological future so future-distinction does not hold. The dotted line is the graph of $-ay^2$ and coincides with the intersection of $x=3$ with $\p J^-(q)$, while the boundary of the filter $B$ is the graph of $h=-ay^2-by^4$.}\label{mcp}
\end{figure}

Now, let $\gamma$ be a future inextendible timelike geodesic approaching  the removed point $(0,3,0)$,  necessarily distinct from the lightlike geodesic $\eta_2$ or the vertical lightlike geodesic $t<0, x=3, y=0$ and approaching the endpoint from the region to the right of $B$. It is a straight line in the  region $t\le 0$. We have $I^-(q)\supset I^-(\gamma)$ but the same does not hold for $q'$ as above. Let $S=\gamma$ (or $S=I^-(\gamma)$) then $\sup S=q$. The spacetime is past-distinguishing (the function $f$ serves to make the pasts of events in $t>0$ distinct) future reflecting but (SC) does not hold (it is also partial imprisoning). Also all points in $\eta_1$ and $\eta_2$ have the same  chronological future thus future distinction does not hold. This shows that strong causality does not hold either and so past reflectivity cannot hold by \cite[Props. 105, 107]{minguzzi18b}.
\end{example}

The following statement is the time dual of Thm.\ \ref{kssw}, obtained by applying it to the opposite time orientation.


\begin{theorem} \label{ksswdual}
Let $(M,g)$ be a future-distinguishing past-reflecting spacetime satisfying {\em (IC)}. Then $(M,D_f)$ is a co-continuous poset and $I$ is the way-above relation of $D_f$, that is $\overline{\ll}\,=I$ as sets of ordered pairs. Conversely, if $(M,\le)$ is a co-continuous poset with $\overline{\ll}\,=I$, then necessarily $\le\,=D_f$ and $(M,g)$ is future-distinguishing, past-reflecting and satisfies {\em (IC)}. Thus the three conditions are necessary and sufficient.
\end{theorem}

\begin{remark}
Properties (SC) and (IC) do not by themselves find  a place in the causal ladder of spacetimes. However, let us define ``improved past-distinction'' as ``past-distinction and (SC)'' and ``improved future-distinction'' as  ``future-distinction and (IC)'' and ``improved distinction'' as ``improved past-distinction and improved future-distinction'', then
by Thm.\ \ref{mmwr}(a) and the standard causal ladder \cite{minguzzi18b}
\begin{quote}
strong causality $\Rightarrow$ improved distinction $\Rightarrow$ improved past-distinction or improved future-distinction $\Rightarrow$ weak distinction  $\Rightarrow$ non-total imprisonment
\end{quote}
thus the improved properties can replace the old ones in the causal ladder. It remains to check if the levels really differ. They do and the examples are the usual ones because the difficulty is finding examples where (SC) or (IC) are not satisfied rather than the opposite.
\begin{example}
Improved distinction, non-strongly causal: the spacetime $\mathbb{R}_t\times S^1_x\backslash (A\cup B)$ where $S^1$ is $[0,3]$ with endpoints identified, $g=\dd t \dd x$, $A=\{t\le 0, x=1\}$, $B=\{ t\ge 0, x=2\}$.

Improved past-distinction, non-improved distinction:  the spacetime $\mathbb{R}_t\times S^1_x\backslash A$ where $S^1$ is $[0,3]$ with endpoints identified, $g=\dd t \dd x+ f(t) \dd x^2$ $A=\{t\le 0, x=1\}$, where $f\ge 0$ is smooth, equal to zero for $t\le 0$ and positive with $f'>0$ for $t>0$. Alternatively, the spacetime of Example \ref{suex} without removing $B$.

Weak distinction, non-improved past-distinction: Fig.\ 13(b) of \cite{minguzzi18b}.
\end{example}
Non-partial imprisonment, which is implied by strong causality, does not find a place in this variant of the causal ladder due to the following example.
\begin{example}
In Example \ref{suex} remove also the sets $C=\{t\ge 0, x=1/2\}$ and $D=\{t\le 0, x=3.5\}$. This spacetime is non-strongly causal non-partial imprisoning, distinguishing but (SC) does not hold, so it is not improved distinguishing. Also it is not past reflecting nor future reflecting.
\end{example}
\end{remark}
\begin{proposition} \label{kffw}
Let $(M,g)$ be a past-distinguishing future-reflecting spacetime which satisfies (SC), and let us consider the continuous poset $(M,D_p)$. The Scott topology is that generated by the future sets $U$, i.e.\ the (necessarily manifold-open) sets $U$ such that $I^+(U)=U$.
\end{proposition}

\begin{proof}
It follows from Thm.\ \ref{nncf}.
\qed \end{proof}

\begin{remark}
As previously mentioned, the dual of \cite[Fig.\ 15]{minguzzi18b} is a spacetime which satisfies ``past-dis\-tin\-guish\-ing and future-reflecting and (SC)'' and that is neither distinguishing nor past-reflecting. This is essentially the dual of Malament's example \cite{malament77,malament12}; hence, in such spacetimes there can exist maps $f: M \to M$ that are chronological bijections but not homeomorphisms (with respect to the manifold topology). In this example we can verify that the manifold topology is strictly finer than the Alexandrov topology (not to be confused with the biScott topology, as we will show that they coincide only under stronger conditions; see Thm.\ \ref{dzxd}). Moreover, the Lawson topology differs from both the manifold and Alexandrov topologies: indeed, there exist open sets in the Lawson topology that are open in neither the manifold topology nor the Alexandrov topology.
\end{remark}

\begin{corollary}
Let $(M,g)$ be a   spacetime such that the past causal cones $J^-(x)$ are closed and (SC) holds, then $(M,J)$ is a continuous poset and $I$ is the way-below relation of $J$. The  imposed causality conditions are necessary if we want the continuity property for the poset and the identifications $\le\,=J$ and $\ll\,=I$.
\end{corollary}
The proof shows that $J=D_p$ so if these equivalent conditions hold the spacetime is past-distinguishing. The condition (SC) seems necessary in the first statement for it seems that Example \ref{suex} can be modified tilting the cones in the region $t<0$ and removing the causal past of $B$ so as to get the closure of $J^-(x)$ for every $x$.
\begin{proof}
The closure of the sets $J^-(x)$ implies future reflectivity \cite[Remark 4.14]{minguzzi18b} and also $D_p=J$ thus causality implies past-distinction.
We satisfy the conditions of Thm.\ \ref{kssw}, so $(M,J)$ is a continuous poset and $I$ is the way-below relation of $J$.

Conversely, suppose that $(M,J)$ is a  continuous poset and $I$ is the way-below relation of $J$. By the second statement of  Thm.\ \ref{kssw} $(M,g)$ is a past-distinguishing future-reflecting spacetime that satisfies (SC) and $J=D_p$ which implies that the sets $J^-(x)$, $x\in M$, are closed.
\qed \end{proof}

\begin{theorem}
Let $(M,g)$ be a  spacetime which is future reflecting. Then $(M,g)$ is future distinguishing if and only if, for all
$x,y\in M$,
\[
(x,y)\in D_p \ \text{ and } \ x\in \overline{J^+(y)} \ \Longrightarrow \ x=y .
\]
\end{theorem}

\begin{proof}
Recall that $D_p=\{(x,y): x\in \overline{I^-(y)}\}$ and
$D_f=\{(x,y): y\in \overline{I^+(x)}\}$ and that future reflectivity is equivalent to $D_p\subset D_f$, namely   $I^-(x)\subset I^-(y) \Rightarrow I^+(y)\subset I^+(x)$.

($\Leftarrow$) Let $x,z\in M$ be such that $(x,z)\in J^+$ and
$x\in\overline{J^+(z)}$. Since $J^+\subset D_p$, we have $(x,z)\in D_p$,
hence by assumption $x=z$. By \cite[Thm.\ 4.44]{minguzzi18b}, $(M,g)$ is
future distinguishing.

($\Rightarrow$) Let $(M,g)$ be future distinguishing and let
$x,y\in M$ be such that $(x,y)\in D_p$ and $x\in \overline{J^+(y)}$.
By future reflectivity, $(x,y)\in D_f$, that is,
$I^+(y)\subset I^+(x)$. On the other hand,
$x\in\overline{J^+(y)}=\overline{I^+(y)}$ means $(y,x)\in D_f$, that is,
$I^+(x)\subset I^+(y)$. Therefore $I^+(x)=I^+(y)$, and future
distinction gives $x=y$.
\end{proof}

\begin{corollary} \label{bcot}
Let $(M,g)$ be a spacetime such that $(M,D_p)$ is a continuous poset, hence $\ll\, = I$. It is distinguishing if and only if ``$x\le  y$ and $\twoheaduparrow x =\,\twoheaduparrow y \Rightarrow x=y$''.
\end{corollary}

\subsection{Bicontinuity, causal continuity and topologies}

The property (SC) disappears when discussing stronger abstract properties such as (weak) bicontinuity.

\begin{theorem} \label{bxxd}
Let $(M,g)$ be a causally continuous spacetime, then $(M,D)$ is a bicontinuous poset and $I$ is the way-below relation of $D$. The choice $\le\,=D$ and the imposed causality condition are necessary if we want the bicontinuity property for the poset and the identification of  the chronological relation with the way-below (equiv.\ way-above) relation.
\end{theorem}

Under causal continuity $D=D_p=D_f=\bar J=K=J_S$ (due to the characterization $D=\bar J$ of reflectivity \cite[Def.\ 4.9(f)]{minguzzi18b}, the transverse ladder \cite[Thm.\ 4.13]{minguzzi18b} ``reflectivity implies transitivity of $\bar J$'', and the identity $K=J_S$ under stable causality  \cite[Thm.\ 4.99]{minguzzi18b}, $J_S$ being the Seifert relation), thus in the previous statement we can replace $D$ with $D_p$, $D_f$, $\bar J$, $K$ or $J_S$.

Note that asking weak bicontinuity and equality of way-below and way-above relations with $I$ immediately promotes the assumed property to bicontinuity as way-below and way-above relations coincide.

\begin{proof}
Let us prove that the proposed structure is a bicontinuous poset. Every causally continuous spacetime is distinguishing and reflecting so (SC) holds by Thm.\ \ref{mmwr}(c). The former property means antisymmetry of $D_p$ and $D_f$, hence of $D=D_f\cap D_p$. The latter property implies $D=D_p=D_f$. We already know by Thm.\ \ref{kssw} that $(M,D_p)$ is a continuous poset with way-below relation $I$, so $(M,D)$ is a continuous poset with way-below relation $I$.


But distinguishing and reflecting are invariant under time duality so (IC) holds by the dual of  Thm.\ \ref{mmwr}(c). Thus the dual of  Thm.\ \ref{kssw}, Thm.\ \ref{ksswdual}, gives that the
 way-above relation for $D_f$ is $I$, that $I^+(x)$ is filtered for every $x\in M$ and $\inf I^+(x)=x$. Since $D=D_f$ this is precisely the property of bicontinuity.

To prove necessity, since a bicontinuous poset is a continuous poset, we have by Thm.\ \ref{kssw} that $\le \,=D_p$, that future reflectivity holds and that $D_p$ is antisymmetric. Similarly, by co-continuity, Thm.\ \ref{ksswdual} gives  $\le \,=D_f$,  that past reflectivity holds and that $D_f$ is antisymmetric.

We have just proved $D_f=D_p$ and so $D=D_f=D_p$.
Since we have proved both reflectivity and antisymmetry of $D$ (weak distinction) the spacetime is causally continuous \cite[Def.\ 4.109]{minguzzi18b}.
\qed \end{proof}

\begin{theorem} \label{dzxd}
Let $(M,g)$ be a causally continuous spacetime, and consider the bicontinuous poset  $(M,D)$. The interval (biScott)  topology coincides with the Alexandrov topology of relativity theory and  with the manifold topology.
\end{theorem}

\begin{proof}
By Remark \ref{reqt} on a bicontinuous poset the biScott topology is that generated by the basis $\twoheaduparrow a \, \cap \twoheaddownarrow b=I^+(a)\cap I^-(b)$, thus it   coincides with the so-called Alexandrov topology of relativity theory. Since a causally continuous spacetime is strongly causal, the Alexandrov topology coincides with the manifold topology \cite{penrose72,minguzzi18b}.
\qed \end{proof}

The following recovers, with new proofs, the answer to a question by Martin and Panangaden on the possible coincidence between causal simplicity and bicontinuity (which they expressed in the conclusions of their work \cite{martin06}). Observe that they refer to their more restrictive framework in which they always impose $\le =J$.  They could only prove under continuity \cite[Thm.\ 5.2]{martin06} that causal simplicity is equivalent to the Lawson topology being included in the interval topology (a consequence of bicontinuity, Prop.\ \ref{pldf}), see also  \cite[Lemma 4.1]{sharifzadeh19}, thus proving the direction ``bicontinuity implies causal simplicity''. The converse direction, and hence the equivalence, was  established by Ebrahimi \cite{ebrahimi14}. We provide novel proofs of both directions.


\begin{corollary} \label{cssr}
Let $(M,g)$ be a causally simple spacetime, then $(M,J)$ is a bicontinuous poset and $I$ is the way-below relation of $J$. The imposed causality condition is necessary if we want the bicontinuity property for the poset and the identifications $\le\, =J$ and $\ll\,=I$.
\end{corollary}

\begin{proof}
Let  $(M,g)$ be a causally simple spacetime, then it is causally continuous \cite{minguzzi18b} and so, by the first statement of Thm.\ \ref{bxxd},  $(M,D)$ is a bicontinuous poset and $I$ is the way-below relation of $D$. Causal simplicity implies the closure of the causal relation and so the identity $D=J$. Thus $(M,J)$ is a bicontinuous poset and $I$ is the way-below relation of $J$.

Conversely, suppose that $(M,J)$ is a bicontinuous poset such that $\ll\,=I$. By the second statement of Thm.\ \ref{bxxd}, we have causal continuity of $(M,g)$ and $J=D$. But the latter condition expresses the closure of the causal relation (recall \cite[Thm.\ 4.12]{minguzzi18b}) thus $(M,g)$ is causally simple.
\qed \end{proof}

A corollary of Theorem \ref{dzxd} is

\begin{corollary} \label{daxd}
Let $(M,g)$ be a causally simple spacetime, and consider the bicontinuous poset  $(M,J)$. The interval (biScott)  topology coincides with the Alexandrov topology of relativity theory and  with the manifold topology.
\end{corollary}

%
%
%

It remains to identify the Lawson topology.





\begin{theorem} \label{kxqp}
If $(M,g)$ is causally continuous then the bicontinuous poset $(M,D)$ has  Lawson topology coincident with the manifold topology (and hence with the interval  and Alexandrov topologies, cf.\ Thm.\ \ref{dzxd}). Moreover, for $(M,D)$ density holds.
\end{theorem}

We recall that under causal continuity $D=D_p=D_f=\bar J=K=J_S$.

As we shall see from the proof, density is a consequence of the limit curve theorem.

\begin{proof}[Proof not using Sec.\ \ref{coinc} but using the exponential map]
Let us prove that the Lawson topology for $(M,D_p)$ is finer than the manifold topology (we note that this direction does not use past reflectivity).
Since $(M,g)$ is strongly causal \cite[Thm.\ 4.107]{minguzzi18b} the point $x\in M$ admits a causally convex open neighborhood $O$ with compact closure contained in a convex neighborhood $U$. Moreover, $O$ and $U$ can be found arbitrarily small, namely contained in any chosen open neighborhood $O'\ni x$. By causal convexity $J_{O}=J\cap (O\times O)$, where $J_O$ is the causality relation of the spacetime $O$ with the restricted metric. But, again by causal convexity, $J_O=J_U\cap (O\times O)$, where $J_U$ is closed \cite[Thm.\ 2.11]{minguzzi18b}, thus $J_O$ is closed in the spacetime $(O, g\vert_O)$ and the same holds for $J\cap (O\times O)$. If $a,b\in O$, $(a,b)\in D_p$ means $a\in \overline{J^-(b)}=J^-(b)$, hence $D_p\cap (O\times O)=J\cap (O\times O)$.

We want to prove that $x$ admits a Lawson neighborhood included in $O'$. Let us consider the future causal cone $C_x\subset T_xM$. Let $\{v^0, v^1, \ldots, v^n\}$ be linear coordinates on $T_xM$ so that $v^0=1$ is a disk section $D_x$ of $C_x$ such that $\exp_x D_x\subset O$ (otherwise rescale $v^0$). Consider the compact set $A=\exp_x D_x$. Every point of $\exp_x(\{v^0\le 1\}\cap C_x)$ lies on a radial causal geodesic from $x$ to a point of $A$, thus by causal convexity $\exp_x(\{v^0\le 1\}\cap C_x) \subset O$; in particular, every causal curve issued from $x$ meets $A$ before escaping $O$. Since $J\cap (O\times O)$ is closed, for every $a\in A$ we can find $r_a\in I_O^-(a)$ such that $x\notin J_O^+(r_a)$ (otherwise $(a,x)\in J_O$, contradicting the antisymmetry of $J$). The compact set $A$ can be covered by a finite number of neighborhoods $I^+_O(r_{a_i})$, $i=1,\ldots,k$. Let $z\in I^-_O (x)$ be a point chosen so close to $x$ that every causal curve issued from $z$ intersects $\cup_i I^+_O(r_{a_i})$ before escaping $O$ (this follows from a limit curve argument, since $M\backslash \cup_i I^+_O(r_{a_i})$ is closed and every causal curve issued from $x$ meets $A\subset \cup_i I^+_O(r_{a_i})$ before escaping $O$). We have $x\in I^+(z)\backslash \cup_i J^+_O(r_{a_i})$, and since $D_p\cap (O\times O)=J_O=J\cap (O \times O)$ and $x,r_{a_i}\in O$, we can rewrite it as $x\in I^+(z)\backslash \cup_i D_p^+(r_{a_i})$.

Let us prove that $I^+(z)\backslash \cup_i D_p^+(r_{a_i})$ is a Lawson open set contained in $O'$. It is included in $O$, and so in $O'$: if $y\in I^+(z)\backslash O$, a timelike curve from $z$ to $y$ meets some $I^+_O(r_{a_i})$ at a point $w$ before escaping $O$, so that $(r_{a_i},w)\in I$ and $(w,y)\in I$ give $y\in I^+(r_{a_i})\subset D_p^+(r_{a_i})$.
Since $(M,g)$ is strongly causal and future reflecting, and (SC) holds by Thm.\ \ref{mmwr}(a), Thm.\ \ref{kssw} gives that $\ll\,=I$, so $I^+(z)=\,\twoheaduparrow z$; as $D_p^+(r)=\,\uparrow r$, the set equals $\twoheaduparrow z\backslash \uparrow\{r_{a_1},\ldots,r_{a_k}\}$, a basic Lawson open set. This concludes the proof that the Lawson topology is finer than the manifold topology.

For the converse, we know that the interval topology coincides with the manifold topology (Thm.\ \ref{dzxd}) and that in a weakly bicontinuous poset the interval topology is finer than the Lawson topology (Prop.\ \ref{pldf}), thus the manifold topology is finer than the Lawson topology.

Density holds by \cite[Thm.\ 2.53(ii)]{minguzzi18b} using $D=\bar J$.
\qed \end{proof}

\begin{proof}[Proof using Sec.\ \ref{coinc}]
Following the previous proof every point $x$ admits a causally convex open neighborhood $O$ such that $J_O=D\cap O\times O=J\cap O\times O$, $J_O$ is closed and $\bar O$ is compact. Every causally continuous spacetime is stably causal and hence $K$-causal (recall that under causal continuity $D=\bar J=K$). Moreover, we can find some open set $U\subset O$, $x\in U$, which is $K$-convex \cite[Lemma 16]{sorkin96} \cite[Lemma 5.5]{minguzzi07}.

In particular, every point belongs to a compact causal diamond $J^+(a)\cap J^-(b)$ contained in $U$ and hence in $O$, hence to a compact interval $[a,b]$. This proves $\mathscr{I}$-local compactness.  Moreover, density holds by \cite[Thm.\ 2.53(ii)]{minguzzi18b}.
 By Thm.\ \ref{dzxd} the interval topology coincides with the manifold topology. The diamonds $D^+_p(x) \cap D^-_p(y)$ are clearly connected in the manifold and hence interval topology. By Thm.\ \ref{thm:lawequalsint2} the interval topology coincides with the Lawson topology.
\qed \end{proof}

\begin{corollary} \label{mmgt}
In a causally simple spacetime $(M,g)$ the Lawson topology of the bicontinuous poset $(M,J)$ coincides with the manifold topology (and hence with the interval  and Alexandrov topologies, cf.\ Thm.\ \ref{dzxd}).
\end{corollary}

\subsection{Global hyperbolicity}

We arrive at the statement by  Martin and Panangaden \cite{martin06} to which we add a necessity result.

\begin{theorem} \label{cppr}
Let $(M,g)$ be a globally hyperbolic spacetime then $(M,J)$ is a globally hyperbolic  poset and $I$ is the way-below relation of $J$. The choice $\le\,=J$ and the imposed causality condition are necessary if we want the poset to be globally hyperbolic and the identification of  the chronological relation with the way-below (equiv.\ way-above) relation.
\end{theorem}

\begin{proof}
First we prove  that the proposed structure is a globally hyperbolic poset (this is the direction proved by Martin and Panangaden \cite{martin06}).

Global hyperbolicity implies that $J$ is closed and so $D=J$. As $(M,J)=(M,D)$, and global hyperbolicity implies causal continuity, the poset is bicontinuous and the way-below relation is $I$. By Thm.\ \ref{dzxd} the biScott (interval), Alexandrov and manifold topologies coincide.
The sets $[a,b]=\,\uparrow a\, \cap \downarrow b=J^+(a)\cap J^-(b)$ are compact in the manifold topology and so in the interval topology.
 We conclude that $(M,J)$ is a globally hyperbolic poset.

For the necessity, suppose that we know the poset is globally hyperbolic and its way-below relation is $I$. The latter fact implies that the biScott topology coincides with the Alexandrov topology. Since it is bicontinuous, by Thm.\ \ref{bxxd}, we have $\le\,=D$ and the spacetime is causally continuous which implies strong causality and so  the Alexandrov topology coincides with the manifold topology. The following sets are compact $[a,b]=\,\uparrow a \,\cap \downarrow b=D^+(a)\cap D^-(b)\supset I^+(a)\cap I^-(b)$, (in the biScott topology and so in the manifold topology) thus the chronological diamonds are relatively compact in the manifold topology. A causally continuous spacetime is non-total imprisoning and a non-total imprisoning spacetime with relatively compact chronological diamonds is globally hyperbolic \cite[Def.\ 4.117]{minguzzi18b} \cite[Cor.\ 3.3]{minguzzi08e}. In a globally hyperbolic spacetime $J$ is closed, so $D=J$ and hence $\le=J$.
\qed \end{proof}

\section{Chain realizability and an abstract causal ladder}\label{s:ladder}

The results of Sec.\ \ref{qcft}, and Theorem \ref{bxxd} in particular, suggest identifying bicontinuity as the abstract poset property corresponding to causal continuity. Theorem \ref{cppr}, in turn, suggests regarding globally hyperbolic posets as the abstract version of globally hyperbolic spacetimes, in line with the proposal by Martin and Panangaden.

A difficulty remains, however: there does not seem to be room for the intermediate rung of causal simplicity. One would like to have an intermediate abstract condition that, when applied to $(M,g)$, forces the spacetime to become causally simple and the relation $\le\, = D_p$ on $(M,g)$ to become equal to $J$, hence closed. The problem is that, under bicontinuity, we already have $\le\, = D_p = \bar J$ and the interval topology coincident with the manifold topology, so demanding closure of $\le$ is not the right way to obtain closure of $J$.

This section shows how to recover the missing rung intrinsically, thereby reconstructing the top portion of the causal ladder of spacetimes. As we shall see, global hyperbolicity will need to be strengthened relative to the Martin–Panangaden definition.

\begin{definition}\label{def:chainrel}
Let $(X,\le)$ be weakly bicontinuous. For $x,y\in X$ put
\[
x\preceq_{c} y
\]
 if $x=y$, or if there is a chain $C\subset [x,y]$ with $\min C=x$ and $\max C=y$ which is compact and connected in the interval topology.
\end{definition}

\begin{definition} \label{def:chainrel2} Let $(X,\le)$ be weakly bicontinuous.
We call $X$ \emph{chain-realizable}, or \emph{causally simple}, if
$\preceq_{c}\ =\ \le$.
\end{definition}
The definition asks for a compact connected chain, not for a continuous curve. The two agree under a countability hypothesis by Lemma~\ref{lem:chainiscurve}.


\begin{lemma}\label{lem:chainrel}
$\preceq_{c}$ is reflexive, transitive, contained in $\le$, and satisfies
$\preceq_{c}\circ\ll\ \subset\ \ll$ and $\ll\circ\preceq_{c}\ \subset\ \ll$.
\end{lemma}

\begin{proof}
Reflexivity and $\preceq_{c}\,\subset\,\le$ are immediate. For transitivity, let $C_1$ run from $x$ to $y$ and $C_2$ from $y$ to $z$. Every element of $C_1$ lies below $y$ and every element of $C_2$ above it, so $C_1\cup C_2$ is a chain in $[x,z]$ with least element $x$ and greatest element $z$; it
is compact as a union of two compact sets and connected as a union of two connected sets with the common point $y$. The two absorption properties follow from $\preceq_{c}\,\subset\,\le$ and the analogous properties for $\le$.
\qed \end{proof}

\begin{remark}\label{rem:chainnottrivial}
The inclusion $\ll\,\subset\,\preceq_{c}$ is not automatic, so Definition~\ref{def:chainrel2} is a genuine restriction. In Example~\ref{keim} the interval topology is discrete and every interval is finite, so a chain with two or more elements is disconnected; hence $\preceq_{c}$ is the diagonal while $\ll\,=\,\le$ is not.
\end{remark}

\subsection{Maximal chains on semiclosed ordered spaces}

The results of this subsection do not use the approximation relations, and they use the interval topology only through the properties established in Sec.\ \ref{cpto}. Their natural setting is therefore the theory of ordered topological spaces of Nachbin \cite{nachbin65}, Wallace \cite{wallace45} and Ward \cite{ward54}, and we shall work at that level, returning to weakly bicontinuous posets at the end by way of Cor.\ \ref{cnpp}.

Recall that a topological space $X$ endowed with a partial order $\le$ is a
\emph{semiclosed ordered space} if $\uparrow x$ and $\downarrow x$ are closed for every $x\in X$, and a \emph{closed ordered space} if $\le$ is closed in $X\times X$. Ward \cite{ward54} calls the former a partially ordered topological space and speaks of a semicontinuous, respectively continuous, order. Every closed ordered space is a semiclosed ordered space, for $\uparrow x$ and $\downarrow x$ are the preimages of $\le$ under $y\mapsto (x,y)$ and $y\mapsto (y,x)$, respectively. In a semiclosed ordered space every interval $[x,y]=\, \uparrow x\,\,\cap \downarrow y$ is closed, so it is meaningful to demand that the intervals be compact without any reference to order completeness. By Cor.\ \ref{cnpp}, a weakly bicontinuous poset equipped with the interval topology $\mathscr{I}$ is a closed ordered space, hence a semiclosed ordered space, and its intervals are the sets $[x,y]$ used throughout.

Definition \ref{def:chainrel} makes sense in any semiclosed ordered space, once compactness and connectedness are read in the given topology. It should be kept in mind that the relation $\preceq_{c}$ so obtained depends on that topology, so that it is an order-theoretic invariant only when the topology is itself recovered from the order, as $\mathscr{I}$ is on a weakly bicontinuous poset. This is why the causal ladder is formulated for the interval topology and not at the level of generality of the present
subsection.

The next lemma is due to Wallace \cite{wallace45}, where it appears as the statement that a maximal simple set is closed; it is also  \cite[Lemma 4]{ward54}.

\begin{lemma}[Wallace] \label{lem:maxchainclosed}
Every maximal chain in a semiclosed ordered space is closed.
\end{lemma}

\begin{proof}
Let $C$ be a maximal chain. We claim that
\[
C=\bigcap_{x\in C} (\downarrow x\, \cup \uparrow x).
\]
The inclusion from left to right holds because $C$ is a chain. Conversely, if $y$ belongs to the right hand side then $y$ is comparable with every element of $C$, so $C\cup \{y\}$ is a chain containing $C$, and maximality gives $y\in C$. Each set $\downarrow x\, \cup \uparrow x$ is closed, hence so is $C$.
\qed \end{proof}

\begin{lemma} \label{lem:chainhausdorff}
Let $C$ be a chain in a semiclosed ordered space, endowed with the induced topology and the induced order. Then $C$ is a closed ordered space; in particular $C$ is Hausdorff.
\end{lemma}

\begin{proof}
The induced order on $C$ is linear and semiclosed, and for a linear order the semiclosure of the order is  equivalent to closure by  \cite[Lemma 3]{ward54}. We recall the argument. Let $a\not\le b$ in $C$, hence $b<a$. If there is $c\in C$ with $b<c<a$, put $U=C\backslash \downarrow c$ and $V=C\backslash \uparrow c$; then $a\in U$, $b\in V$, and for $x\in U$, $y\in V$ we get $y<c<x$, so $x\not\le y$. If there is no such $c$, put $U=C\backslash \downarrow b$ and $V=C\backslash \uparrow a$; then $a\in U$ and $b\in V$, and if $x\in U$, $y\in V$ satisfied $x\le y$ we would have $b<x\le y<a$, so that $x$ lies strictly between $b$ and $a$, a contradiction. Thus the order of $C$ is closed in $C\times C$, and a closed ordered space is Hausdorff.
\qed \end{proof}

\begin{lemma}\label{lem:compactchain}
Let $C$ be a non-empty compact chain in a semiclosed ordered space. Then:
\begin{enumerate}
  \item[(i)] every non-empty $A\subset C$ has a supremum and an infimum in $C$;
  \item[(ii)] the order topology of $C$ coincides with the induced topology.
\end{enumerate}
\end{lemma}

\begin{proof}
(i) The sets $C\, \cap \uparrow a$, $a\in A$, are closed and have the finite intersection property: a finite subset of the chain $A$ has a greatest element $a_0$, and the corresponding intersection is $C\,\cap \uparrow a_0$, which contains $a_0$. By compactness the set $U=\cap_{a\in A} (C\,\cap \uparrow a)$ of upper bounds of $A$ in $C$ is non-empty, and it is closed, hence compact. The sets $U\,\cap \downarrow u$, $u\in U$, are closed with the finite intersection property for the same reason, so their intersection contains some $m$, which is then the least element of $U$, i.e.\ $m=\sup_C A$. The infimum is dual.

(ii) The order topology is generated by the rays $\{c\in C: c>a\}=C\,\backslash \downarrow a$ and $\{c\in C: c<a\}=C\,\backslash \uparrow a$, which are open in the induced topology, so the order topology is coarser. It is Hausdorff, being the order topology of a linear order, and the induced topology is compact; a continuous bijection from a compact space onto a
Hausdorff space is a homeomorphism.
\qed \end{proof}

Part (ii) says that such a $C$ is a linearly ordered topological space in its own right. As $C$ satisfies (i), if in addition $C$ is connected it is an \emph{ordered continuum}, but this does not yet make it a curve in the usual sense, and the following records exactly when it is one.

\begin{lemma}[When a compact connected chain is a curve]\label{lem:chainiscurve}
Let $X$ be a semiclosed ordered space and let $C$ be a compact connected chain with at least two elements, $\min C=x$ and $\max C=y$. The following are equivalent:
\begin{enumerate}
  \item[(i)] there is an order isomorphism $[0,1]\to C$, necessarily a homeomorphism;
  \item[(ii)] $C$ is the image of a continuous map $[0,1]\to X$;
  \item[(iii)] $C$ is metrizable;
  \item[(iv)] $C$ is second countable;
  \item[(v)] $C$ contains a countable order-dense subset.
\end{enumerate}
In particular all of these hold whenever $X$ is second countable; for $X$ a weakly bicontinuous poset with the interval topology this holds under
$\mathscr{I}$-second countability (separability suffices under bicontinuity, see Prop.\ \ref{pwdt}), for instance when $X$ is a manifold with a manifold topology finer than $\mathscr{I}$.
\end{lemma}

\begin{proof}
By Lemma~\ref{lem:compactchain} the order of $C$ is complete, and connectedness of a linearly ordered space forces it to be densely ordered as well\footnote{If there are no points in between $x,y$, $x<y$, then $(-\infty, y)$ and $(x,+\infty)$ disconnect $C$.}, so $C$ is an ordered continuum with endpoints. Observe that removing $\sup C$ and $\inf C$ does not disconnect $C$, namely there are two non-cut points.   A characterization by Sierpi\'nski and others of the interval $[0,1]$ of the real line states that such a linear continuum is order isomorphic to $[0,1]$ precisely when it has a countable order-dense subset \cite[6.3.8(c)]{engelking89} \cite[Thm.\ 28.13]{willard70}, which is (v) $\Rightarrow$ (i); the converse is clear, and an order isomorphism between linear orders is a homeomorphism for the order topologies, which by Lemma~\ref{lem:compactchain}(ii) are the induced ones. (i) $\Rightarrow$ (ii) is immediate. For (ii) $\Rightarrow$ (iii), a continuous image of a compact metric space in a Hausdorff space is metrizable \cite[Thm.\ 4.4.18]{engelking89}. For a compact Hausdorff space (iii) and (iv) are equivalent (Urysohn's metrization theorem), and (iv) $\Rightarrow$ (v) by choosing a point in each non-empty member of a countable basis of nondegenerate intervals. Finally, a subspace of a second countable space is second countable.
\qed \end{proof}

In the terminology of Wallace \cite{wallace62} the equivalence of (i) and (iii) reads: a separable ordered continuum is a real arc, that is, it is homeomorphic to the closed unit interval. The result goes back to Eilenberg \cite{eilenberg41}, and it is in this form that it is used by Ward \cite{ward54} in the study of maximal chains.


The following result will be used to prove ``global hyperbolicity implies causal simplicity'' at the abstract level.

\begin{theorem}\label{thm:ghimpliescs}
Let $X$ be a semiclosed ordered space whose intervals are all compact and whose order is dense. Then for all $x\le y$ there is a chain $C\subset [x,y]$  which is maximal in $[x,y]$. Any such chain has  $\min C=x$, $\max C=y$, it is compact and connected. These chains are ordered continua with endpoints; they are curves in the usual sense, that is images of continuous maps
from $[0,1]$, exactly under the equivalent conditions of Lemma~\ref{lem:chainiscurve}. In particular, a weakly bicontinuous poset, endowed with the interval topology, that is interval-compact and dense is
chain-realizable, i.e.\ $\preceq_{c}\ =\ \le$.
\end{theorem}

\begin{proof}
We may assume $x<y$. The subspace $[x,y]$, with the induced topology and order, is again a semiclosed ordered space, and it is compact by hypothesis. By Zorn's lemma the chains contained in $[x,y]$ and containing $x$ and $y$, ordered by inclusion, have a maximal element $C$; the hypothesis of Zorn's lemma holds because a union of a chain of chains is a chain.

\emph{$C$ is compact.} By Lemma~\ref{lem:maxchainclosed}, applied in the semiclosed ordered space $[x,y]$, the chain $C$ is closed in $[x,y]$, hence compact.

\emph{$C$ has no jumps.} Suppose $u<v$ in $C$ with no element of $C$ strictly between. By density there is $z$ with $u<z<v$, and $x\le u<z<v\le y$ puts $z$ in $[x,y]$. Every $c\in C$ is comparable with $u$ and with $v$, and $u<c<v$ is excluded, so $c\le u<z$ or $c\ge v>z$; hence $C\cup\{z\}$ is a chain in $[x,y]$ containing $x$ and $y$ and strictly larger than $C$, contradicting maximality.

\emph{$C$ is connected.} By Lemma~\ref{lem:compactchain}(i) the linear order $C$ is complete, and by the previous paragraph it is densely ordered; a linearly ordered topological space with these two properties is connected \cite[Thm.\ 24.1]{munkres00}. By Lemma~\ref{lem:compactchain}(ii) the order topology agrees with the induced one, so $C$ is connected as a subspace of $X$.

Finally $\min C=x$ and $\max C=y$, since $C\subset [x,y]$ contains both. The last statement follows since, by Cor.\ \ref{cnpp}, a weakly bicontinuous poset with the interval topology is a closed ordered space, hence a semiclosed ordered space.
\qed \end{proof}

\begin{remark} \label{rem:wardattribution}
Theorem 4 of \cite{ward54} states that a connected chain in a semiclosed ordered space is order dense, and that if the space has compact maximal chains then every order dense maximal chain is connected. What is added
here is the observation that compactness of the intervals supplies compact maximal chains inside $[x,y]$, and that density of the order of $X$ makes those maximal chains order dense; the conclusion is then read as a statement about the pair $x\le y$ rather than about the maximal chains of $X$. Let us also mention Theorem 6 of \cite{ward54}: in a semiclosed ordered space with compact maximal chains every maximal chain is connected if and only if $\downarrow x\,\,\cap \uparrow y$ is connected for every pair $x,y$. Thus
connectedness of all the intervals is what makes \emph{all} maximal chains connected, whereas density only guarantees that the maximal chains lying inside a given interval are.
\end{remark}

\begin{proposition}\label{prop:csimpliesdense}
Let $X$ be a semiclosed ordered space in which every pair $x\le y$ is joined by a compact connected chain contained in $[x,y]$. Then the order of $X$ is dense. Thus, in  a semiclosed ordered space with compact intervals,  chain realizability is \emph{equivalent} to density. Consequently, on a weakly bicontinuous interval-compact poset (endowed with the interval topology), chain realizability is \emph{equivalent} to density.
\end{proposition}

\begin{proof}
If $u<m$ with $[u,m]=\{u,m\}$, the only chain in $[u,m]$ from $u$ to $m$ is $\{u,m\}$, and in this two-point subspace both $\{u\}$ and $\{m\}$ are closed, since $\downarrow u$ and $\uparrow m$ are; hence the subspace is discrete and disconnected, so that $u\not\preceq_{c} m$ while $u\le m$.
The equivalence now follows from Theorem~\ref{thm:ghimpliescs}.
\qed \end{proof}

Under compactness of the intervals density can also be replaced by connectedness of the intervals, so that the three conditions that will be used in the sequel coincide.

\begin{theorem} \label{thm:densityconnected}
Let $X$ be a semiclosed ordered space whose intervals are all compact. The following
are equivalent:
\begin{enumerate}
\item[(i)] every pair $x\le y$ is joined by a compact connected chain contained in
$[x,y]$
(i.e.\ $X$ is chain-realizable);
\item[(ii)] the order of $X$ is dense;
\item[(iii)] every interval $[a,b]$ is connected.
\end{enumerate}
\end{theorem}

\begin{proof}
(i) $\Leftrightarrow$ (ii) is Theorem \ref{thm:ghimpliescs} together with Proposition
\ref{prop:csimpliesdense}.

(iii) $\Rightarrow$ (ii). If $a<m$ with $[a,m]=\{a,m\}$, then in this two-point subspace both $\{a\}$ and $\{m\}$ are closed, since $\downarrow a$ and $\uparrow m$ are, so it is discrete and disconnected.

(ii) $\Rightarrow$ (iii). Let $a\le b$, where we may assume $a<b$, and let
$z\in [a,b]$. The set $\{a,z,b\}$ is a chain, so by Zorn's lemma it is contained in a maximal chain $C_z$ of the semiclosed ordered space $[a,b]$. Exactly as in the proof of Theorem \ref{thm:ghimpliescs}, the chain $C_z$ is closed in $[a,b]$ hence compact by Lemma \ref{lem:maxchainclosed}, it has no jumps by density and maximality, it is order complete by Lemma \ref{lem:compactchain}(i) and carries the order topology by Lemma \ref{lem:compactchain}(ii), hence it is connected. Since
\[
[a,b]=\bigcup_{z\in [a,b]} C_z
\]
is a union of connected sets all containing $a$, the interval $[a,b]$ is connected.
\qed \end{proof}

The proof shows that  density does not merely provide one connected chain from $a$ to $b$; it provides one through every point of the interval, and these sweep out the interval.

\begin{remark}\label{cor:pfinnotcs}
Example~\ref{keim} is a bicontinuous interval-compact poset that is not chain-realizable.
Hence interval-compactness alone does not imply chain realizability, and density cannot be
dropped from Theorem~\ref{thm:ghimpliescs}. Indeed, the order has jumps: $\emptyset<\{n\}$ with nothing strictly between. Apply
Proposition~\ref{prop:csimpliesdense}.
\end{remark}

As an application, we have

\begin{theorem} \label{thm:ghlawson}
In a globally hyperbolic poset (in the MP sense) the following conditions are equivalent:
\begin{enumerate}
\item[(a)] the intervals $[a,b]$ are connected;
\item[(b)] density: for each $a<b$ there is $x$ such that $a<x<b$;
\item[(c)] the poset is chain-realizable.
\end{enumerate}
If they hold then $\mathscr{L}=\mathscr{I}$.
\end{theorem}

\begin{proof}
A globally hyperbolic poset is bicontinuous, hence by Cor.\ \ref{cnpp} it is a closed, in particular semiclosed, ordered space, and its intervals are compact. The equivalence of (a), (b) and (c) is therefore Theorem \ref{thm:densityconnected}. As for the topologies,
interval-compactness implies local compactness by Remark \ref{mvvr} and Prop.\ \ref{nncr}, so Thm.\ \ref{thm:lawequalsint2} applies under (a).
\qed \end{proof}

\subsection{Chain realizability implies domination} \label{chart}

Chain realizability feeds back into the Lawson topology theory of Section~\ref{coinc}, and does so under a
hypothesis which can be convenient: one connected chain through the
interval is enough.

\begin{proposition}\label{prop:csimpliesD}
A chain-realizable weakly bicontinuous poset satisfies condition \textup{(D)}.
\end{proposition}

\begin{proof}
Let $a\,\,\overline{\ll}\,\, b$ and $z\in M_b\,\,\cap\uparrow a$. If $z\le b$ then $z\in M_b\cap [a,b]$ and there is
nothing to prove, so assume $z\not\le b$ and let $C$ be a compact connected chain from $a$ to
$z$. From $a\,\,\overline{\ll}\,\, b$ we get $a\in\bar \twoheaddownarrow \, b=\operatorname{Int}_{\mathscr{I}}(\downarrow b)$ by Thm.\ \ref{thm:pointed},
while $z$ lies in the open set $X\setminus\downarrow b$. If $C$ missed $\p_{\mathscr{I}}(\downarrow b)$, these two
disjoint open sets would cover $C$ and both meet it, contradicting connectedness. Choose
$w\in C\cap\p_{\mathscr{I}}(\downarrow b)$. Then $w\le b$ as $\downarrow b$ is closed, $a\le w\le z$ since $w\in C$,
and $w\notin\operatorname{Int}_{\mathscr{I}}(\downarrow b)=\bar \twoheaddownarrow \,b$, so $w\in M_b\cap [a,b]$ and $w\le z$.
\qed \end{proof}

\begin{corollary}\label{cor:cslawson}
A chain-realizable, locally compact, weakly bicontinuous poset satisfies $\mathscr{L}=\mathscr{I}$.
\end{corollary}

\begin{proof}
Proposition~\ref{prop:csimpliesD} and Theorem~\ref{thm:lawequalsint}.
\qed \end{proof}

%
%
%
%

\subsection{Stably causal continuous posets and the full poset ladder}

The following result shows that the abstract definition of causal simplicity we gave, Def.\ \ref{def:chainrel2}, returns the usual notion when applied to $(M,g)$. In particular, it forces $D=D_p=J$ as desired.

\begin{theorem} \label{prop:spacetimechain}
A causally simple spacetime $(M,g)$ is such that $D_p=D=J$ and $(M,J)$ is a chain-realizable bicontinuous poset such that $\ll\,= I$. Moreover, the manifold topology coincides with the interval topology $\mathscr{I}$ (and the Lawson topology $\mathscr{L}$). Conversely, on a spacetime $(M,g)$ the poset $(M,\le)$ is chain-realizable weakly bicontinuous with way-below and way-above relation coincident with $I$ only if $(M,g)$ is causally simple in which case $\le\, =\preceq_{c}\,=J$ and $\mathscr{I}$ is the manifold topology. The chains connecting $(x,y)\in J$ are continuous causal curves in the Hawking-Ellis sense \cite{hawking73,minguzzi18b}.
\end{theorem}


\begin{proof}
If $(M,g)$ is causally simple then $D_p=D=J$, and any $x\le y$ is joined by a causal curve; its image is compact, connected, and totally ordered by $J$, hence a chain as required. The other poset properties have been already established in Cor.\ \ref{cssr} and results thereafter.

Conversely, on a spacetime $(M,g)$ consider the poset $(M,\le)$ which is  weakly bicontinuous with way-below and way-above relation coincident with $I$, hence bicontinuous as they are coincident. By Theorem \ref{bxxd} $(M,g)$ is causally continuous and $\le\, =D$. By Theorem \ref{kxqp} the interval topology, Lawson topology and manifold topology coincide. Let $C$ be a compact connected chain from $x$ to $y$, $(x,y)\in D$. Since $M$ is second countable and $\mathscr{I}$ is the manifold topology, Lemma~\ref{lem:chainiscurve} applies: $C$ is order isomorphic and homeomorphic to $[0,1]$. Parametrizing by the order therefore gives a
genuine continuous curve $\gamma$ from $x$ to $y$, injective and $D$-nondecreasing. Each point of $M$ has a causally convex neighbourhood $O$ inside a convex one on which $D\cap(O\times O)=J_O$, as established in the proof of Theorem \ref{kxqp}. On such an $O$ the curve $\gamma$ is a continuous $J_O$-nondecreasing curve, hence a continuous causal curve in the Hawking-Ellis sense \cite{minguzzi18b}; and causality of a continuous curve is a local property, so $\gamma$ is a continuous causal curve and $(x,y)\in J$. Thus $D\subset J$,  and by \cite[Thm.\ 4.12]{minguzzi18b} $J=\bar J$ and $(M,g)$ is causally simple. The identity $\preceq_{c}\,=J$ follows, since $\preceq_{c}\,\subset\,\le\,=D=J$ and each $J$-related pair is joined by a causal curve hence a connected compact chain.
\qed \end{proof}

%

The topology $\tau$ was introduced in Def.\ \ref{cctpn}.

\begin{proposition} \label{hhcf}
In a bicontinuous poset $\tau=\mathscr{I}$ hence $\tau$ is Hausdorff.
\end{proposition}

\begin{proof}
Under bicontinuity $\twoheaddownarrow x=\bar  \twoheaddownarrow \, x$.
\qed \end{proof}

We recall that under minimal topological conditions Theorem \ref{levthf} establishes the existence of a continuous increasing function on a continuous poset with $\tau$ Hausdorff. This suggests taking the last two properties as the definition of ``stably causal continuous poset''.

\begin{table}[htbp]
\centering
\renewcommand{\arraystretch}{1.35}
\begin{tabular}{cc}
\hline
\textbf{Spacetime $(M,g)$} & \textbf{Poset $(X,\le)$} \\
\hline
global  &  bicontinuity, density, \\
   hyperbolicity  &  interval-compactness \\
$\Downarrow$ & $\Downarrow$ \\
causal simplicity &  bicontinuity, chain-realizability \\
$\Downarrow$ & $\Downarrow$ \\
causal continuity &  bicontinuity \\
$\Downarrow$ & $\Downarrow$ \\
stable causality and future reflectivity &  continuity and $\tau$ is $T_2$ \\
$\Downarrow$ & $\Downarrow$ \\
past-distinction, (SC) and future reflectivity &  continuity \\
\hline
\end{tabular}
\caption{The five rungs of the causal ladder for continuous posets $(X,\le)$ (on the right) and the corresponding  causal ladder for spacetimes $(M,g)$ (on the left).
 Density could be replaced by interval-connectedness by Thm.\ \ref{thm:ghlawson}.
The properties on the right do not assume absence of compact elements i.e.\ irreflexivity of the way-below relation.
The properties on the right applied to $(M,g)$ with the identification $\ll\,=I$ lead to the properties on the left and allow to identify $\le$ with $D_p$, which in the top three cases  coincides with (from top to bottom), $J(=D=\bar J)$, $J(=D=\bar J)$, $D(=\bar J)$, respectively.
The levels are all different: for the top three this is classical \cite{minguzzi18b}; causal continuity is different from ``stable causality and future reflectivity'', just remove a past timelike ray from 1+1 Minkowski spacetime;  ``stable causality and future reflectivity''  is different from ``past-distinction, (SC) and future reflectivity'', just consider the dual of Malament's example \cite[Fig.\ 15]{minguzzi18b} \cite{malament77,malament12}. There are no other levels because adding just future distinction (or strong causality) to the lower level on the left gives the level just above it by \cite[Props.\ 105,107]{minguzzi18b}. Future reflectivity holds at all levels since this is a characteristic property of continuous posets (Prop.\ \ref{jsse}).
}
\label{tab:ladder}
\end{table}

The following result confirms that this is the correct identification.

\begin{theorem} \label{prop:spacesta}
A ``stably causal and future reflecting''  spacetime $(M,g)$ is such that $(M,D_p)$ is a continuous poset with $\ll\,= I$ and $\tau$  Hausdorff and coincident with the manifold topology. Conversely, on a spacetime $(M,g)$ the poset $(M,\le)$ is continuous with Hausdorff $\tau$ and with way-below relation coincident with $I$ only if $(M,g)$ is ``stably causal and future reflecting'' in which case $\le\, =D_p$.
\end{theorem}

\begin{proof}
A ``stably causal and future reflecting'' spacetime is in particular strongly causal, so (SC) holds by Thm.\ \ref{mmwr}, and it is, in particular past-distinguishing, (SC) and future reflecting thus $\ll=I$ by Thm.\ \ref{kssw}. Observe that $\tau$ is then the Alexandrov topology. Since stable causality implies strong causality the Alexandrov topology, and hence $\tau$, is Hausdorff \cite{penrose72} \cite[Thm.\ 4.75]{minguzzi18b} and coincident with the manifold topology.

Conversely, if $(M,\le)$ is a continuous poset with $\ll\,=I$ then, by Thm.\   \ref{kssw}, $(M,g)$ is past-distinguishing and future reflecting with $\le\, =D_p$. As $\tau$ is Hausdorff the Alexandrov topology is Hausdorff and so $(M,g)$ is strongly causal \cite{penrose72} \cite[Thm.\ 4.75]{minguzzi18b}. Strong causality implies future distinction. Future distinction and future reflectivity imply stable causality \cite[Props.\ 105,107]{minguzzi18b}.
\qed \end{proof}

We arrive at the  abstract causal ladder for posets, see Table \ref{tab:ladder} which mimics the top portion of the causal ladder of spacetimes $(M,g)$ \cite{minguzzi18b}.
Observe that global hyperbolicity for a poset is here augmented with the property of density (equiv.\ interval-connectedness cf.\ Thm.\ \ref{thm:ghlawson}) with respect to Martin and Panangaden's original proposal \cite[Sec.\ 6]{martin06}.



\begin{theorem}[Poset causal ladder]
bicontinuity \& interval-compactness \& density  $\Rightarrow$  bicontinuity \& chain-realizability  $\Rightarrow$  bicontinuity  $\Rightarrow$  continuity \& $\tau$ is $T_2$  $\Rightarrow$  continuity.
\end{theorem}

The first two implications hold also replacing in the involved rungs `weak bicontinuity' with `bicontinuity'.
\begin{proof}
The first implication is Theorem \ref{thm:ghimpliescs}, the second is trivial, the third is Prop.\ \ref{hhcf}, the fourth is trivial.
\qed \end{proof}
The vertical implications on the left of Table \ref{tab:ladder} are standard \cite{minguzzi18b} save for the bottom one which follows from  Thm.\ \ref{mmwr}(a).

Each level when imposed on $(M,g)$ together with the said identification of way-below relation with $I$, $\ll\,=I$ leads to the corresponding classical causality property. This is Theorem \ref{cppr} for the first (stronger) globally hyperbolic  rung (density is clear on a globally hyperbolic spacetime $(M,g)$ as any two causally related points are connected by a causal curve and we know that in this case $\le\,=J$); Theorem \ref{prop:spacetimechain} for the second causal simplicity rung; Theorem \ref{bxxd} for the third causal continuity rung; Theorem   \ref{prop:spacesta} for the ``stably causal and future reflecting''  rung; and Theorem \ref{kssw} for the ``past-distinguishing and future reflecting'' rung.


\subsection{Non-partial imprisonment and inextendible chains} \label{s:imprisonment}



In this section we obtain two main results. The former, Thm.\ \ref{thm:escape},  is the order-theoretic counterpart  of the fact that stable causality implies non-partial imprisonment \cite{minguzzi18b}, in fact it is available already at the stably causal rung of Table \ref{tab:ladder}. It once again confirms the correct identification of the stably causal continuous posets.

The latter, Thm.\ \ref{maxconn}, tells us that maximal chains are, under the globally hyperbolic rung and a countability condition, standard inextendible curves. It will be used to prove that under the same globally hyperbolic assumption there are Cauchy time functions.


\begin{definition} \label{def:locbound}
Let $(X,\le)$ be a poset and let $\beta$ be a topology on $X$. A subset $S\subset X$ is \emph{$\beta$-locally bounded above} if every $x\in S$ has a $\beta$-neighbourhood contained in $\downarrow b$ for some $b\in X$, and \emph{$\beta$-locally bounded below} if every $x\in S$ has a $\beta$-neighbourhood contained in $\uparrow a$ for some $a\in X$.
\end{definition}

\begin{lemma} \label{lem:boundchain}
Let $(X,\le)$ be a poset, $\beta$ a topology on $X$, $K\subset X$ a $\beta$-compact subset and $C$ a non-empty chain of $X$.
\begin{itemize}
\item[(i)] If $K$ is $\beta$-locally bounded above and $C\cap K$ is cofinal in $C$, then $C$ is bounded above.
\item[(ii)] If $K$ is $\beta$-locally bounded below and $C\cap K$ is coinitial in $C$, then $C$ is bounded below.
\end{itemize}
\end{lemma}

\begin{proof}
(i) For $x\in K$ choose a $\beta$-open $V_x\ni x$ and $b_x\in X$ with $V_x\subset\, \downarrow b_x$. By compactness $K\subset V_{x_1}\cup\dots\cup \, V_{x_n}$ for finitely many $x_i$, and $n\ge 1$ since $C\cap K\ne \emptyset$. Put $C_i:=C\cap K\cap V_{x_i}$, so that $C\cap K=\bigcup_{i=1}^n C_i$.

Some $C_i$ is cofinal in $C$. Indeed, otherwise for each $i$ there would be $c_i\in C$ such that no element of $C_i$ is above $c_i$; as $C$ is a chain the finite set $\{c_1,\dots,c_n\}$ has a greatest element $c$, and by cofinality of $C\cap K$ there is $c'\in C\cap K$ with $c\le c'$; then $c'\in C_j$ for some $j$ and $c_j\le c\le c'$, a contradiction.

Fix such an $i$. For every $c\in C$ there is $c'\in C_i$ with $c\le c'$, and $c'\in V_{x_i}\subset\, \downarrow b_{x_i}$; hence $c\le b_{x_i}$, and $b_{x_i}$ is an upper bound of $C$. Part (ii) is dual.
\qed \end{proof}


\begin{lemma} \label{lem:maxnobound}
In a poset without maximal elements no maximal chain has an upper bound. Dually, in a poset without minimal elements no maximal chain has a lower bound.
\end{lemma}

\begin{proof}
Let $C$ be a maximal chain and let $b$ be an upper bound of $C$. Then $C\cup\{b\}$ is a chain containing $C$, so $b\in C$ by maximality and $b=\max C$. Since $b$ is not maximal there is $y$ with $b<y$, and $C\cup\{y\}$ is a chain properly containing $C$, a
contradiction.
\qed \end{proof}


\begin{lemma} \label{lem:t1uu}
Let $(X,\le)$ be a continuous poset whose topology $\tau$ is $T_1$. If $x\in X$ is not maximal then $\twoheaduparrow x\ne \emptyset$.
\end{lemma}

\begin{proof}
Suppose $\twoheaduparrow x=\emptyset$ and let $W=\,\twoheaduparrow a\,\cap
\twoheaddownarrow b_1\cap\dots\cap \twoheaddownarrow b_n$ be a basic $\tau$-open set containing $x$. Each $b_i$ would satisfy $x\ll b_i$, that is $b_i\in \,\twoheaduparrow x=\emptyset$; hence $n=0$ and $W=\,\twoheaduparrow a$, which is an upper set. Every $\tau$-neighbourhood of $x$ therefore contains $\uparrow x$. If $x$ were not maximal there would be $y>x$ lying in every neighbourhood of $x$, so that $x\in \operatorname{Cl}_\tau(\{y\})$ with $x\ne y$, against the $T_1$ property.
\qed \end{proof}

We say that $C$ is {\em eventually outside} $K$ {\em in the forward direction} if $C\cap K$ is not cofinal in $C$, that is, if there is $c_0\in C$ such that $c\notin K$ for every $c\in C$ with $c_0\le c$; similarly, $C$ is eventually outside $K$ in the backward direction if $C\cap K$ is not coinitial in $C$.

\begin{theorem} \label{thm:escape}
Let $(X,\le)$ be a continuous poset and let $C$ be a maximal chain of $X$.
\begin{itemize}
\item[(i)] If $X$ has no minimal elements then, for every $\mathscr{S}$-compact $K\subset X$, the set $C\cap K$ is not coinitial in $C$.
\item[(ii)] If $\tau$ is $T_1$ and $X$ has no maximal elements then, for every $\tau$-compact $K\subset X$, the set $C\cap K$ is not cofinal in $C$.
\end{itemize}
In particular, under either hypothesis no maximal chain of $X$ is contained in a $\tau$-compact subset of $X$; and on a stably causal continuous poset with neither maximal nor minimal elements every maximal chain is eventually outside every $\tau$-compact set, in the forward as well as in the backward direction.
\end{theorem}

\begin{proof}
(i) Every $x\in X$ satisfies $\twoheaddownarrow x\ne \emptyset$, this set being directed, so we may pick $a\ll x$; the set $\twoheaduparrow a$ is Scott-open by Thm.\ \ref{nncf} and $\twoheaduparrow a\subset\,\uparrow a$. Thus $X$, and a fortiori $K$, is $\mathscr{S}$-locally bounded below. Were $C\cap K$ coinitial in $C$, Lemma \ref{lem:boundchain}(ii) would make $C$ bounded below, against Lemma
\ref{lem:maxnobound}.

(ii) By Lemma \ref{lem:t1uu} every $x\in X$ admits $b$ with $x\ll b$; the set
$\twoheaddownarrow b$ is $\tau$-open by Def.\ \ref{cctpn} and $\twoheaddownarrow b\subset\, \downarrow b$. Thus $X$ is $\tau$-locally bounded above, and Lemma \ref{lem:boundchain}(i) together with Lemma \ref{lem:maxnobound} concludes as before.

For the last statements recall $\mathscr{S}\subset \tau$, so that a $\tau$-compact set is $\mathscr{S}$-compact, and observe that $C\subset K$ would make $C\cap K=C$ both cofinal and coinitial in $C$.
\qed \end{proof}

\begin{remark} \label{rem:escapesharp}
The hypotheses on the extreme elements cannot be dropped. Let $X=[0,1]$ with its usual order, a weakly bicontinuous poset (Example \ref{ex:unitinterval}) whose unique maximal chain is $C=X$. Here $\twoheaduparrow a=(a,1]$ for $a>0$, $\twoheaduparrow 0=X$, $\twoheaddownarrow b=[0,b)$ for $b>0$ and $\twoheaddownarrow 0=\{0\}$, so that $\tau$ is the Euclidean topology with $0$ made isolated; in particular $\tau$ is Hausdorff. The only Scott-open set containing $0$ is $X$, so $K=[0,\tfrac12]$ is $\mathscr{S}$-compact, and it is coinitial in $C$: this shows that the absence of minimal elements is needed in (i). The set $K=[\tfrac12,1]$ is $\tau$-compact and cofinal in $C$: this shows that the absence of maximal elements is needed in (ii). On a spacetime both hypotheses are automatic, $I^{\pm}(x)$ being non-empty. The $T_1$ property enters only through Lemma \ref{lem:t1uu}, so it may be replaced by the direct assumption $\twoheaduparrow x\ne \emptyset$ for every $x\in X$.
\end{remark}%
%


\begin{corollary} \label{escomp}
Let $(X,\le)$ be a weakly bicontinuous poset. If $X$ has no maximal elements then
maximal chains are eventually outside $\mathscr{S}^{\textrm{op}}$-compact subsets in the forward direction, and if $X$ has no minimal elements then maximal chains are eventually outside $\mathscr{S}$-compact subsets in the backward direction.
Since $\mathscr{I}$-compact subsets are  both  $\mathscr{S}$-compact and $\mathscr{S}^{\textrm{op}}$-compact, the result applies to them and under interval-compactness to intervals $[a,b]$.
\end{corollary}


\begin{proof}
By co-continuity $\bar{\twoheaduparrow}\, x$ is filtered, hence non-empty, so there is $b$ with $x\,\overline{\ll}\, b$; the set $\bar \twoheaddownarrow \, b$ is $\mathscr{S}^{\textrm{op}}$-open and contained in $\downarrow b$, so $X$ is $\mathscr{S}^{\textrm{op}}$-locally bounded above and Lemmas \ref{lem:boundchain}(i) and \ref{lem:maxnobound} apply. The second statement is Thm.\ \ref{thm:escape}(i), and the last follows since $\mathscr{S}\cup \mathscr{S}^{\textrm{op}}\subset \mathscr{I}$.
\qed \end{proof}

We turn to the structure of maximal chains. The first observation is that they are connected, hence abstract counterparts of inextendible causal curves.

%

\begin{proposition} \label{maxconn}
Let $X$ be a weakly bicontinuous poset with compact intervals and dense order, let $Y\subset X$ be order convex, and let $C$ be a maximal chain of $Y$. Then for every $x,y\in C$ with $x\le y$ the set $C\cap [x,y]$ is a maximal chain of the semiclosed ordered space $[x,y]$, hence compact and connected, and $C$ is connected. In particular this applies to $Y=X$, to $Y=\uparrow m$ and $Y=\downarrow m$ for $m\in X$, and to $Y=[a,b]$.
\end{proposition}

\begin{proof}
Order convexity gives $[x,y]\subset Y$. If $z\in [x,y]$ is comparable with every element of $C\cap [x,y]$ then it is comparable with every element of $C$, since for $c\in C$ either $c\le x$, and then $c\le x\le z$, or $c\ge y$, and then $c\ge y\ge z$, or $c\in C\cap[x,y]$; hence $z\in C$ by maximality of $C$ in $Y$, and $z\in C\cap[x,y]$. Thus $C\cap [x,y]$ is a maximal chain of $[x,y]$, and the proof of Thm.\ \ref{thm:ghimpliescs} shows that it is
compact and connected. Finally, $C$ is the union of the sets $C\cap [x,y]$ over the pairs $x\le y$ in $C$, a family of connected sets directed by inclusion and with pairwise non-empty intersection, so $C$ is connected.
\qed \end{proof}

By Thm.\ \ref{thm:ghimpliescs} and Lemma \ref{lem:chainiscurve} the compact connected chains joining two points are curves in the usual sense as soon as $X$ is second countable. A maximal chain has, by Lemma \ref{lem:maxnobound}, no endpoints at all, so the model interval $[0,1]$ must be replaced by $\mathbb{R}$; this is the content of the next result. Recall that in a bicontinuous poset $\tau=\mathscr{I}$ by Prop.\ \ref{hhcf}, so that $\omega$-continuity makes $\mathscr{I}$ second countable by Prop.\
\ref{omsc} (see also Prop.\ \ref{mmft2}); alternatively $\mathscr{I}$-separability suffices by Prop.\ \ref{pwdt}, and second countability is automatic when $X$ is a manifold with a manifold topology finer than $\mathscr{I}$.

\begin{proposition} \label{maxcurve}
Let $X$ be a weakly bicontinuous poset with compact intervals and dense order, and suppose $(X,\mathscr{I})$ is second countable (e.g.\ $X$ is actually weakly $\omega$-bicontinuous, Prop.\ \ref{mmft2}). Let $C$ be a maximal chain with at least two elements. Then the topology induced on $C$ is its order topology, and $C$ is order isomorphic, hence homeomorphic, to a non-degenerate interval of $\mathbb{R}$, closed or open at each end according to whether $C$ has a greatest, respectively least, element. If $X$ has neither maximal nor minimal elements then $C\cong \mathbb{R}$, an order
isomorphism $\gamma:\mathbb{R}\to C$ is a continuous, injective, strictly increasing map which is a homeomorphism onto its image, and by Cor.\ \ref{escomp} for every $\mathscr{I}$-compact $K$ there is $T>0$ with $\gamma(t)\notin K$ for $|t|\ge T$. Maximal chains are therefore the order-theoretic inextendible curves.
\end{proposition}

\begin{proof}
If $C$ has a least and a greatest element then $C=C\cap[\min C,\max C]$ is compact and connected by Prop.\ \ref{maxconn}, and second countable as a subspace of $X$, so $C\cong[0,1]$ by Lemma \ref{lem:chainiscurve}; this is the case already covered by Thm.\ \ref{thm:ghimpliescs}. The same lemma applies to every $C\cap[x,y]$ with $x<y$ in $C$, which is therefore an arc, and this is the only input used below.

Assume now that $C$ has neither a greatest nor a least element, which by Lemma \ref{lem:maxnobound} is the case when $X$ has no maximal and no minimal elements. The sets $C\backslash \uparrow y$, $y\in C$, are open in $C$ because $\uparrow y$ is $\mathscr{I}$-closed, and they cover $C$: given $c\in C$ choose $y\in C$ with $c<y$, so that $y\nleq c$. As $C$ is second countable, hence Lindel\"of, countably many of them cover $C$, and the corresponding $y_n$ form a countable cofinal subset of $C$: if $c\in
C\backslash \uparrow y_n$ then $y_n\nleq c$ and, $C$ being a chain, $c<y_n$. Dually one obtains a countable coinitial subset. Replacing each element by the greatest, respectively least, of finitely many of them, we may choose $x_n,y_n\in C$ with
\[
x_{n+1}<x_n<y_n<y_{n+1}, \qquad C=\bigcup_{n\ge 1} C\cap [x_n,y_n].
\]

We construct order isomorphisms $\varphi_n: C\cap[x_n,y_n]\to [-n,n]$ with
$\varphi_{n+1}$ extending $\varphi_n$. Let $\varphi_1$ be any order isomorphism onto $[-1,1]$, which exists by the first paragraph. Given $\varphi_n$, the sets $C\cap[x_{n+1},x_n]$ and $C\cap [y_n,y_{n+1}]$ are again arcs, so there are order isomorphisms onto $[-n-1,-n]$ and $[n,n+1]$ respectively, and composing with the unique increasing affine self-maps we may assume that they send $x_n$ to $-n$ and $y_n$ to $n$. Since the three pieces are order convex in $C$ and meet only at $x_n$ and $y_n$, their
union is an order isomorphism $\varphi_{n+1}$ extending $\varphi_n$. Let
$\varphi=\bigcup_n \varphi_n$, an order isomorphism of $C$ onto
$\bigcup_n[-n,n]=\mathbb{R}$.

Finally $\varphi$ is a homeomorphism for the induced topology. By Lemma
\ref{lem:compactchain}(ii) the topology induced on the compact chain $C\cap[x_n,y_n]$ is its order topology, so $\varphi_n$ is a homeomorphism; hence so is its restriction to $O_n:=C\backslash(\downarrow x_n\cup \uparrow y_n)$, which is open in $C$ and is mapped onto $(-n,n)$. The $O_n$ cover $C$ and their images cover $\mathbb{R}$, so the bijection $\varphi$ is a homeomorphism, and, an order isomorphism being a homeomorphism for the
order topologies, the topology induced on $C$ is its order topology. The case in which $C$ has exactly one extreme element is obtained by exhausting on one side only, and gives $[0,+\infty)$ or $(-\infty,0]$. The last statements follow because the inclusion $C\hookrightarrow X$ is continuous and because, by Cor.\ \ref{escomp}, $C\cap K$ is neither cofinal nor coinitial in $C$.
\qed \end{proof}

\begin{remark} \label{rem:ghcurves}
The hypotheses of Prop.\ \ref{maxcurve} in its unbounded form are those of the top rung of Table \ref{tab:ladder}, augmented with a countability condition and with the absence of extreme elements: on an  $\omega$-bicontinuous globally hyperbolic poset, in the strengthened sense of Sec.\ \ref{s:ladder}, without maximal or minimal elements, the maximal chains are exactly the inextendible curves $\gamma:\mathbb{R}\to X$, strictly increasing and escaping every compact set in both directions. On $(M,g)$ globally hyperbolic, with $\le\,=J$, these are the inextendible causal curves of Hawking and Ellis
by the argument of Thm.\ \ref{prop:spacetimechain}.
\end{remark}

\subsection{Existence of  Cauchy time functions for posets} \label{s:geroch}

In this section we  confirm the correct identification of global hyperbolicity for posets by showing the existence of a Cauchy time function under a countability condition.  The construction is Geroch's \cite{geroch70}, with the averaging device of Burtscher and Garc\'ia-Heveling \cite{burtscher25}: on a spacetime one uses that $\partial I^{\pm}(p)$ is a hypersurface, hence Lebesgue-null, in order to make the volume functions continuous, and no such argument is available here, but integrating over a thickening parameter makes the difficulty disappear (averaging ideas were first employed by Hawking \cite{hawking68,hawking73}).

We recall that an {\em antichain} is a subset of a partially ordered set where no two distinct elements are comparable to each other.

\begin{proposition} \label{cauchycrit}
Let $X$ be a weakly bicontinuous poset with compact intervals and dense order, with $(X,\mathscr{I})$ second countable and with neither maximal nor minimal elements, and let $f:X\to \mathbb{R}$ be $\mathscr{I}$-continuous, strictly increasing along $<$ and unbounded from above and from below along every maximal chain. Then $f$ restricts to an order isomorphism of every maximal chain onto $\mathbb{R}$, and every level set $f^{-1}(s)$
is a closed antichain which every maximal chain meets exactly once.
\end{proposition}

\begin{proof}
Level sets are closed by continuity and are antichains because $x<y$ gives $f(x)<f(y)$. Let $C$ be a maximal chain. By Prop.\ \ref{maxcurve} $C$ is connected, so $f(C)$ is an interval of $\mathbb{R}$, unbounded in both directions by hypothesis, whence $f(C)=\mathbb{R}$. Two distinct elements of a chain being strictly related, $f|_C$ is strictly increasing, hence an order isomorphism onto $\mathbb{R}$, and $C\cap f^{-1}(s)$ is a single point.
\qed \end{proof}

Throughout the rest of the subsection $X$ is a weakly bicontinuous poset with compact intervals and $(X,\mathscr{I})$ is second countable. By Cor.\ \ref{cnpp} the space $(X,\mathscr{I})$ is Hausdorff, and it is locally compact by interval-compactness, every point having a neighbourhood $(c,e)\subset [c,e]$ with $[c,e]$ compact; being also second countable it is metrizable by Urysohn's theorem, and $\sigma$-compact. A locally
compact, $\sigma$-compact metrizable space admits a compatible \emph{proper} metric \cite{williamson87}, and we fix such a $d$, so that closed balls are compact. By separability $X$ also carries a Borel probability measure $\mu$ of full support, for instance $\mu=\sum_n 2^{-n}\delta_{q_n}$ for a countable dense set $\{q_n\}$. No further property of $d$ or $\mu$ will be required; in particular we do \emph{not} assume that $\mu$ is non-atomic or
that it charges no boundary.

%
%

\begin{definition} \label{avvol}
For $r\in (0,1)$ and $x\in X$ put
\[
D^-_r(x):=\{z\in X:\ d(z,\downarrow x)<r\},\qquad
D^+_r(x):=\{z\in X:\ d(z,\uparrow x)<r\},
\]
and let $\bar D^{\pm}_r(x)$ be the corresponding closed thickenings, obtained by replacing $<r$ with $\le r$. Set $V^{\pm}_r(x):=\mu(D^{\pm}_r(x))$ and $\bar V^{\pm}_r(x):=\mu(\bar D^{\pm}_r(x))$. The \emph{averaged volume functions} of $(d,\mu)$ are
\begin{equation}
t^{\pm}(x):=\int_0^1 V^{\pm}_r(x)\,\mathrm{d}r ,
\end{equation}
the integrals existing because $r\mapsto V^{\pm}_r(x)$ is non-decreasing and bounded by $1$.
\end{definition}

\begin{lemma} \label{avvolcont}
Under the standing assumptions, $0<t^{\pm}\le 1$, the function $t^-$ is strictly increasing and $t^+$ strictly decreasing along $<$, and both are $\mathscr{I}$-continuous.
\end{lemma}

\begin{proof}
\emph{Positivity.} The set $D^-_r(x)$ is open, being the preimage of $[0,r)$ under the continuous map $z\mapsto d(z,\downarrow x)$, and it contains $x$; hence $V^-_r(x)>0$ by full support, while $V^-_r\le \mu(X)=1$. Likewise for $V^+_r$.

\emph{Strict monotonicity.} Let $x<y$. By antisymmetry $y\notin\,\downarrow x$, which is closed by Prop.\ \ref{cvxz}, so $r_0:=d(y,\downarrow x)>0$. Put
\[
\rho:=\tfrac{1}{4}\min(r_0,1),\qquad\text{so that}\qquad 0<\rho\le \tfrac14,\quad 4\rho\le r_0 .
\]
Every $w\in B(y,\rho)$ satisfies $d(w,\downarrow y)\le d(w,y)<\rho$, hence $w\in D^-_r(y)$ for every $r\ge \rho$, and
\[
d(w,\downarrow x)\ \ge\ d(y,\downarrow x)-d(w,y)\ >\ r_0-\rho\ \ge\ 4\rho-\rho\ =\ 3\rho ,
\]
hence $w\notin D^-_r(x)$ for every $r\le 3\rho$. Since $\downarrow x\subset\,\downarrow y$ gives $D^-_r(x)\subset D^-_r(y)$, and hence a non-negative integrand, for every $r$, and since $0<\rho<3\rho<1$,
\[
t^-(y)-t^-(x) = \int_0^1 \!\!\mu\big(D^-_r(y)\backslash D^-_r(x)\big)\,\mathrm{d}r  \ge \int_{\rho}^{3\rho}\! \!\! \mu\big(B(y,\rho)\big)\,\mathrm{d}r  = 2\rho\,\mu\big(B(y,\rho)\big) > 0 ,
\]
the last inequality because $B(y,\rho)$ is a non-empty open set and $\mu$ has full support. The statement for $t^+$ is proved in the same way, exchanging the roles of $\uparrow$ and $\downarrow$ and starting from $r_0:=d(x,\uparrow y)>0$.

\emph{Lower semicontinuity of $V^-_r$.} Here continuity of $X$ is used. Distance to a set equals distance to its closure, so by Prop.\ \ref{clddw} and by the interpolation property of $\ll$ (Thm.\ \ref{oovg}),
\[
d(z,\downarrow x)=d(z,\twoheaddownarrow x)=\inf_{a\ll x} d(z,\twoheaddownarrow a) =\inf_{a\ll x} d(z,\downarrow a),
\]
whence $D^-_r(x)=\bigcup_{a\ll x}D^-_r(a)$, a union directed by directedness of $\twoheaddownarrow x$; as $X$ is hereditarily Lindel\"of, the union of nested open sets becomes a countable union, $V^-_r(x)=\sup_{a\ll x}\mu(D^-_r(a))$. Given $\varepsilon>0$ choose $a\ll x$ with $\mu(D^-_r(a))>V^-_r(x)-\varepsilon$; for $y$ in the Scott-open, hence $\mathscr{I}$-open, set $\twoheaduparrow a$ one has $a\le y$, so $\downarrow a\subset\,\downarrow y$ and $V^-_r(y)>V^-_r(x)-\varepsilon$.

\emph{Upper estimate for $V^-_r$.} Here co-continuity of $X$ is used. Since
$\bar\twoheaduparrow\, x$ is filtered with infimum $x$, its lower bounds form $\downarrow x$, that is $\bigcap_{x\,\overline{\ll}\,b}\downarrow b=\,\downarrow x$; Lindel\"ofness of $X\backslash \downarrow x$ and filteredness give $b_1,b_2,\dots$ with $x\,\overline{\ll}\,b_{n+1} \le b_n$ and $\bigcap_n \downarrow b_n=\,\downarrow x$. If $z\in\bigcap_n D^-_r(b_n)$ pick $w_n\in\,\downarrow b_n$ with $d(z,w_n)<r$; the closed ball of radius $r$ about $z$ being
compact, a subsequence converges to some $w$ with $d(z,w)\le r$, and $w_n\in\,\downarrow b_m$ for $n\ge m$ with $\downarrow b_m$ closed gives $w\in\,\downarrow x$. Hence $\bigcap_n D^-_r(b_n)\subset \bar D^-_r(x)$ and, $\mu$ being finite with the sets decreasing, $\mu(D^-_r(b_n))\to \mu(\bigcap_n D^-_r(b_n))\le \bar V^-_r(x)$. Given $\varepsilon>0$ choose $b$ with $x\,\overline{\ll}\,b$ and $\mu(D^-_r(b))<\bar V^-_r(x)+\varepsilon$; for $y$ in the $\mathscr{S}^{\textrm{op}}$-open, hence $\mathscr{I}$-open, set $\bar\twoheaddownarrow\, b$ one has $y\le b$, so $\downarrow
y\subset\,\downarrow b$ and $V^-_r(y)<\bar V^-_r(x)+\varepsilon$. Together with the previous paragraph, $V^-_r$ is continuous at $x$ whenever $V^-_r(x)=\bar V^-_r(x)$.

\emph{Averaging.} For fixed $x$ the sets $\bar D^-_r(x)\backslash D^-_r(x)\subset \{z:\ d(z,\downarrow x)=r\}$ are pairwise disjoint as $r$ varies, so for each $n$ the set $\{r\in(0,1):\ \mu(\bar D^-_r(x)\backslash D^-_r(x))>1/n\}$ has at most $n$ elements, and $V^-_r(x)=\bar V^-_r(x)$ for all but countably many $r$. If $y_k\to x$ then $V^-_r(y_k)\to V^-_r(x)$ for almost every $r\in(0,1)$, and $V^-_r\le 1$, so $t^-(y_k)\to t^-(x)$ by dominated convergence; as $\mathscr{I}$ is metrizable, $t^-$ is
continuous.

The continuity of $t^+$ is dual.
\qed \end{proof}

\begin{theorem} \label{cauchytime}
Let $X$ be a weakly bicontinuous poset with compact intervals and dense order, without maximal or minimal elements, and with $(X,\mathscr{I})$ second countable. Let $d,\mu$ and $t^{\pm}$ be as above (which exist). Then
\[
f:=\log \frac{t^-}{t^+}
\]
is $\mathscr{I}$-continuous and strictly increasing along $<$, it restricts to an order isomorphism of every maximal chain onto $\mathbb{R}$, and every level set $f^{-1}(s)$ is a closed antichain met exactly once by every maximal chain.
\end{theorem}

\begin{proof}
By Lemma \ref{avvolcont} the quotient $t^-/t^+$ is well defined, positive, continuous and strictly increasing along $<$, hence so is $f$. Let $C$ be a maximal chain and let $x_n\in C$ be decreasing and coinitial in $C$, as furnished by the proof of Prop.\ \ref{maxcurve}. Suppose $\lim_n t^-(x_n)>0$. The sets $D^-_r(x_n)$ decrease in $n$ and $\mu$ is finite, so by dominated convergence
\[
\lim_n t^-(x_n)=\int_0^1 \mu\Big(\bigcap_n D^-_r(x_n)\Big)\,\mathrm{d}r ,
\]
and there is $r\in(0,1)$ with $\bigcap_n D^-_r(x_n)\ne \emptyset$. Pick $z$ in it and $w_n\in\,\downarrow x_n$ with $d(z,w_n)<r$; by properness a subsequence of $(w_n)$ converges to some $w$, and since $w_n\in\,\downarrow x_m$ for every $n\ge m$ with $\downarrow x_m$ closed, $w\le x_m$ for all $m$. By coinitiality $w$ is a lower bound of $C$, contradicting Lemma \ref{lem:maxnobound}. Hence $t^-(x_n)\to 0$, while $t^+(x_n)$ is
non-decreasing in $n$ and lies in $[t^+(x_1),1]\subset (0,1]$, so $f(x_n)\to -\infty$. Dually, along an increasing cofinal sequence of $C$ one gets $t^+\to 0$ and $f\to +\infty$. Proposition \ref{cauchycrit} concludes.
\qed \end{proof}
%

\begin{remark} \label{gerochrec}
For $X=M$ a globally hyperbolic spacetime with $\le\,=J$ and $\ll\,=I$ the interval topology is the manifold topology and Thm.\ \ref{cauchytime} returns Geroch's Cauchy time function \cite{geroch70}.
The use of  weak bicontinuity, the abstract form of reflectivity,  to get continuity of $t^{\pm}$, is  in agreement with the fact that volume functions are continuous precisely on causally continuous spacetimes \cite{hawking74,dieckmann88,minguzzi18b}.
\end{remark}

Combining Thm.\ \ref{cauchytime} with Prop.\ \ref{maxcurve} we obtain the picture one expects at the top of Table \ref{tab:ladder}: on a weakly bicontinuous, interval-compact poset with dense order, without extreme elements and with $\mathscr{I}$ second countable, the maximal chains are inextendible curves $\gamma:\mathbb{R}\to X$, the function $f$ is a
Cauchy time function, each level set $f^{-1}(s)$ is a closed antichain crossed exactly once by every maximal chain, and $X$ is partitioned by these abstract Cauchy sets.
Whether the existence of such an $f$ conversely characterizes the present rung, as it does classically or under the assumption of \cite[Thm.\ 5.4]{burtscher25}, we leave open.

\section{Topological KP-spaces} \label{rjgh}

A {\em Kronheimer and Penrose causal space} \cite{kronheimer67} is a triple $(X,J,I)$ where $(X,J)$ is a poset, $I\subset J$ is irreflexive and $I \circ J \,\cup J \circ I  \subset I$. Let us define a {\em weak Kronheimer and Penrose causal space} with the same axioms but dropping irreflexivity of $I$.

As proved in Theorem \ref{kctt}, we have the following result.

\begin{proposition}
A poset $(X,\le)$ endowed with its way-below relation is a weak KP-causal space.
\end{proposition}

A framework for the study of low regularity spacetime geometry, based on a quintuple $(X,J,I,\rho,d)$, where $(X,J,I)$ is a weak KP-causal space\footnote{Actually, in \cite{kunzinger18} they do not assume antisymmetry of $J$ but, for simplicity we shall assume this property here,  otherwise the relation of their structure with Kronheimer and Penrose's spaces and posets becomes more involved.} was introduced by Kunzinger and S\"amann in \cite{kunzinger18}. The function $d$  is a Lorentzian distance while $\rho$ is an additional metric inducing the topology of the space. A few compatibility conditions between the various notions complete the definition. These will be called KS-Lorentzian length spaces.

Another framework for the study of low regularity spacetime geometry, based only on a Lorentzian distance function $d$, was introduced by S. Suhr and the author, and later improved to the non-bounded case with A. Bykov in \cite{minguzzi22,minguzzi24b,minguzzi26}. A canonical topology and  causal relation can be defined out of the defining pair $(X,d)$. These will be called BMS-Lorentzian metric/length spaces. An introduction will be provided below.

Both frameworks induce a structure of the following type.\footnote{It can be noted that in the BMS-Lorentzian metric spaces, the various ingredients $J$, $I$, $\mathscr{T}$ are tightly related among themselves as they
are all deduced from $d$. Also in those spaces $I$ is irreflexive.}
\begin{definition} \label{deftopkp}
Let $(X,J)$ be a poset, let $I\subset J$ be a relation and let $\mathscr{T}$ be a Hausdorff topology on $X$. The quadruple $(X,J,I, \mathscr{T})$ is called
\begin{itemize}
\item[-] a {\em lower topological KP-space} if $J\circ I\subset I$ and $I^+(x)$ is open for every $x\in X$;
\item[-] an {\em upper topological KP-space} if $I\circ J\subset I$ and $I^-(x)$ is open for every $x\in X$;
\item[-] a {\em topological KP-space} if it is both a lower and an upper topological KP-space, that is, if $(X,J,I)$ is a weak KP-causal space and the sets $I^{\pm}(x)$ are open for every $x\in X$.
\end{itemize}
\end{definition}
Each notion pairs an algebraic condition with a topological one on the same sets: in a lower topological KP-space the chronological future of a point is a $J$-future set, $J^+(I^+(x))=I^+(x)$, and it is open. The terminology {\em lower} refers also to the lower approximation condition ($\star$), $I^-(x)\ne \emptyset$, $x\in \overline{I^-(x)}$ for all $x\in X$, introduced below, with which lower spaces are naturally paired.
Time reversal, namely the replacement of $J$ and $I$ with their inverse relations, exchanges lower and upper spaces, so every statement on lower spaces has a dual one on upper spaces, which we shall not always state. In a lower topological KP-space $I$ is transitive as $I\circ I\subset J\circ I\subset I$.

In the Lorentzian length spaces frameworks mentioned above $I$ is open in $X\times X$, which is stronger than the openness of the sets $I^{\pm}(x)$. None of the following result of this work uses the  openness of $I$.

  The reason for not imposing irreflexivity comes from the fact that such a property is not necessary for the following proofs so not imposing it matches continuous poset theory more closely, as there we can find points such that $x\ll x$, sometimes called {\em compact}.

\begin{remark} \label{jfree}
The relation $J$ will play almost no role in proofs. Let $(X,I,\mathscr{T})$ be such that $I$ is transitive, $\mathscr{T}$ is Hausdorff, the sets $I^+(x)$ are open, and past-distinction holds: $I^-(x)=I^-(y)\Rightarrow x=y$. Then, setting
\[
J:=D_p:=\{(x,y):\ I^-(x)\subset I^-(y)\},
\]
the quadruple $(X,D_p,I,\mathscr{T})$ is a lower topological KP-space. Indeed, $D_p$ is a partial order, $I\subset D_p$ by the transitivity of $I$, and $D_p\circ I\subset I$ holds trivially: if $(x,y)\in I$ and $I^-(y)\subset I^-(z)$ then $x\in I^-(z)$. The choice is canonical, as $D_p$ is the largest relation $J$ such that $J\circ I\subset I$: if $(x,y)\in J$ and $z\in I^-(x)$ then $(z,y)\in J\circ I\subset I$, thus $I^-(x)\subset I^-(y)$. In particular, {\em in every lower topological KP-space $J\subset D_p$}. The opposite absorption $I\circ D_p\subset I$ is, trivially, future reflectivity: $I^-(x)\subset I^-(y)\Rightarrow I^+(x)\supset I^+(y)$ (cf.\ Prop.\ \ref{oovt}).

Dually, under future-distinction and the openness of the sets $I^-(x)$, the quadruple $(X,D_f,I,\mathscr{T})$ with $D_f:=\{(x,y): I^+(x)\supset I^+(y)\}$ is an upper topological KP-space. Finally, if the sets $I^\pm(x)$ are open and ``$I^-(x)=I^-(y)$ and $I^+(x)=I^+(y)\Rightarrow x=y$'', then $(X,D,I,\mathscr{T})$, with $D:=D_p\cap D_f$, is a topological KP-space.
\end{remark}

The purpose of this section is to find conditions on $(X,J,I, \mathscr{T})$ which ensure that $X$ paired with some relation is a poset such that $I$ is the way-below relation. Once they are found they could be readily applied to both frameworks above.

Observe that the conditions to be found, in the case of a BMS-Lorentzian metric space, could conflict with its tight structure. We shall see that this does not happen, showing that the axioms of a BMS-Lorentzian metric space are aptly chosen. In fact some of the conditions to be imposed are automatically satisfied. In the case of a KS-Lorentzian length space adding further conditions is less problematic since its ingredients are less constrained.

The main result of this section is that, adding an approximation condition $x\in \overline{I^\pm(x)}$ and a few topological and connectedness assumptions to a causally continuous topological KP-space converts  $(X,D)$ into a bicontinuous poset whose way-below relation is $I$ (Thm.\ \ref{bicontkp}).


Since continuous posets are future reflecting (Prop.\ \ref{jsse}), the following result will be of importance.

\begin{proposition} \label{ppcb}
In a lower topological KP-space, suppose
\begin{quote}
($\star$)\qquad $I^-(x)$ is non-empty, $x\in \overline{I^-(x)}$ for every $x\in X$.
\end{quote}
Then for every $x\in X$, $\overline{I^-(x)}= \overline{J^-(x)}$ and we have the identity
\[
D_p:=\{(x,y):\  I^-(x) \subset I^-(y) \}=\{(x,y):\  x \in  \overline{J^-(y)} \}
\]
If additionally $J^-(x)$ is closed for every $x$ then $D_p=J$ i.e.\
\begin{equation} \label{hhar}
J=\{(x,y):\  I^-(x) \subset I^-(y) \}.
\end{equation}
If moreover $I\circ J\subset I$, as it happens in a topological KP-space, then future reflectivity holds
\[
I^-(x)\subset I^-(y) \Rightarrow I^+(x) \supset I^+(y).
\]
A dual version holds for upper topological KP-spaces.
\end{proposition}

These results are well known in the regular theory. The last statement is just an instance of the implications in the transverse ladder \cite{minguzzi18b}.

Equation (\ref{hhar}) is analogous to the expression $D_p=\{(x,y): x\in \overline{J^-(y)}\}=\{(x,y): I^-(x)\subset I^-(y)\}$ met in the regular theory, specialized when $J^-(y)$ is closed.

\begin{proof}
From $I^-(x)\subset J^-(x)$ we get $\overline{I^-(x)}\subset \overline{J^-(x)}$. For the other inclusion, let $y\in \overline{J^-(x)}$; since $y\in \overline{I^-(y)}$ by ($\star$), every open neighbourhood $O\ni y$ contains a point $z\in I^-(y)$, so $I^+(z)$ is a neighborhood of $y$ which, by $y\in \overline{J^-(x)}$, contains some point $w\in I^+(z)\cap J^-(x)$ and so $z\in I^-(x)$. As $O$ is arbitrary, $y\in \overline{I^-(x)}$. Hence $\overline{I^-(x)}=\overline{J^-(x)}$.

Suppose $I^-(x)\subset I^-(y)$. Then $x\in \overline{I^-(x)}\subset \overline{I^-(y)}$. For the converse, suppose $x\in \overline{I^-(y)}$ and let $z\in I^-(x)$; then $x\in I^+(z)$, which is open, so there is $w\in I^+(z)\cap I^-(y)$, and so $z\in I^-(y)$. By the arbitrariness of $z$ we get $I^-(x)\subset I^-(y)$. Together with the first paragraph this gives
\[
\{(x,y):\ I^-(x)\subset I^-(y)\}=\{(x,y):\ x\in \overline{I^-(y)}\}
=\{(x,y):\ x\in \overline{J^-(y)}\}.
\]
Assume now that $J^-(y)$ is closed for every $y$. Then $\overline{J^-(y)}=J^-(y)$ and the displayed identity becomes \eqref{hhar}, that is $J=D_p$.

Finally we prove future reflectivity. Suppose $I^-(x)\subset I^-(y)$; by the above
$x\in \overline{I^-(y)}=\overline{J^-(y)}=J^-(y)$, that is $(x,y)\in J$, and the inclusion
$I^+(x)\supset I^+(y)$ follows from $I\circ J\subset I$, which has been used only in this last step.
\qed \end{proof}
%
%
%
%

In Sec.\ \ref{mwcf} we shall see that for a BMS-Lorentzian metric space future reflectivity is sufficient to get the last identity.

We see that  the imposed condition brings the topological KP-space closer to the structure of a continuous poset (see Prop.\ \ref{cnzz}). In fact,  $x\in \overline{I^-(x)}$ looks like a kind of topological counterpart of  the approximation condition of continuous posets $x=\sup \twoheaddownarrow x$.

\begin{proposition} \label{kbpo}
In a lower topological KP-space the condition ($\star$) implies that $I$ is interpolating: if $(x,z)\in I$ then there is $y\in X$ such that $(x,y),(y,z)\in I$. The same conclusion holds in an upper topological KP-space under the dual condition ($\star^{\rm op}$) below.
\end{proposition}

\begin{proof}
The open set $I^+(x)$ contains $z$ and since $z\in \overline{I^-(z)}$, there is $y\in I^+(x)\cap I^-(z)$.
\qed \end{proof}

\begin{proposition} \label{pvpo}
In a first-countable lower topological KP-space the condition ($\star$) implies that for each $z\in X$ there is a sequence $z_i\to z$, $z_i\in I^-(z_{i+1})$.
\end{proposition}

\begin{proof}
Let $z\in X$ and let $O_i$ be an open neighborhood system of $z$. For $O_1$ we can find $z_1 \in O_1\cap I^-(z)$, now take $z_{i+1}\in (I^+(z_i)\cap O_{i+1})\cap I^-(z)$. Then $z_i\to z$, $z_i\in I^-(z_{i+1})$.
\qed \end{proof}

\begin{proposition} \label{pvp3}
In a lower topological KP-space the condition ($\star$) implies the equivalence of (a) $J$ is closed and (b) for each $x$, $J^+(x)$ and $J^-(x)$ are closed.
\end{proposition}

The proof in the regular case is classical \cite{minguzzi06c} and can be adapted with minimal variations.
\begin{proof}
The direction from (a) to (b) is clear. For the converse, let $(x,y)\in \overline{J}$ and let $O$ be an open neighborhood of $x$, we can find $z\in I^-(x)\cap O$. Let $U$ be an open neighborhood of $y$, as $I^+(z)\cap O$ is an open neighborhood of $x$, there are $x'\in I^+(z)\cap O$, $y'\in U$ such that $(x',y')\in J$, thus $(z,y')\in I\subset J$. By the closure of $J^+(z)$ and the arbitrariness of $U$, we get $y\in J^+(z)$ and so $z\in J^-(y)$. By the arbitrariness of $O$, $x\in \overline{J^-(y)}=J^-(y)$.
\qed \end{proof}

The following propositions follow the logic of \cite[Sec.\ 4]{martin06} but adapted to a lower topological KP-space $(X,J,I,\mathscr{T})$.
Observe that the topology is not necessarily  second countable.
Whenever the notions of upper bound, supremum, or way-below relation are mentioned, these refer to the poset $(X,J)$.

\begin{lemma} \label{jjvt}
Let $S\subset J^-(x)$ be a directed subset in a lower topological KP-space for which $J^-(y)$ is closed for every $y$,  and let us regard it as a net $x_s:=s$, $s\in S$,
 \[
 x_s\to x \ \textrm{in the topology of $X$} \ \Rightarrow \ \sup S=x.
 \]
\end{lemma}

\begin{proof} Observe that $x$ is an upper bound for $S$.
Let $y$ be an upper bound for $S$, that is,   $(s, y)\in J$ for every $s\in S$. By assumption $J^-(y)$ is closed. Since $x_s \in J^-(y)$, $x = \lim x_s \in J^-(y)$, and thus $(x, y)\in J$. This proves $x = \sup S$.
\qed \end{proof}

\begin{lemma} \label{pper}
Let $X$ be a lower topological KP-space which satisfies ($\star$),  for which $J^-(x)$ is closed for every $x$, and
\begin{itemize}
\item[($\sharp$)] Every point $x$ admits an open neighborhood basis such that, for every element $O$  in it and for every $z\in O$, $I^+(z)\cap I^-(x)\subset O$
(e.g.\ because the topology  is the Alexandrov topology of $I$).
\end{itemize}
Then, for every $x \in X$, $I^-(x)$ is a directed subset with supremum $x$: $x=\sup I^-(x)$.
\end{lemma}
Observe that ($\sharp$) is a down-sided chronological convexity condition.
\begin{proof}
If $y,z\in I^-(x)$, $I^+(y)\cap I^+(z)$ provides an open neighborhood of $x$, so by $x\in \overline{I^-(x)}$ there is $w\in I^-(x)\cap I^+(y)\cap I^+(z)\subset I^-(x)\cap J^+(y)\cap J^+(z)$ which proves directedness. Let $S=I^-(x)$ then $x_s:=s$ is a net. We want to prove convergence to $x$. It will be sufficient to consider open neighborhoods of $x$ belonging to the family mentioned in ($\sharp$).


For every one-sided chronologically convex open basis element  $O\ni x$, again by $x\in \overline{I^-(x)}$,  we can find $s\in S$ such that $x_s\in O$. The set $I^+(x_s)\cap O$ is again an open neighborhood of $x$ so there is $s'\in S$ such that $x_{s'}\in I^+(x_s)\cap O$. Now, $J^+(x_{s'})\cap S\subset I^+(x_s)\cap I^-(x)\subset O$, which shows that for every larger index than $s'$ the net stays in $O$, that is $x_s\to x$. By Lemma \ref{jjvt} $\sup S=x$.
\qed \end{proof}

Let us denote with $\ll$ the way-below relation of $(X,J)$.
\begin{proposition} \label{wblr}
Let $X$ be a lower topological KP-space which satisfies ($\star$), $(\sharp)$  and
\begin{itemize}
\item[($\dagger$)] the causal diamonds $J^+(x)\cap J^-(y)$, $x,y\in X$, are compact.
\end{itemize}
 Then, $J^-(z)$ is closed for every $z$, so that $J=D_p$ by Prop.\ \ref{ppcb}, and for all $x, y \in X$
\[
x \ll y \iff (x,y)\in I\circ J,
\]
namely iff there is $z$ with $(x,z)\in J$ and $(z,y)\in I$. In particular, $\ll\,=I$ if and only if $I\circ J\subset I$, which is the case if $X$ is a topological KP-space.
\end{proposition}
Observe that ($\dagger$) is a kind of global hyperbolicity condition.

As the proof shows, the closure of $J^-(x)$ really only requires ($\star$) and ($\dagger$).

\begin{proof}
Let $w\in \overline{J^-(z)}$, by Prop.\ \ref{ppcb} $I^-(w)\subset I^-(z)$, but $w\in \overline{I^-(w)}$, thus with $p\in I^-(w)$, $w\in I^+(p)$.
Let $O$ be an open neighborhood of $w$; since $O\cap I^+(p)$ is again an open neighborhood of $w$, it must intersect $I^-(w)$ and so $I^-(z)$, hence $O\cap J^+(p)$ must intersect $I^-(w)$ and so $J^-(z)$. This means $O\cap (J^+(p)\cap J^-(z))\ne \emptyset$ for every $O$ which implies that $w$ is in the closure of $J^+(p)\cap J^-(z)$. But this set is compact, hence closed, so $w\in J^+(p)\cap J^-(z)\subset J^-(z)$.

Let $(x,y)\in I$. Let  $S$ be a directed set such that, $\sup S$ exists and $(y, \sup S)\in J$. Composing the relations we get $(x,\sup S)\in I$.
Let $V:=I^+(x)$ and let
  $\mathcal{F}=\{ J^+(s)\cap J^-(\sup S), \ s\in S\}$, which  is a filtered family of compact subsets.
 Let $u\in \cap_{s\in S}[J^+(s)\cap J^-(\sup S)]$ then $u$ is an upper bound for $S$ so $(\sup S , u)\in J$, but it also belongs to $J^-(\sup S)$ thus $u=\sup S$. In particular $u\in V$.

Recall the topological result according to which given an open set $V$ and a family  $\mathcal{F}=\{F_\alpha\}$ of compact subsets, if $\bigcap \mathcal{F}\subset V$ then there is a finite selection $\{F_{\alpha_1}, \ldots, F_{\alpha_k}\}$ such that $\cap_i F_{\alpha_i}\subset V$. This is precisely our case, so using the filter property, we conclude that there is $s\in S$ such that  $s\in J^+(s)\cap J^-(\sup S)\subset V$, which gives $s\in I^+(x)\subset J^+(x)$. This proves $x\ll y$.

 For the converse, let $x\ll y$. By Lemma \ref{pper}, since $y=\sup I^-(y)$ there is $z\in X$, $(x, z)\in J$, $z\in I^-(y)$, thus $(x,y)\in I\circ J$. Conversely, if $(x,z)\in J$ and $(z,y)\in I$, then $z\ll y$ by the previous paragraph, and $x\ll y$ by Thm.\ \ref{kctt}. Finally, as $J$ is reflexive, $I\subset I\circ J$, so $I\circ J=I$ iff $I\circ J\subset I$.
\qed \end{proof}

Under the assumptions of Prop.\ \ref{wblr} the relation $I\circ J=I\circ D_p$ is the common-future relation $C$ of Sec.\ \ref{kebbkp} (Lemma \ref{comfutkp}), so the proposition is a first instance of the identification of the way-below relation with the future-reflective hull of $I$, see Thm.\ \ref{fixedkp}.

This result will be sufficient to obtain analogous results for BMS-Lorentzian metric spaces.


\subsection{BMS-Lorentzian metric/length spaces} \label{mwcf}
Let $(X,d)$ be a set endowed with a function $d:X\times X\to [0,+\infty)$.  Its {\em future  boundary} is $X^+=\{x\in X:  d(x,y)=0, \ \forall y\in X\}$ while the {\em past boundary} is $X^-=\{x\in X:  d(y,x)=0, \ \forall y\in X\}$. The chronological relation is $I=\{d>0\}$. The chronological future and past of a point or set are defined in the usual way. The {\em chronological diamonds} are sets of the form $I(x,y):=I^+(x)\cap I^-(y)$.
The space $X$ has no boundary points iff $I(X)=X$ where $I(X)=I^+(X)\cap I^-(X)$.
  We say that $d$ satisfies the {\em reverse triangle inequality} if for every  $(x,y)\in I$ and $(y,z)\in I$, we have $d(x,z)\ge d(x,y)+d(y,z)$. We say that {\em $d$ distinguishes points} if for every pair $x,y\in X$, $x\ne y$, there is $z\in X$ such that $d(x,z)\ne d(y,z)$ or $d(z,x)\ne d(z,y)$ (a property analogous to weak distinction).

We recall the following, from \cite[Def.\ 2.1, Prop.\ 2.8]{minguzzi24b}:

\begin{definition}
A {\em Lorentzian metric space} (LMS) $(X,d)$ without  boundary is a set $X$ endowed with a function $d:X\times X\to [0,+\infty)$, called {\em Lorentzian distance}, with the following properties:
\begin{itemize}
\item[(i)] $d$ satisfies the reverse triangle inequality,
\item[(ii)] There is a topology $T$ on $X$ such that $d$ is continuous in the product topology and the chronological diamonds are relatively compact,
\item[(iii)] $d$ distinguishes points,
\item[(iv)] $I(X)=X$.
\end{itemize}
\end{definition}

In what follows we shall always assume that the LMS has no boundary.

As a consequence of the definition, $I$ is open, irreflexive and transitive. The topology $T$ turns out to be unique and coincident with the initial topology of the pointed Lorentzian distance functions $\{d_x,d^x\}$, where $d_x:=d(x,\cdot)$, $d^x:=d(\cdot,x)$. Owing to its uniqueness, it is called {\em the}  Lorentzian metric space topology. It is Hausdorff and locally compact. A subset $G\subset X$ is {\em generating} if $I(G)=X$, in other words, for every $x\in X$ there are $y,z\in G$ such that $x\in I(y,z)$.  An LMS is {\em countably generated} if there is a countable generating set. For a countably generated LMS the LMS topology is $\sigma$-compact, second countable and Polish.

The causal relation for an LMS is \cite[Sec.\ 5.1]{minguzzi22} \cite[Sec.\ 4.1]{minguzzi24b}
\[
J:=\{(x,y):\ d^x \le d^y \ \textrm{and} \ d_x \ge d_y \}
\]
It is closed, reflexive, transitive and antisymmetric. Moreover, $I\subset J$ and $I\circ J\,\cup\, J\circ I \subset I$. The reverse triangle inequality extends its validity to pairs $(x,y)\in J$ and $(y,z)\in J$. The causal diamonds $J(x,y):=J^+(x)\cap J^-(y)$ are compact \cite[Thm.\ 4.5]{minguzzi24b}.
A continuous curve $\sigma: [a, b] \to X$ is {\em isocausal} if $a \le s < t \le b$ implies $\sigma(s) < \sigma(t)$. An isocausal curve is {\em maximizing} (or {\em maximal}) if for $a \le t < t'<t'' \le b$ it satisfies
\[
d(\sigma(t),\sigma(t'))+d(\sigma(t'),\sigma(t''))\le d(\sigma(t),\sigma(t'')).
\]
An LMS is a {\em Lorentzian pre-length space} (pre-LLS) if for every two points  $(x,y) \in I$, there is an isocausal curve connecting them.
An LMS is a {\em  Lorentzian length space} (LLS) if for every two points  $(x,y)\in I$ there is a maximizing isocausal curve connecting them.

In a Lorentzian length space the inclusion $I^+(q) \subset I^+(p)$ is equivalent to the property $d_q \le d_p$, see \cite[Prop.\ 35]{minguzzi25b}. Dually, the inclusion $I^-(p)\subset I^-(q)$ is equivalent to the property $d^p \le d^q$. Thus for a Lorentzian length space the causal relation reads
\[
J:=\{(x,y):\  I^-(x) \subset I^-(y) \ \textrm{and} \ I^+(x) \supset I^+(y) \}.
\]
We see that this is formally the same expression as that of the relation $D$ of regular Lorentzian geometry for $(M,g)$. This tight relation between $I$ and $J$ is precisely what will allow us to prove the (bi)continuity property of $(X,J)$ and the fact that $I$ is the way-below relation of $J$.

By the no-gaps theorem \cite[Sec.\ 5.2]{minguzzi24b} \cite[Thm.\ 9]{minguzzi26} the approximation condition $x\in \overline{I^+(x)}\cap \overline{I^-(x)}$, where closure is with respect to the LMS topology, implies  $\overline{I}=J$ (so $\overline{I^+(x)}= J^+(x)$, $\overline{I^-(x)}= J^-(x)$ for every $x\in X$) and that the chronological diamonds form a basis for the LMS topology, which thus coincides with the Alexandrov topology.

Since in an LMS the causal relation is closed, Proposition \ref{ppcb} specializes as follows

\begin{proposition} \label{ppcn}
In an LMS $(X,d)$, if $x\in \overline{I^-(x)}$ for every $x\in X$, then for every $x\in X$, $\overline{I^-(x)}= J^-(x)$ (so closed) and future reflectivity holds:  $I^-(x) \subset I^-(y) \Rightarrow I^+(x) \supset I^+(y)$. Under future reflectivity for an LLS
\[
J:=\{(x,y):\  I^-(x) \subset I^-(y) \}.
\]
 A dual version holds.
\end{proposition}
%
%

The importance of future or past reflectivity for LMS emerges also from the following result

\begin{theorem}
Let $(M,g)$ be a regular spacetime, and let us consider the pair $(M,d)$, where $d$ is the Lorentzian distance of $(M,g)$. The following conditions are equivalent
\begin{itemize}
\item[(a)] $(M,g)$ is globally hyperbolic,
\item[(b1)] $(M,d)$ is a future reflecting LLS,
\item[(b2)] $(M,d)$ is a past reflecting LLS.
\end{itemize}
 Moreover, in this case the LLS topology coincides with the manifold topology and the LLS causal and chronological relations coincide with those of $(M,g)$.
\end{theorem}
Of course, if either of these holds, reflectivity holds. Since the approximation conditions $x\in \overline{I^-(x)}$  are stronger than the corresponding reflectivity property by  Prop.\ \ref{ppcn}, further versions could be obtained using those conditions.

\begin{proof}
The direction (a) $\Rightarrow$ (b1) and (b2) follows from \cite[Thm.\ 21]{minguzzi25b}, see also \cite[Prop. 2.4, Thm.
2.14]{minguzzi24b}, reflectivity being clear.\\
 Assume (b1), the other case being analogous. By  \cite[Thm.\ 20]{minguzzi25b}, if $(M,d)$ is a future $d$-distinguishing or past $d$-distinguishing LMS, then $(M,g)$ is globally hyperbolic. Here past $d$-distinguishing means the property: $d^p=d^q\Rightarrow p=q$. By \cite[Prop.\ 35]{minguzzi25b} in an LLS $d^p=d^q$ iff $I^-(p)=I^-(q)$, thus we need only to prove  $I^-(p)=I^-(q) \Rightarrow p=q$. We know that an LMS satisfies  ``$d^p=d^q$ and $d_p=d_q \Rightarrow p=q$'' which for an LLS converts to ``$I^-(p)=I^-(q)$ and $I^+(p)=I^+(q) \Rightarrow p=q$''. But under future reflectivity this is equivalent to ``$I^-(p)=I^-(q) \Rightarrow p=q$'' which is what we wanted to prove.
\qed \end{proof}

It is intriguing that the approximation condition that has proved so useful in the study of LMS is so similar to the continuity condition of a poset. As mentioned in previous works, we did not include it in the very definition of LMS because: (a) it would restrict the family of LMS so much that causets \cite{surya19} (discrete LMS) would not be included; (b) it is expected to behave badly under Gromov-Hausdorff convergence.

\begin{proposition} \label{mxjc}
In an LMS conditions ($\dagger$) and ($\sharp$) hold.
\end{proposition}

\begin{proof}
The LMS topology admits basis elements that are  chronologically convex \cite[Cor.\ 3.8]{minguzzi24b} thus  ($\sharp$) holds. Condition  ($\dagger$) holds because the causal diamonds in an LMS are compact. Note also that the causal relation is closed in an LMS.
\qed \end{proof}

\begin{theorem}
Let $(X,d)$ be an LMS.
\begin{itemize}
\item[(a1)] If  $z\in \overline{I^-(z)}$ for every $z\in X$, then $(X,J)$ is a continuous poset and the way-below relation is $I$,
\item[(a2)] If  $z\in \overline{I^+(z)}$ for every $z\in X$, then $(X,J)$ is a co-continuous poset and the way-above relation is $I$,
 \item[(b)] If $z\in   \overline{I^+(z)} \cap \overline{I^-(z)}$    for every $z\in X$, then $(X,J)$ is a bicontinuous poset and the way-below (and way-above) relation is $I$.
\end{itemize}
\end{theorem}


\begin{proof}
(a1). Clearly $(X,J)$ is a poset because $J$ is antisymmetric. It is continuous by Lemma \ref{pper} and Prop.\ \ref{mxjc}. Its way-below relation is $I$ by Prop.\ \ref{wblr}, as $I\circ J\subset I$ in an LMS.

(a2). The previous lemmas have a dual formulation which can be obtained working with the way-above relation. Thus, the proof of (a2) is analogous to that of (a1).

(b). By (a1) and (a2)  the poset $(X,J)$ is both continuous and co-continuous. Moreover, the way-below relation coincides with $I$ which coincides with the way-above relation, thus $(X,J)$  is bicontinuous.
\qed \end{proof}

%
%
%
%

\subsection[Continuous posets from lower topological KP-spaces under $\le \,=D_p$]{Continuous posets from  lower topological KP-spaces under $\le \,=D_p$}

Condition ($\dagger$) might be too strong for some applications. We are going to drop it at the price of imposing  other conditions.
This time  we follow  Sec.\ \ref{qcft} on the regular case.

Proposition \ref{cnzz} once again shows that, for posets that wish to be continuous, imposing $\ll \,= I$ implies $\le =D_p$, where
\[
D_p:=\{(x,y): I^-(x)\subset I^-(y)\}.
\]
If one is interested in the poset $(X,J)$ the condition $D_p=J$, namely (by Prop.\ \ref{ppcb}) the closure of $J^-(x)$ for $x\in X$  will have to be imposed in the results that follow. This is usually a trivial matter.


In what follows the notions of upper bound, supremum, chain, way-below relation, refer to the poset $(X,D_p)$. No first countability assumption will be needed in this and in the next subsections: the countable chains of Prop.\ \ref{pvpo} are replaced by the directed sets $I^-(x)$ (Lemma \ref{dirkp}), and compactness arguments use cluster points rather than convergent subsequences. The antisymmetry of $D_p$ is the past-distinguishing condition $I^-(x)=I^-(y)\Rightarrow x=y$, which thus we have to impose. Note that if $D_p=J$ then this condition can be dropped as causality is already present in the definition of topological KP-space.

We start with the analog of Prop.\ \ref{mmtr}.

\begin{proposition} \label{mmtb}
Let $X$ be a past-distinguishing lower topological KP-space that satisfies ($\star$).
Let $S$ be a subset, then its upper bounds $u$ are characterized by  the inclusion $I^-(S)\subset I^-(u)$. Thus, the upper bounds of $S$ coincide with those of $I^-(S)$,  $\sup S$ exists iff $\sup I^-(S)$ exists and in this case
 \begin{equation}
\sup S=\sup  I^-(S).
\end{equation}
\end{proposition}

The result implies that any supremum is the supremum of a chronologically past open set.
\begin{proof}
If $u$ is an upper bound for $S$, we have for every $s\in S$, $(s,u)\in D_p$, thus for every $s\in S$, $I^-(s)\subset I^-(u)$ hence $I^-(S)\subset I^-(u)$. Conversely, if $u$ is such that $I^-(S)\subset I^-(u)$ then for every $s\in S$,  $I^-(s)\subset I^-(u)$, which reads   $s\in D_p^{-}(u)$ and so $S\subset D_p^{-}(u)$, namely $u$ is an upper bound. The inclusion $I^-(S)\subset I^-(u)$ also proves that every upper bound of $S$ is an upper bound for $I^-(S)$. Conversely, if $u$ is an upper bound for $I^-(S)$, then for every $y\in I^-(S)$, $I^-(y)\subset I^-(u)$. Let $p\in S$ and let $q\in I^-(p)$, by the interpolating property (Prop.\ \ref{kbpo}) there is  $r\in I^-(p)\cap I^+(q)$, thus $r \in I^-(p)\subset I^-(S)$ and $q\in I^-(r)$ so $q\in I^-(u)$, and by the arbitrariness of $q$, $I^-(p)\subset I^-(u)$, which proves, by the arbitrariness of $p$, that $u$ is an upper bound for $S$.

The last statement follows from the fact that the upper bounds for $S$ coincide with those of $I^-(S)$.
\qed \end{proof}

We say that the chain is {\em timelike} or an $I$-chain if $(c, c')\in I$ whenever $(c,c')\in D_p$ and $c\ne c'$. With reference to chains and elsewhere the term {\em countable} includes finite.

\begin{definition}
We say that $X$ is {\em chronologically separable} if there is a countable set $B\subset X$ such that $B\cap I^+(x)\cap I^-(y)\ne\emptyset$ whenever $(x,y)\in I$. 
\end{definition}
If the sets $I^-(x)$ are also open, as in a topological KP-space, and $I$ is interpolating, chronological separability follows from the separability of $\mathscr{T}$, because $I^+(x)\cap I^-(y)$ is then a non-empty open set. In a lower topological KP-space the latter set need not be open, and chronological separability is the condition that has to be imposed.

\begin{proposition} \label{jjil}
Let $X$ be a past-distinguishing lower topological KP-space that satisfies ($\star$). Moreover, assume that $X$ is chronologically separable.
Let $S$ be a directed subset. Then $I^-(S)$ is $I$-directed (hence directed), and there is a non-empty countable $I$-chain $C\subset I^-(S)$ such that $I^-(S)=I^-(C)$. When  $I$ is irreflexive, the chain constructed is infinite.
\end{proposition}

\begin{proof}
Let $x,y\in I^-(S)$, then there are $p,q\in S$, such that $x\in I^-(p)$, $y\in I^-(q)$. Since $S$ is directed, there is $r\in S$ such that  $(p,r),(q,r)\in D_p$, that is $I^-(p),I^-(q)\subset I^-(r)$, thus $x, y \in I^-(r)$, or $r\in I^+(x)\cap I^+(y)$. As $I^+(x)\cap I^+(y)$ is open and $r\in \overline{I^-(r)}$, there is $z\in I^+(x)\cap I^+(y)\cap I^-(r)$, thus $z\in I^{-}(r) \subset I^-(S)$, and $(x,z)\in I (\subset D_p)$, $(y,z)\in I (\subset D_p)$, which proves that $I^-(S)$ is $I$-directed.

Let $B$ be as in the definition of chronological separability. The set $B\cap I^-(S)$ is non-empty since $S\ne\emptyset$, ($\star$) holds and $I$ is interpolating. Enumerate it as a sequence $(x_i)_{i\ge 1}$, repeating its members if it is finite. Set $z_1=x_1$. Using $I$-directedness of $I^-(S)$, choose $z_{i+1}\in I^-(S)\cap I^+(z_i)\cap I^+(x_i)$ for each $i$. Then $C:=\{z_i:i\ge 1\}$ is a non-empty countable $I$-chain, possibly finite, such that $x_i\in I^-(C)$ for every $i$.


Since by construction, $C\subset I^-(S)$, we have $I^-(C)\subset I^-(S)$. Let $x\in I^-(S)$, namely there is $s\in S$ such that $x\in I^-(s)$.  Since $(x,s)\in I$, the set $B\cap I^+(x)\cap I^-(s)$ is non-empty and contained in $I^-(S)$, so it contains some $x_i$, hence $x\in I^{-}(x_i)\subset I^-(z_{i+1}) \subset I^-(C)$.
\qed \end{proof}

We say that the chain $C$ (observe that it is a directed set) has a {\em future endpoint} if there is $z\in X$ such that, regarding the chain as a  net $x_c:=c$, $x_c\to z$.

The property ($\star^{\textrm{op}}$) is: $I^+(x)$ is non-empty and $x\in \overline{I^+(x)}$ for every $x\in X$. The property  `` ($\star$) and  ($\star^{\textrm{op}}$)'' can also be denoted ($\star\star$).
Strong causality is the existence of a basis of $I$-convex open neighborhoods (open can be dropped). Strong causality and ($\star\star$) imply distinction because

\begin{proposition} \label{scdist}
Let $X$ be a strongly causal lower topological KP-space. Then ($\star$) implies past-distinction. Dually, in a strongly causal upper topological KP-space ($\star^{\rm op}$) implies future distinction; in particular, in a strongly causal topological KP-space ($\star\star$) implies distinction.
\end{proposition}

\begin{proof}
We prove that ($\star^{\rm op}$) gives future distinction, the other statement being dual. Suppose $I^+(x)=I^+(y)$ with $x\ne y$. Since $\mathscr{T}$ is Hausdorff and $x$ admits a neighbourhood basis of $I$-convex open sets, there are disjoint open sets $U\ni x$ and $V\ni y$ with $U$ $I$-convex.

By ($\star^{\rm op}$) there is $b\in I^+(x)\cap U$. As $b\in I^+(x)=I^+(y)$ we have $(y,b)\in I$, so $I^-(b)$ is an open neighbourhood of $y$, and so is $V\cap I^-(b)$; by ($\star^{\rm op}$) applied at $y$ there is
\[
a\in I^+(y)\cap V\cap I^-(b).
\]
Then $a\in I^+(y)=I^+(x)$, so $(x,a)\in I$ and $(a,b)\in I$, while $x,b\in U$ and $a\in V$, hence $a\notin U$. This contradicts the $I$-convexity of $U$.
\qed \end{proof}

\begin{proposition} \label{vpoj}
Let $X$ be a past-distinguishing lower topological KP-space that satisfies ($\star$).  Moreover, assume that $X$ is chronologically separable.

Let $S$ be a directed subset, then $I^-(S)$ is directed and there is a  countable  timelike chain $C\subset I^-(S)$ such that $I^-(S)=I^-(C)$.
The upper bounds of $S$ are given by the points in the set
\[
US:=\cap_{y\in C} I^+(y)=\cap_{y\in I^-(S)} I^+(y).
\]

If the timelike chain $C$ has a future endpoint $p$ which is an upper bound, i.e.\ $I^-(C)\subset I^-(p)$, then  $\sup S$ exists,  $\sup S=p$, $C\subset I^-(p)$ and
\begin{equation} \label{nnnn}
I^-(S)=I^-(\sup S)=I^-(C)=I^-(p).
\end{equation}
Under
\begin{itemize}
\item[$\diamond$] for every $x\in X$ there is an open neighborhood $O_x\ni x$ with the property that for every $y\in O_x$, $x\in \overline{I^+(y)} \Rightarrow y\in \overline{I^-(x)}$ (a form of pointed local past reflectivity) (e.g.\ because $J\vert_{O\times O}$ is closed).
\end{itemize}
every future endpoint of $C$ is an upper bound.

If, additionally, the sets $I^-(x)$ are open, ($\star^{\textrm{op}}$), strong causality, the topology is locally compact, and any two points $(x,y)\in I$ are connected by an isocausal curve ($J$-chain which is the image of a continuous curve $\sigma:[0,1]\to X$, $\sigma(0)=x$, $\sigma(1)=y$) then a converse holds: if $\sup S$ exists the countable timelike chain $C$ has future endpoint $p$, so $\sup S=p$, $C\subset I^-(p)$ and Eq.\ (\ref{nnnn}) holds.

The same conclusion holds if, in place of ($\star^{\textrm{op}}$) and strong
causality, one assumes past reflectivity, keeping the openness of the sets $I^-(x)$, local compactness and isocausal connectedness.

Finally, under the above assumptions, suppose that $I^-(S)=I^-(E)$ for a non-empty countable timelike chain $E$. Then $S$ has a supremum if and only if $E$ has a future endpoint. When they exist, the endpoint is $\sup S$ and $I^-(S)=I^-(E)=I^-(\sup S)$. The further inclusion $E\subset I^-(\sup S)$ holds if $E$ has no greatest member, or if its greatest member $m$ satisfies $(m,m)\in I$.
\end{proposition}

\begin{proof}

The countable timelike chain $C$ is the image of the sequence $(z_i)$ constructed in Prop.\ \ref{jjil}. If $C$ is finite, this increasing sequence is eventually constant at the greatest member $m$ of $C$; in particular $(m,m)\in I$. Thus the sequence used below is available also when $C$ is finite.
Observe that $I^-(C)=I^-(S)$ and the upper bounds of the latter set coincide with those of $S$ and those of $C$ by Prop.\ \ref{mmtb}. The condition that $u$ is an upper bound for $I^-(C)$  is $I^-(u)\supset I^-(C)$ or equivalently $I^-(u)\supset C$, which is equivalent to $u\in \cap_{y\in C} I^+(y)$. Note that if $z\in I^-(w)$, $w\in C$ then $I^+(z)\supset I^+(w)$, so  $u\in \cap_{y\in I^-(C)} I^+(y)=\cap_{y\in C} I^+(y)$.

Suppose that $C$ has a future endpoint $p$ which is an upper bound, that is, for every $c\in C$, $I^-(c)\subset I^-(p)$.
Let $q$ be an upper bound for $S$, so for $I^-(S)$ and so for $C$. Let $y_k$ be a sequence on $C$, $y_k\in I^-(y_{k+1})$, $y_k\to p$. We have $y_k \in I^-(C) \subset I^-(q)$. If $z\in I^-(p)$ then $I^+(z)$ is an open neighborhood of $p$, so it contains some $y_k$, $z\in I^-(y_k)\subset I^-(C)\subset I^-(q)$. So $I^-(p)\subset I^-(C)\subset I^-(q)$ which proves $(p,q)\in D_p$ and that $p$ is less than all upper bounds. Thus if  $p$ is an upper bound then $p=\sup S$. In this case $I^-(p)\supset I^-(y_k)$ and so $I^-(p)\supset I^-(C)$ which implies $I^-(p)=I^-(C)$.

Let us prove that the future endpoint $p$ is necessarily an upper bound using  ($\diamond$). Let $O$ be an open neighborhood of $p$. Regarding  $C$ as a sequence $x_c=c$,  $(x_c,x_{d})\in (O\times O)$ for sufficiently large $c,d$, thus letting $x_d\to p$ and using $(x_c,x_d)\in I$, we get $p\in \overline{I^+(x_c)}$. By $\diamond$ this gives $x_c\in \overline{I^-(p)}$ which, since this holds for every $c\in C$ and the chain $C$ is timelike, implies  $x_c\in I^-(p)$ and $I^-(x_c)\subset I^-(p)$. In particular, $C\subset I^-(p)$ and $I^-(S)=I^-(C)\subset I^-(p)$, that is $p$ is an upper bound and so $\sup S=p$.

%
%
%
%

For the penultimate statement, suppose by contradiction that $C$ has no future endpoint and let $q=\sup S$.  There is a relatively compact neighborhood $V\ni q$ and a  sequence $y_k$ on $C$, $y_k\in I^-(y_{k+1})\backslash \{y_{k+1}\}$ such that $y_k\notin \bar V$ for every $k$.

The point $q$ is an upper bound of $S$, hence of $I^-(S)$ and so of $C$, thus $y_k\in I^-(q)$. Let $\sigma_k$ be an isocausal curve connecting $y_k$ to $q$, as $y_k\notin \bar V$ and $q\in V$, by connectedness of $\sigma_k$
there is a point $r_k\in \sigma_k\cap \p V$, and since $\p V$ is compact the sequence $r_k$ has a cluster point $r\in \p V$. As $(r_k,q)\in J$ for every $k$, we get $r\in \overline{J^-(q)}$. Also for fixed $j$, we have for any $k\ge j$, $(y_j, y_k)\in I$, $(y_k, r_k)\in J$, so $(y_j, r_k)\in I$ and $r\in \overline{I^+(y_j)}$ ($r$ being also a cluster point of $\{r_k\}_{k\ge j}$).

If $r\in I^-(q)$ we can find $s\in I^-(q)\cap I^+(r)$, hence $y_k\in I^-(s)$ for every $k$, thus $C \subset I^-(s)$ and $I^-(C)\subset I^-(s)$. This would give that $s$ is an upper bound of $I^-(C)$ and $(s,q)\in I\subset D_p$, a contradiction with $q$ being a supremum. Thus $(r,q)\notin I$. Let $r'\in I^+(r)$, then it cannot be $(r',q)\in D_p$, otherwise $r\in I^-(r')\subset I^-(q)$, which we excluded. However, $r\in I^-(r')$ which implies $y_k\in I^-(r')$ and so $C\subset I^-(r')$, or equivalently $I^-(C)\subset I^-(r')$. This proves that $r'$ is an upper bound.

Suppose that past reflectivity holds. Since $q=\sup S$ and $r'$ is an upper bound, we have $(q,r')\in D_p$, that is $I^-(q)\subset I^-(r')$; hence every $q'\in I^-(q)$ satisfies $q'\in I^-(r')$, that is $(q',r')\in I$. As $r'\in I^+(r)$ was arbitrary, this gives $I^+(r)\subset I^+(q')$, and past reflectivity, in the form $I^+(p)\supset I^+(y)\Rightarrow
I^-(p)\subset I^-(y)$, yields
\[
I^-(q')\subset I^-(r)\qquad \text{for every } q'\in I^-(q).
\]
By the interpolation property of $I$ (Prop.\ \ref{kbpo}) every $z\in I^-(q)$ admits $q'$ with $(z,q')\in I$ and $(q',q)\in I$, so that $z\in I^-(q')\subset I^-(r)$; therefore $I^-(q)\subset I^-(r)$. On the other hand $r\in \overline{J^-(q)}$, so $(r,q)\in D_p$ by Prop.\ \ref{ppcb}, that is $I^-(r)\subset I^-(q)$. Hence $I^-(r)=I^-(q)$ and past-distinction
gives $r=q$, which is impossible because $q\in V$ and $r\in \partial V$. The contradiction proves that $C$ has a future endpoint. Note that this route uses neither ($\star^{\textrm{op}}$) nor strong causality, and that the step showing $(r,q)\notin I$ is not needed for it.

Suppose that ($\star^{\textrm{op}}$) and strong causality hold and let $U\ni q$, $V\ni r$, be disjoint causally convex open neighborhoods. Let (by ($\star^{\textrm{op}}$)) $b\in I^+(q)\cap U$, then $r\in I^-(b)\cap V$, thus we can find (by ($\star^{\textrm{op}}$)) $c\in I^+(r)\cap  I^-(b)\cap V$, thus $r\in I^-(c)$, so $c$ is an upper bound (by the argument above for $r'$), hence being $q$ the supremum, $(q, c)\in D_p$, and so $I^-(q)\subset I^-(c)$. Let $a\in I^-(q)\cap U$, then $a\in I^-(c)$, $c\in I^-(b)$ with $a,b\in U$ and $c\in V$, which contradicts the causal convexity of $U$. The contradiction proves that $C$ has a future endpoint.


The openness of the chronological pasts has been used only in the last two parts of the proof, where $I^-(s)$, $I^-(r')$ and $I^-(b)$ enter as open neighbourhoods of $r$.

Finally, suppose that $I^-(S)=I^-(E)$ for a non-empty countable timelike chain $E$. If $E$ has a greatest member $m$, then $I^-(E)=I^-(m)$. By Prop.\ \ref{mmtb}, $m$ is the supremum of $S$; the net indexed by $E$ is eventually equal to $m$ and, since the topology is Hausdorff, has the unique future endpoint $m$. In this case $E\subset I^-(m)$ holds precisely when $(m,m)\in I$. If $E$ has no greatest member, it admits a strictly increasing cofinal sequence, and the preceding arguments apply to that sequence in place of the sequence on $C$. This proves the stated conclusions for $E$.
\qed \end{proof}

\subsection{The way-below relation of $D_p$ is the chronological relation}
\label{kebbkp}

Let us denote by
\begin{equation}
C^+(x):=\operatorname{Int}\bigcap_{r\in I^-(x)} I^+(r)
\end{equation}
the \emph{common future} of $x$, and let us write $(x,y)\in C$ for $y\in C^+(x)$. In the
regular theory the identification of this set with $\operatorname{Int}D_p^+(x)$ and the
characterization of future reflectivity through it are \cite[Prop.\ 4.5, Def.\
4.6]{minguzzi18b}; here they have to be established, and both follow from ($\star$) alone.

\begin{lemma} \label{comfutkp}
In a lower topological KP-space satisfying ($\star$) one has $I^-(I^-(x))=I^-(x)$ and
\begin{equation} \label{comfutid}
\bigcap_{r\in I^-(x)} I^+(r)=\{u\in X:\ I^-(x)\subset I^-(u)\}=D_p^+(x),
\end{equation}
so that $C^+(x)=\operatorname{Int}D_p^+(x)$. Moreover $I\subset D_p$ and hence
$I^+(x)\subset C^+(x)$. Finally, $C$ is transitive,
\begin{equation} \label{hullid}
C=I\circ D_p,\qquad C^+(x)=\bigcup_{w\in D_p^+(x)} I^+(w),\qquad C^-(x)=D_p^-\big(I^-(x)\big),
\end{equation}
and $I\circ J\subset C$, with equality if the sets $J^-(x)$, $x\in X$, are closed.
\end{lemma}

\begin{proof}
The inclusion $I^-(I^-(x))\subset I^-(x)$ is the transitivity of $I$.
Conversely, if $r\in I^-(x)$ then by the interpolation property of Prop.\ \ref{kbpo} there is $w$ with $(r,w),(w,x)\in I$, so $r\in I^-(I^-(x))$.

For \eqref{comfutid}, $u$ belongs to the left-hand side iff $(r,u)\in I$ for every $r\in I^-(x)$, that is iff $I^-(x)\subset I^-(u)$, which is $(x,u)\in D_p$. Finally $(x,u)\in I$ gives $I^-(x)\subset I^-(u)$ by transitivity of $I$, so $I\subset D_p$; as $I^+(x)$ is open and contained in $D_p^+(x)$, it is contained in its interior.

If $(x,y)\in C$ and $(y,z)\in C$, there is a neighborhood $V\ni z$ such that $V\subset D_p^+(y)$, so, since $C\subset D_p$ and $D_p$ is transitive $V\subset D_p^+(x)$ and so $z\in \textrm{Int} D_p^+(x)$, namely $(x,z)\in C$.

If $(x,w)\in D_p$ and $(w,y)\in I$, then $I^+(w)$ is an open neighbourhood of $y$ contained in $D_p^+(w)\subset D_p^+(x)$, so $y\in C^+(x)$. Conversely, if $y\in C^+(x)$, then by ($\star$) the open neighbourhood $C^+(x)$ of $y$ meets $I^-(y)$ at some $w$, so that $(x,w)\in D_p$ and $(w,y)\in I$. This proves $C=I\circ D_p$, and the expressions for $C^{\pm}(x)$ follow. Finally, $J\subset D_p$ (Remark \ref{jfree}) gives $I\circ J\subset I\circ D_p=C$, and if the sets $J^-(x)$ are closed then $J=D_p$ by Prop.\ \ref{ppcb}.
\qed \end{proof}

\begin{proposition} \label{frcomf}
In a lower topological KP-space satisfying ($\star$), future reflectivity, that is $D_p\subset
D_f$, or equivalently $I\circ D_p\subset I$, holds if and only if
\[
C^+(x)=I^+(x)\qquad \text{for every } x\in X,
\]
i.e.\ $C=I$.
\end{proposition}

\begin{proof}
Assume future reflectivity and let $y\in C^+(x)$, so that there is an open $O$ with $y\in O\subset D_p^+(x)$. By ($\star$), $y\in \overline{I^-(y)}$, so there is $w\in O\cap I^-(y)$. From $w\in D_p^+(x)$ we get $(x,w)\in D_p$, hence $I^+(w)\subset I^+(x)$ by future reflectivity, and $y\in I^+(w)$ gives $y\in I^+(x)$. Thus $C^+(x)\subset I^+(x)$, and the reverse inclusion is Lemma \ref{comfutkp}.

Conversely, assume $C^+=I^+$ and let $I^-(x)\subset I^-(y)$, namely $(x,y)\in D_p$. By transitivity of $D_p$ we have $D_p^+(y)\subset D_p^+(x)$, hence $C^+(y)\subset C^+(x)$, that is $I^+(y)\subset I^+(x)$, which is future reflectivity.
\qed \end{proof}

The relation $C=I\circ D_p$ will be called the {\em future-reflective hull} of $I$. It is the smallest relation $R\supset I$ such that $R\circ D_p\subset R$, and it reproduces the structure of $I$, with future reflectivity in addition.

\begin{proposition}[Future-reflective hull] \label{hullkp}
Let $(X,J,I,\mathscr{T})$ be a past-distinguishing lower topological KP-space satisfying ($\star$). Then $(X,D_p,C,\mathscr{T})$ is a past-distinguishing lower topological KP-space satisfying ($\star$). Moreover:
\begin{itemize}
\item[(a)] $I\subset C\subset D_p$, $C$ is transitive and interpolating, and $C$ is irreflexive if and only if $I$ is irreflexive;
\item[(b)] $C$ induces the same order as $I$: $\{(x,y):\ C^-(x)\subset C^-(y)\}=D_p$;
\item[(c)] $C\circ D_p\subset C$, so that $(X,D_p,C)$ is a weak KP-causal space and $C$ is future reflecting: $C^-(x)\subset C^-(y)\Rightarrow C^+(x)\supset C^+(y)$;
\item[(d)] the hull is idempotent: the future-reflective hull of $C$ is $C$;
\item[(e)] $C=I$ if and only if $I$ is future reflecting;
\item[(f)] $I$ satisfies {\em (SC)} if and only if $C$ does, namely, if and only if $C^-(S)=C^-(\sup S)$ for every directed set $S$ admitting a supremum;
\item[(g)] if $I$ is chronologically separable, so is $C$.
\end{itemize}
\end{proposition}

\begin{proof}
By Lemma \ref{comfutkp}, $C=I\circ D_p$ and the sets $C^+(x)$ are open. By past-distinction $D_p$ is a partial order, and $C\subset D_p\circ D_p=D_p$ since $I\subset D_p$. Moreover $D_p\circ C=D_p\circ I\circ D_p\subset I\circ D_p=C$ because $D_p\circ I\subset I$ (Remark \ref{jfree}). As $D_p$ is reflexive, $I\subset C$, hence $I^-(x)\subset C^-(x)$ and ($\star$) holds for $C$. Thus $(X,D_p,C,\mathscr{T})$ is a lower topological KP-space satisfying ($\star$). The  past-distinguishing property will follow in a moment from  (b).

(a) Transitivity is in Lemma \ref{comfutkp}, and interpolation follows from Prop.\ \ref{kbpo} applied to $(X,D_p,C,\mathscr{T})$. If $(x,x)\in C$ there is $w$ with $(x,w)\in D_p$ and $(w,x)\in I$, hence $w\in I^-(x)\subset I^-(w)$. The converse follows from $I\subset C$.

(b) As $C^-(x)=D_p^-(I^-(x))$, the inclusion $I^-(x)\subset I^-(y)$ implies $C^-(x)\subset C^-(y)$. Conversely, let $C^-(x)\subset C^-(y)$ and $a\in I^-(x)$. By interpolation there is $b$ with $(a,b),(b,x)\in I$, so $b\in C^-(x)\subset C^-(y)$, that is, $(b,w)\in D_p$ and $(w,y)\in I$ for some $w$. Then $a\in I^-(b)\subset I^-(w)\subset I^-(y)$. Thus $I^-(x)\subset I^-(y)$.

(c) $C\circ D_p=I\circ D_p\circ D_p=C$. By (b) this is the future reflectivity of $C$. Together with $C\subset D_p$ and $D_p\circ C\subset C$ it shows that $(X,D_p,C)$ is a weak KP-causal space.

(d) By the first part, Lemma \ref{comfutkp} applies to $(X,D_p,C,\mathscr{T})$, and by (b) the relevant order is again $D_p$, so the hull of $C$ is $C\circ D_p=C$.

(e) This is Prop.\ \ref{frcomf}.

(f) By (b), directed sets and suprema are the same for $I$ and $C$. Let $S$ be directed with $p=\sup S$. Since $C^-(S)=D_p^-(I^-(S))$ and $C^-(p)=D_p^-(I^-(p))$, the identity $I^-(S)=I^-(p)$ implies $C^-(S)=C^-(p)$. Conversely, suppose $C^-(S)=C^-(p)$ and let $a\in I^-(p)$. Interpolate $(a,b),(b,p)\in I$; then $b\in C^-(p)=C^-(S)$, so there are $w\in X$ and $s\in S$ with $(b,w)\in D_p$ and $(w,s)\in I$. Hence $a\in I^-(b)\subset I^-(w)\subset I^-(s)$, that is $a\in I^-(S)$. The inclusion $I^-(S)\subset I^-(p)$ holds as $p$ is an upper bound.

(g) Let $B$ be a countable set as in the definition of chronological separability for $I$, and let $(x,y)\in C$, so that $(x,w)\in D_p$ and $(w,y)\in I$ for some $w$. There is $b\in B\cap I^+(w)\cap I^-(y)$, and then $(x,b)\in I\circ D_p=C$ and $(b,y)\in I\subset C$.
\qed \end{proof}

The hull is not, in general, a topological KP-space, as the sets $C^-(x)$ need not be open. This is the reason for introducing the one-sided notion of Def.\ \ref{deftopkp}.

\begin{example} \label{hullnotopen}
Let $M$ be 2-dimensional Minkowski spacetime, $g=-\dd t^2+\dd x^2$, with the closed half-line $\{t=0,\ x\ge 0\}$ removed; it is strongly causal but not future reflecting. Let $p=(-1,1)$, $w=(1,\tfrac12)$, $z=(2,\tfrac12)$, where the first coordinate is $t$. The point $p$ lies on the null ray ending at the tip $(0,0)$ of the removed set, and one checks that $I^-(p)\subset I^-(w)$; as $(w,z)\in I$, we get $(p,z)\in C$. On the other hand, every future-directed timelike curve from a point of $\{t<0\}$ to $z$ crosses $\{t=0\}$ at some $x<0$, hence $I^-(u)\cap\{t<0\}\subset\{x<-t\}$ for every $u\in I^-(z)$. For $\epsilon>0$ the point $p_\epsilon=(-1,1+\epsilon)$ has in its chronological past points of the region $\{t<0,\ x>-t\}$, so $I^-(p_\epsilon)\not\subset I^-(u)$ for every $u\in I^-(z)$, namely $p_\epsilon\notin D_p^-(I^-(z))=C^-(z)$. As $p_\epsilon\to p$, $C^-(z)$ is not open.
\end{example}

The role played in the regular theory by a timelike curve with a future endpoint is played here by the chains of Prop.\ \ref{pvpo}.

\begin{lemma} \label{endpkp}
Let $X$ be a past-distinguishing lower topological KP-space satisfying ($\star$) whose topology is first-countable. Then for every $y\in X$ there is a countable timelike chain $E\subset I^-(y)$ with future endpoint $y$, and
\[
\sup I^-(E)=\sup E=y,\qquad I^-(E)=I^-(y).
\]
\end{lemma}

\begin{proof}
By Prop.\ \ref{pvpo} there is a sequence $E=\{z_i\}$ with $z_i\to y$, $z_i\in I^-(z_{i+1})$ and, by its construction, $z_i\in I^-(y)$ for every $i$; the chain is timelike and $E\subset I^-(y)$, so $I^-(z_i)\subset I^-(y)$ and $y$ is an upper bound of $E$.

Let $q$ be an upper bound of $E$, that is $I^-(E)\subset I^-(q)$, and let $z\in I^-(y)$. Then $I^+(z)$ is an open neighbourhood of $y$, so it contains some $z_i$, whence $z\in I^-(z_i)\subset I^-(E)\subset I^-(q)$. Therefore $I^-(y)\subset I^-(q)$, that is $(y,q)\in D_p$, and $y=\sup E$.
If $w\in I^-(y)$, then $I^+(w)$ is a neighborhood of $y$, so there is some $z_i$ from which we get $w\in I^-(E)$, and so $I^-(y)\subset I^-(E)$. The reverse inclusion follows from $y$ being an upper bound.

The identity $\sup I^-(E)=\sup E$ comes from Prop.\ \ref{mmtb}.
\qed \end{proof}

Lemma \ref{endpkp} is the analogue of the existence of timelike curves with future endpoint $y$. It will not be needed: the following result, which uses no countability assumption, replaces it.

\begin{lemma} \label{dirkp}
Let $X$ be a past-distinguishing lower topological KP-space satisfying ($\star$). For every $y\in X$ the set $I^-(y)$ is $I$-directed, hence directed, and $\sup I^-(y)=y$.
\end{lemma}

\begin{proof}
Let $a,b\in I^-(y)$. Then $I^+(a)\cap I^+(b)$ is an open neighbourhood of $y$, so by ($\star$) there is $c\in I^-(y)\cap I^+(a)\cap I^+(b)$, and $(a,c),(b,c)\in I\subset D_p$. A point $u$ is an upper bound of $I^-(y)$ iff $I^-(w)\subset I^-(u)$ for every $w\in I^-(y)$, that is iff $I^-(I^-(y))\subset I^-(u)$, which by Lemma \ref{comfutkp} reads $I^-(y)\subset I^-(u)$, namely $(y,u)\in D_p$. Thus $y$ is the least upper bound.
\qed \end{proof}

\begin{theorem} \label{comfutthmkp}
Let $X$ be a past-distinguishing lower topological KP-space satisfying ($\star$), and consider the poset $(X,D_p)$. Then
\[
\twoheaduparrow x\ \subset\ C^+(x)\qquad \text{for every } x\in X,
\]
and if $X$ satisfies {\em (SC)} the inclusion is an equality, that is $\ll\,=C$.
\end{theorem}

\begin{proof}
\emph{The inclusion.} Let $x\ll y$. By Lemma \ref{dirkp} the set $S:=I^-(y)$ is directed with $\sup S=y$, so by the definition of the way-below relation there is $s\in I^-(y)$ with $(x,s)\in D_p$. Thus $(x,y)\in I\circ D_p=C$ by Lemma \ref{comfutkp}.

\emph{The converse under {\em (SC)}.} Let $y\in C^+(x)$ and let $S$ be a directed subset with supremum $p$ such that $(y,p)\in D_p$, that is $I^-(y)\subset I^-(p)$. By ($\star$), $y\in \overline{I^-(y)}\subset \overline{I^-(p)}$, and $C^+(x)$ is an open neighbourhood of $y$, so there is $w\in C^+(x)\cap I^-(p)$. From $w\in C^+(x)\subset D_p^+(x)$ we get
$(x,w)\in D_p$, and from $w\in I^-(p)=I^-(S)$, which is {\em (SC)}, we get $(w,s)\in I\subset D_p$ for some $s\in S$. Transitivity gives $(x,s)\in D_p$. As $S$ was arbitrary, $x\ll y$.
\qed \end{proof}

\begin{corollary} \label{iiskkkp}
Let $X$ be a past-distinguishing lower topological KP-space satisfying ($\star$), and consider the poset $(X,D_p)$.
\begin{enumerate}
  \item[(i)] If $X$ satisfies {\em (SC)} then $I\subset\,\ll$.
  \item[(ii)] If $X$ is future-reflecting then $\ll\,\subset I$.
  \item[(iii)] Under {\em (SC)}, one has $\ll\,=I$ if and only if $X$ is future reflecting.
  \item[(iv)] Under future reflectivity, one has $\ll\,=I$ if and only if {\em (SC)} holds.
\end{enumerate}
\end{corollary}
Actually, the direction to the right of (iv) does not need future reflectivity.

\begin{proof}
(i) By Lemma \ref{comfutkp} and Thm.\ \ref{comfutthmkp}, $I^+(x)\subset
C^+(x)=\,\twoheaduparrow x$.

(ii) By Thm.\ \ref{comfutthmkp} and Prop.\ \ref{frcomf}, $\twoheaduparrow x\subset C^+(x)=I^+(x)$.

(iii) By Thm.\ \ref{comfutthmkp} the identity $\ll\,=I$ reads $C^+(x)=I^+(x)$ for every $x$, which is future reflectivity by Prop.\ \ref{frcomf}.

(iv) Sufficiency of {\em (SC)} follows from (i) and (ii). For the necessity, assume $\ll\,=I$, let $S$ be directed with $p=\sup S$ and let $x\in I^-(p)$. By Prop.\ \ref{kbpo} interpolate $(x,x')\in I$, $(x',p)\in I$; since $\ll\,=I$ we have $x'\ll p$, so, $p$ being the supremum of $S$, there is $s\in S$ with $(x',s)\in D_p$. From $(x,x')\in I$ and $I^-(x')\subset I^-(s)$ we get $(x,s)\in I$, that is $x\in I^-(S)$, hence $I^-(p)\subset I^-(S)$; the inclusion $I^-(S)\subset I^-(p)$ is immediate since $p$ is an upper bound of $S$.
\qed \end{proof}

\begin{theorem} \label{contkp}
Let $X$ be a past-distinguishing lower topological KP-space satisfying ($\star$) and {\em (SC)}. Then $(X,D_p)$ is a continuous poset, its way-below relation is the common-future relation, $\ll\,= C=I\circ D_p$, and  $I\subset\, \ll\, \subset D_p$.
\end{theorem}

\begin{proof}
\emph{$I^-(x)$ is directed with supremum $x$.} This is Lemma \ref{dirkp}.

\emph{Continuity.} By Cor.\ \ref{iiskkkp}(i), which uses {\em (SC)} and no reflectivity, we have $I\subset\,\ll$, that is $I^-(x)\subset\, \twoheaddownarrow x$ for every $x$. Proposition \ref{wwgf} applied with $D_x=I^-(x)$ then gives that $(X,D_p)$ is continuous.

\emph{Identification of $\ll$.} This is Thm.\ \ref{comfutthmkp}, whose second half uses {\em (SC)} alone. The inclusions $I\subset\,\ll\,\subset D_p$ follow from Lemma \ref{comfutkp}.
\qed \end{proof}
%
%
%
%
%

Future reflectivity is not needed for the poset to be continuous: it is needed only in order to identify the way-below relation with the chronological relation, and by Prop.\ \ref{frcomf} it is precisely the condition under which the common future of a point reduces to its chronological future.

\begin{theorem} \label{ksswkp}
Let $X$ be a past-distinguishing lower topological KP-space satisfying ($\star$), which is future-reflecting and satisfies {\em (SC)}. Then $(X,D_p)$ is a continuous poset and $I$ is its way-below relation. If moreover $J^-(x)$ is closed for every $x$, then $D_p=J$ by Prop.\ \ref{ppcb} and the statement holds for the poset $(X,J)$.

Conversely, let $X$ be a lower topological KP-space such that $(X,\le)$ is a continuous poset with $\ll\,=I$. Then necessarily $\le\,=D_p$, the space is past-distinguishing and future-reflecting, {\em (SC)} holds, $I^-(x)$ is non-empty for every $x$ and $I$ is interpolating. Condition ($\star$) is, however, not forced: it holds whenever $\mathscr{T}$ is coarser than the Alexandrov topology of $I$, but not conversely, and no property of the poset can imply it.
\end{theorem}

\begin{proof}
\emph{Sufficiency.} By Thm.\ \ref{contkp} the poset $(X,D_p)$ is continuous with $\ll\,=C$, and by Prop.\ \ref{frcomf} future reflectivity gives $C^+(x)=I^+(x)$ for every $x$, that is $\ll\,=I$.

\emph{Necessity of the order-theoretic conditions.} The necessity of $\le\,=D_p$ is Prop.\ \ref{cnzz} together with $\ll\,=I$, past-distinction is the antisymmetry of $\le$, future reflectivity is Prop.\ \ref{jsse}, and {\em (SC)} is the necessity statement of Cor.\ \ref{iiskkkp}(iv). Non-emptiness of $I^-(x)=\,\twoheaddownarrow x$ is the non-emptiness of directed sets, and the interpolation property of $I$ is that of the way-below relation of a
continuous poset, Thm.\ \ref{oovg}.

\emph{The clause on ($\star$).} Assume $\mathscr{T}$ coarser than the Alexandrov topology of $I$ and let $O\ni x$ be $\mathscr{T}$-open. Then $O$ contains an Alexandrov basic set $I^+(z)\cap I^-(w)$ containing $x$. Since $(z,x)\in I$ and $I$ is interpolating by the previous paragraph, there is $z'$ with $(z,z')\in I$ and $(z',x)\in I$; then $z'\in I^+(z)$ and $z'\in I^-(x)\subset I^-(w)$, so $z'\in O\cap I^-(x)$. As $O$ is arbitrary,
$x\in\overline{I^-(x)}$, which together with the non-emptiness of $I^-(x)$ is ($\star$). That the converse implication fails, and that ($\star$) is not forced, is shown in Rem.\ \ref{starnotnec}.
\qed \end{proof}

\begin{remark} \label{starnotnec}
One would expect ($\star$) to be the topological counterpart of the approximation condition $x=\sup \twoheaddownarrow x$ of continuous posets, and hence to be recovered from the continuity of $(X,D_p)$. It is not, and the reason is instructive. For the poset $(X,D_p)$ the relation $x=\sup I^-(x)$ is a tautology: a point $u$ is an upper bound of $I^-(x)$ iff
$I^-(I^-(x))\subset I^-(u)$, that is, $I$ being transitive and interpolating, iff $I^-(x)\subset I^-(u)$, which is precisely $(x,u)\in D_p$; so $x$ is the least upper bound of $I^-(x)$ whatever the topology does. What the continuity condition genuinely delivers is therefore not the topological approximation but its order-theoretic shadow, namely the non-emptiness of $I^-(x)$ together with the interpolation property of $I$ --- which is indeed what Prop.\ \ref{kbpo} extracts from ($\star$).

That no property of the poset can imply the topological half of ($\star$) is seen at once: $(X,D_p)$, its way-below relation and its continuity depend on $I$ alone, whereas $x\in\overline{I^-(x)}$ relates $I$ to $\mathscr{T}$. Refining $\mathscr{T}$ to any finer Hausdorff topology in which the sets $I^+(x)$ are still open --- the discrete topology, for instance --- leaves the poset and its continuity untouched while destroying $x\in \overline{I^-(x)}$, as soon as $I$ is irreflexive.
\end{remark}


\begin{theorem}[The way-below relation as a fixed point] \label{fixedkp}
Let $X$ be a past-distinguishing lower topological KP-space satisfying ($\star$) and {\em (SC)}.
\begin{itemize}
\item[(i)] $(X,D_p)$ is a continuous poset whose way-below relation is the future-reflective hull of $I$: $\ll\,=C=I\circ D_p$.
\item[(ii)] $(X,D_p,C,\mathscr{T})$ satisfies the same hypotheses as $X$, induces the same order $D_p$ and the same way-below relation, and $C$ is its own hull.
\item[(iii)] The following conditions are equivalent: (a) $I$ is a fixed point of the hull, $C=I$; (b) $I\circ D_p\subset I$, i.e.\ future reflectivity; (c) $\ll\,=I$.
\end{itemize}
\end{theorem}

\begin{proof}
(i) is Thm.\ \ref{contkp} together with Lemma \ref{comfutkp}. (ii) By Prop.\ \ref{hullkp}, $(X,D_p,C,\mathscr{T})$ is a past-distinguishing lower topological KP-space satisfying ($\star$) and {\em (SC)}, with the same topology and the same order, and its hull is $C$. Applying (i) to it, the way-below relation of $(X,D_p)$ is the hull of $C$, namely $C$. (iii) The equivalence (a)$\Leftrightarrow$(b) is Prop.\ \ref{hullkp}(e) and Remark \ref{jfree}, while (a)$\Leftrightarrow$(c) follows from (i).
\qed \end{proof}

Theorem \ref{fixedkp} clarifies the role of continuity. All the hypotheses of Thm.\ \ref{ksswkp} but future reflectivity are preserved by the hull $I\mapsto C$, and they already imply the continuity of $(X,D_p)$. The way-below relation of the continuous poset generated by $I$ is the hull $C$, which is a fixed point of the construction, and the identification $\ll\,=I$ is precisely the fixed-point equation $I=I\circ D_p$. Equivalently, Thm.\ \ref{contkp} can be recovered by applying the future reflecting case, Thm.\ \ref{ksswkp}, to the hull. Since $C=I\circ D_p$ and $D_p$ are defined from $I$ alone, the topology enters only through ($\star$), which guarantees that $I$ is interpolating and that $C^+(x)=\operatorname{Int}D_p^+(x)$; a purely relational version is given in Cor.\ \ref{wbfixed}.

It remains to record when {\em (SC)} holds. In the regular case it follows from strong causality or from past reflectivity through Prop.\ \ref{vpoi}. Both routes survive here, the second in the form of the past-reflectivity variant of Prop.\ \ref{vpoj}, which does not pass through stable causality and is therefore available in the present setting. Both routes use the openness of the chronological pasts $I^-(x)$, so they will be stated for lower topological KP-spaces having this additional property. Under this openness, separability of $\mathscr{T}$ implies chronological separability, as required by Props.\ \ref{jjil} and \ref{vpoj}. 

\begin{theorem} \label{sckp}
Let $X$ be a strongly causal lower topological KP-space with open chronological pasts $I^-(x)$, satisfying ($\star$), ($\star^{\rm op}$) and ($\diamond$), whose topology is separable and locally compact,  and in which any two points $(x,y)\in I$ are connected by an isocausal curve. Then $X$ satisfies {\em (SC)}.
\end{theorem}

Observe that ($\diamond$) is used in conjunction with strong causality, which makes it easy to satisfy. For instance, it holds whenever $J$ is closed on a neighbourhood of each point, in particular for KS- and BMS-Lorentzian length spaces.

\begin{proof}
Past-distinction follows from Prop.\ \ref{scdist}. Let $S$ be directed with $p=\sup S$. By Prop.\ \ref{jjil} there is a countable timelike chain $C$ with $I^-(S)=I^-(C)$, and by the penultimate statement of Prop.\ \ref{vpoj} the chain $C$ has a future endpoint, which by ($\diamond$) is an upper bound; Eq.\ (\ref{nnnn}) then gives
$I^-(S)=I^-(C)=I^-(p)$.
\qed \end{proof}

\begin{corollary} \label{ksswsc}
Let $X$ be a strongly causal future reflecting lower topological KP-space with open chronological pasts $I^-(x)$, satisfying ($\star$), ($\star^{\rm op}$) and ($\diamond$), whose topology is separable and locally compact, and in which any two points $(x,y)\in I$ are connected by an isocausal curve. Then $(X,D_p)$ is a continuous poset and $I$ is its way-below relation. If moreover $J^-(x)$ is closed for every $x$, then $D_p=J$ and the statement holds for the poset $(X,J)$.
\end{corollary}

\begin{proof}
By Thm.\ \ref{sckp} the space satisfies {\em (SC)}, so Thm.\ \ref{ksswkp} applies.
\qed \end{proof}

The second route dispenses with strong causality, with ($\star^{\rm op}$) and with ($\diamond$), the last because past reflectivity makes it global.

%

\begin{theorem} \label{reflsc}
Let $X$ be a past-distinguishing past-reflecting lower topological KP-space with open chronological pasts $I^-(x)$, satisfying ($\star$), whose topology is separable and locally compact, in which any two points $(x,y)\in I$ are connected by an isocausal curve. Then $X$ satisfies {\em (SC)}. Dually, a future-distinguishing, future-reflecting upper topological KP-space with open chronological futures $I^+(x)$, satisfying ($\star^{\rm op}$) and the same topological and connectedness assumptions, satisfies {\em (IC)}.
\end{theorem}

\begin{proof}
First observe that past reflectivity, that is $D_f\subset D_p$, implies ($\diamond$) with $O_x=X$ for every $x$ (indeed, by the dual of the argument in the proof of Prop.\ \ref{ppcb}, which uses the openness of the sets $I^-(z)$, $x\in \overline{I^+(y)}$ implies $I^+(x)\subset I^+(y)$, hence $I^-(y)\subset I^-(x)$ and $y\in\overline{I^-(y)}\subset \overline{I^-(x)}$). Similarly, future reflectivity implies the dual condition ($\diamond^{\rm op}$).

Let $S$ be directed with $p=\sup S$ and let $C$ be the countable timelike chain of Prop.\ \ref{jjil}, so that $I^-(S)=I^-(C)$. By the past-reflectivity variant of Prop.\ \ref{vpoj} the chain $C$ has a future endpoint, and, as observed, condition ($\diamond$) holds, so that endpoint is an upper bound; Eq.\ (\ref{nnnn}) then gives $I^-(S)=I^-(C)=I^-(p)$.
\qed \end{proof}

\begin{corollary} \label{ksswrefl}
Let $X$ be a past-distinguishing  reflecting lower topological KP-space with open chronological pasts $I^-(x)$, satisfying ($\star$), whose topology is separable and locally compact, in which any two points $(x,y)\in I$ are connected by an isocausal curve. Then $(X,D_p)$ is a continuous poset and $I$ is its way-below relation. If moreover $J^-(x)$ is closed for every $x$, then $D_p=J$ and the statement holds for the poset $(X,J)$.
\end{corollary}

\begin{proof}
Past reflectivity gives {\em (SC)} by Thm.\ \ref{reflsc}, and future reflectivity is among the hypotheses of Thm.\ \ref{ksswkp}, which therefore applies.
\qed \end{proof}

%
%

\begin{theorem} \label{bicontkp}
Let $X$ be a topological KP-space satisfying ($\star\star$), which is distinguishing, reflecting, and satisfies {\em (SC)} and {\em (IC)}. Then $(X,D)$ is a bicontinuous poset with $\ll\,=\overline{\ll}\,=I$, where $D=D_p=D_f$. If moreover $J^{\pm}(x)$ are closed for every $x$, then $D=J$ by Prop.\ \ref{ppcb} and the statement holds for the poset $(X,J)$.

Conversely, let $X$ be a topological KP-space such that $(X,\le)$ is a bicontinuous poset  with $\ll\,=\overline{\ll}\,=I$. Then necessarily $\le\,=D_p=D_f=D$, the space is distinguishing and reflecting, {\em (SC)} and {\em (IC)} hold, $I^{\pm}(x)$ are non-empty for every $x$ and $I$ is interpolating. Condition ($\star\star$) is not forced, for the reasons given in Rem.\ \ref{starnotnec}.
\end{theorem}

\begin{proof}
Reflectivity is $D_p=D_f$, so $D_p=D_f=D$. Past-distinction, future reflectivity and {\em (SC)} put us in the hypotheses of Thm.\ \ref{ksswkp}, so $(X,D)$ is continuous with $\ll\,=I$; future-distinction, past reflectivity and {\em (IC)} put us in those of its time dual, so $(X,D)$ is co-continuous with $\overline{\ll}\,=I$. Hence $(X,D)$ is weakly
bicontinuous, and $\ll\,=\overline{\ll}$ because both equal $I$, so it is bicontinuous.

For the converse, apply the converse half of Thm.\ \ref{ksswkp} and its time dual: the first gives $\le\,=D_p$, past-distinction, future reflectivity and {\em (SC)}, the second gives $\le\,=D_f$, future-distinction, past reflectivity and {\em (IC)}. Hence $D_p=\,\le\,=D_f$, which is reflectivity, and the two distinction conditions give distinction.
\qed \end{proof}

\begin{corollary} \label{bicontcc}
Let $X$ be a topological KP-space satisfying ($\star\star$), whose topology is separable and locally compact, and in which any two points $(x,y)\in I$ are connected by an isocausal curve. Then the following are equivalent:
\begin{enumerate}
  \item[(i)] there is a partial order $\le$ on $X$ making $(X,\le)$ a bicontinuous poset   with $\ll\,=\overline{\ll}\,=I$;
  \item[(ii)] $X$ is distinguishing and reflecting, i.e.\ causally continuous.
\end{enumerate}
In this case $\le\,=D_p=D_f=D$, and if moreover $J^{\pm}(x)$ are closed for every $x$ then $D=J$ by Prop.\ \ref{ppcb} and the statement holds for the poset $(X,J)$.
\end{corollary}

\begin{proof}
Under these hypotheses Thm.\ \ref{reflsc} shows that past reflectivity implies {\em (SC)} and future reflectivity implies {\em (IC)}, so (ii) implies the hypotheses of Thm.\ \ref{bicontkp}. The converse is the converse half of that theorem.
\qed \end{proof}

\begin{remark} \label{bicontrem}
Conditions {\em (SC)} and {\em (IC)}, which have to be imposed at the continuous and co-continuous rungs by Thm.\ \ref{ksswkp} and its dual, disappear at the bicontinuous rung if minimal topological and connectedness assumptions are imposed: each is supplied by the reflectivity of the opposite sense. This mirrors the regular case, where by Thm.\ \ref{mmwr} past reflectivity implies {\em (SC)}, and it is the reason why the
characterization (ii) involves only causality conditions.
\end{remark}

\subsection{KS-Lorentzian pre-length spaces} \label{ksll}

We recall \cite[Def.\ 2.8, 2.18, 3.1, 3.4]{kunzinger18}
\begin{definition}
A quintuple $(X,J,I,d,\rho)$ where $(X,J,I)$ is a KP-causal space,\footnote{Kunzinger and S\"amann do not really impose antisymmetry of $J$ but, with some abuse, we do it here to simplify the connection with posets. Also our terminology matches the intended meaning of {\em causal} by Kronheimer and Penrose.} $\rho$ is a metric on $X$ and $d:X\times X\to [0,+\infty]$ is a lower semi-continuous map (with respect to the metric topology induced by $\rho$) that satisfies the reverse triangle inequality for pairs in $J$ and $d(x, y) = 0$ if $(x,y)\notin J$ and $d(x, y) > 0$ iff  $(x,y)\in I$, is called a {\em KS-Lorentzian pre-length space}.

An isochronal/isocausal curve $\sigma: [0,1]\to X$ is a ($\rho$-)locally Lipschitz curve which is isochronal/isocausal: $(\sigma(t),\sigma(t'))\in I$ (resp. $J$) for $t<t'$.

It is called {\em causally path connected} if for all $x, y \in X$ with $(x,y)\in I$ there is an isochronal curve from $x$ to $y$ and for $(x,y)\in J$ there is an isocausal curve from $x$ to $y$.

It is called {\em locally causally closed} if every point has a closed neighborhood $B$ such that $J\vert_{B\times B}$ is closed.
\end{definition}

Observe that the topology, being induced by a metric, is Hausdorff and first-countable. However, it is not, a priori, locally compact and separable.

Observe also that the notion of causal path connectedness is quite strong compared with ``connectedness by isocausal curves of chronologically related points'' which we need to get continuity of $(X,D_p)$. The latter condition matches instead nicely the pre-length condition of BMS-Lorentzian metric spaces (Sec.\ \ref{mwcf}).

For a study of the causal ladder for these spaces see \cite{ake20}.

\begin{theorem} \label{kcct}
Let $(X,J,I,\rho,d)$ be a locally causally closed, causally path connected KS-Lorentzian pre-length space (in particular a KS-LLS)  which has separable and locally compact topology, which is  future reflecting and which satisfies $(\star)$ and one of the following
\begin{itemize}
\item[-]  strong causality,  and   $(\star^{\textrm{op}})$,
\item[-] past-distinction  past reflectivity,
\end{itemize}
Then $(X,D_p)$ is a continuous poset and $I$ is the way-below relation of $(X,D_p)$.
\end{theorem}

%
%
%

\begin{proof}
The space $(X,J,I,\mathscr{T}_\rho)$ is a topological KP-space: $(X,J,I)$ is a KP-causal space by hypothesis, the metric topology $\mathscr{T}_\rho$ is Hausdorff, and $I$ is open because $d$ is lower semi-continuous and $I=\{d>0\}$ (so the sets $I^{\pm}(x)$ are open). It is first-countable, and separable and locally compact by hypothesis.

\emph{Condition ($\diamond$).} Let $x\in X$ and let $B$ be a closed neighbourhood of $x$ with $J|_{B\times B}$ closed, which exists by local causal closedness; put $O_x=\operatorname{int}B$. Let $y\in O_x$ with $x\in \overline{I^+(y)}$, and let $z_n\in I^+(y)$ with $z_n\to x$. For $n$ large $z_n\in B$, and $y\in B$, so $(y,z_n)\in J|_{B\times B}$; by closedness $(y,x)\in J$, that is $y\in J^-(x)$, and $J^-(x)\subset \overline{I^-(x)}$ by Prop.\ \ref{ppcb}, which uses ($\star$). Hence $y\in \overline{I^-(x)}$, which is ($\diamond$).

\emph{Isocausal connectedness.} Causal path connectedness gives, for $(x,y)\in I\subset J$, an isocausal curve from $x$ to $y$; in fact an isochronal one, which is in particular isocausal. This is the hypothesis ``any two points $(x,y)\in I$ are connected by an isocausal curve'' used in Prop.\ \ref{vpoj}.

\emph{First alternative.} Assume strong causality and ($\star^{\rm op}$). By Prop.\ \ref{scdist} the space is distinguishing, in particular past-distinguishing, so $D_p$ is a partial order. All the hypotheses of Thm.\ \ref{sckp} are met --- ($\star$), ($\star^{\rm op}$), ($\diamond$), separability, local compactness, strong causality
and isocausal connectedness --- hence {\em (SC)} holds.

\emph{Second alternative.} Assume past reflectivity then condition ($\diamond$) holds globally, and all the hypotheses of Thm.\ \ref{reflsc} are met, hence {\em (SC)} holds.

In either case the space is past-distinguishing, satisfies ($\star$), is future-reflecting and satisfies {\em (SC)}, so Thm.\ \ref{ksswkp} applies: $(X,D_p)$ is a continuous poset and $I$ is its way-below relation.
\qed \end{proof}

\begin{remark}Future reflectivity has been used only in the last step. Without it, {\em (SC)} still holds by the same proof, and Thm.\ \ref{fixedkp} shows that $(X,D_p)$ is a continuous poset whose way-below relation is the future-reflective hull $C=I\circ D_p$.
\end{remark}

\begin{theorem} \label{kcctbi}
Let $(X,J,I,\rho,d)$ be a KS-Lorentzian pre-length space with separable and locally compact topology, satisfying $(\star\star)$, and such that any two chronologically related points are connected by an isocausal curve. Then the following are equivalent:
\begin{enumerate}
  \item[(i)] there is a partial order $\le$ on $X$ making $(X,\le)$ a bicontinuous poset   with $\ll\,=\overline{\ll}\,=I$;
  \item[(ii)] $X$ is causally continuous, that is distinguishing and reflecting.
\end{enumerate}
In this case $\le\,=D_p=D_f=D$, and if moreover $J^{\pm}(x)$ are closed for every $x$, then $D=J$ and the statement holds for the poset $(X,J)$. In particular the equivalence holds for every causally path connected KS-Lorentzian length space with separable and locally compact topology satisfying $(\star\star)$, local causal closedness being unnecessary.
\end{theorem}

\begin{proof}
As in the proof of Thm.\ \ref{kcct}, the quadruple $(X,J,I,\mathscr{T}_\rho)$ is a topological KP-space, first-countable, separable and locally compact, and causal path connectedness provides an isocausal curve between any two points $(x,y)\in I$. The hypotheses of Cor.\ \ref{bicontcc} are therefore met, and the equivalence is that corollary. For the last statement, reflectivity gives $D_p=D_f=D$, and if the sets $J^{\pm}(x)$ are closed then $D_p=J$ by Prop.\ \ref{ppcb}, whence $D=J$.
\qed \end{proof}

\subsection{Equivalence of past indecomposable sets and round-ideals}
The following observation will be used only when discussing completions.
A subset $S$ is a {\em past set}\footnote{This is the terminology of \cite{geroch72} which is different from that of \cite{hawking73,minguzzi18b} which call past a set such that $I^-(S)\subset S$.} if $I^-(S)=S$. A non-empty past set $P$ is {\em indecomposable} (an IP) if it cannot be written as the union of two past sets which are proper subsets of $P$. For regular spacetimes Geroch, Kronheimer and Penrose proved that the indecomposable past sets are chronological pasts of timelike curves \cite{geroch72}. An indecomposable past set of the form $I^-(x)$, $x\in X$, is called a {\em proper indecomposable past set} (PIP); the remaining indecomposable past sets are called {\em terminal} (TIPs).

%

We recall that an ideal of a poset $(X,\le)$ is a directed lower set. A {\em round-ideal} is an ideal $A$ such that for every $x\in A$ we can find $y\in A$ with $x\in\, \twoheaddownarrow y$. In particular, it is Scott open as the union of sets of the form $\twoheaddownarrow y$. A (round-)ideal is {\em primitive} if it is of the form $\twoheaddownarrow x$ for some $x\in X$.

\begin{theorem} \label{ipeqv}
Let $X$ be a lower topological KP-space satisfying the hypotheses of Thm.\ \ref{ksswkp} --- that is, past-distinguishing and future-reflecting, satisfying ($\star$) and {\em (SC)} --- and assume in addition that $X$ is chronologically separable. Then $(X,D_p)$ is a continuous poset with $\ll\,=I$ and the following families of sets coincide
(the notion of ideal and the other order-theoretic notions refer to the poset $(X,D_p)$)
\begin{itemize}
\item[(i)] Indecomposable past sets,
\item[(ii)] $I^-(C)$ where $C$ is a non-empty countable timelike chain,
\item[(iii)] $I^-(S)$ where $S$ is a directed set,
\item[(iv)] Scott-open ideals,
\item[(v)] round-ideals.
\end{itemize}
A set in (i) is a PIP if and only if, regarded as a set in (iv) or (v), it is primitive. Chronological separability is needed only for family (ii).
\end{theorem}

%

\begin{proof}
(ii) $\Rightarrow$ (i). Put $P=I^-(C)$. Then $P$ is a non-empty past set and is $I$-directed by the first paragraph of the proof of Prop.\ \ref{jjil}, since the non-empty chain $C$ is directed. Suppose that $P=P_1\cup P_2$ for two proper past subsets $P_1,P_2$. Choose $x\in P\setminus P_1$ and $y\in P\setminus P_2$. There is $z\in P$ with $(x,z),(y,z)\in I$. If $z\in P_1$, then $x\in I^-(P_1)=P_1$, a contradiction; if $z\in P_2$, then $y\in I^-(P_2)=P_2$, another contradiction. Thus $P$ is indecomposable.


(i) $\Rightarrow$ (iii). If $P$ is a past set, $P=I^-(P)$, thus we need only to show that $P$ is directed. Suppose it is not; then there are $x,y\in P$ with no $z\in P$ such that $(x,z)\in D_p$ and $(y,z)\in D_p$. In particular, no $c\in P$ satisfies $x,y\in I^-(c)$, for such a $c$ would be a common $D_p$-upper bound of $x$ and $y$ in $P$, as $I\subset D_p$. Let
\[
V=\bigcup\{I^-(c):\ c\in P,\ x\notin I^-(c)\}, \qquad U=\bigcup\{I^-(c):\ c\in P,\ y\notin I^-(c)\}.
\]
Both $U$ and $V$ are past sets contained in $P$, and $P=U\cup V$ because every $p\in P=I^-(P)$ lies in $I^-(c)$ for some $c\in P$, and $I^-(c)$ misses $x$ or $y$. Moreover $x\notin V$ and $y\notin U$, so $U$ and $V$ are proper subsets of $P$, contradicting indecomposability. Hence $P$ is directed.


(iii) $\Rightarrow$ (ii). This is Prop.\ \ref{jjil}, which requires only ($\star$) and chronological separability.


(ii) $\Rightarrow$ (v). Since $C$ is directed $I^-(C)$ is directed. Also if $y\in I^-(C)$ and $x\ll y$ then $x\in I^-(C)$ because $\ll\,=I$. This shows that $I^-(C)$ is an ideal. Moreover, if $y\in I^-(C)$, there is $c\in C$ such that $y\in I^-(c)$, and since $I$ is interpolating by Prop.\ \ref{kbpo} there is $z$ with $y\ll z\ll c$, which implies
$z\in I^-(C)$ and proves that $I^-(C)$ is round.

(v) $\Rightarrow$ (iv).
If $S$ is a round-ideal, for every $x\in S$ there is $y\in S$ such that $x\in \twoheaddownarrow y$, thus $S=\cup_{s\in S} \twoheaddownarrow s$. But this is the characteristic expression of a Scott-open ideal, Thm.\ \ref{nncf}.

(iv) $\Rightarrow$ (iii). If $S$ is a Scott-open ideal, as it is Scott-open  $S=\cup_{s\in S}\twoheaddownarrow s = \cup_{s\in S} I^-(s)=I^-(S)$ (Thm.\ \ref{nncf}), and since $S$ is an ideal, it is directed.

The last statement is clear from the identity $\twoheaddownarrow x=I^-(x)$ and the fact that if a set belongs to one family it belongs to all families.

Chronological separability has been used only in (iii) $\Rightarrow$ (ii). Without it, (iii) $\Rightarrow$ (v) follows directly: for $S$ directed, $I^-(S)$ is directed by the first paragraph of the proof of Prop.\ \ref{jjil}, which does not use separability, and the argument of (ii) $\Rightarrow$ (v) applies with $S$ in place of $C$.
\qed \end{proof}

Future reflectivity can be dropped by passing to the hull. We call {\em $C$-past} a set $P$ such that $C^-(P)=P$, and define $C$-indecomposable past sets and $C$-chains as above with $C$ in place of $I$.

\begin{corollary} \label{ipeqvC}
Let $X$ be a past-distinguishing lower topological KP-space satisfying ($\star$) and {\em (SC)}. Then $(X,D_p)$ is a continuous poset with $\ll\,=C$, and the families (i), (iii)-(v) of Thm.\ \ref{ipeqv} coincide once $I$ is replaced by $C$: $C$-indecomposable past sets, sets $C^-(S)$ with $S$ directed, Scott-open ideals, round-ideals. If $X$ is chronologically separable they also coincide with the sets $C^-(E)$ with $E$ a non-empty countable $C$-chain.
\end{corollary}

\begin{proof}
By Thm.\ \ref{fixedkp}, $(X,D_p,C,\mathscr{T})$ satisfies the hypotheses of Thm.\ \ref{ksswkp}, with the same order $D_p$, and it is chronologically separable if $X$ is, by Prop.\ \ref{hullkp}(g), so Thm.\ \ref{ipeqv} applies to it.
\qed \end{proof}

The $C$-past sets are not new objects: they are in canonical correspondence with the past sets of $I$.

\begin{lemma} \label{pastiso}
Let $X$ be a lower topological KP-space satisfying ($\star$). The maps
\[
\Phi(P):=D_p^-(P),\qquad \Psi(Q):=I^-(Q)
\]
are mutually inverse inclusion-preserving bijections between the past sets ($I^-(P)=P$) and the $C$-past sets ($C^-(Q)=Q$). They preserve unions, hence indecomposability, and $\Phi(P)=C^-(P)$, $\Phi(I^-(x))=C^-(x)$.  The map $\Phi$ is the identity iff $X$ is future reflecting.
\end{lemma}

\begin{proof}
We use $C^-(A)=D_p^-(I^-(A))$ and $I^-(I^-(A))=I^-(A)$ (Lemma \ref{comfutkp}). Let $P$ be a past set. Then $\Phi(P)=D_p^-(I^-(P))=C^-(P)$, in particular $\Phi(I^-(x))=C^-(x)$. If $(x,y)\in I$ and $(y,p)\in D_p$ with $p\in P$ then $x\in I^-(y)\subset I^-(p)\subset P$, so $I^-(D_p^-(P))\subset P$, while $P=I^-(P)\subset I^-(D_p^-(P))$; thus $\Psi\Phi(P)=P$. It follows that $C^-(\Phi(P))=D_p^-(I^-(\Phi(P)))=D_p^-(P)=\Phi(P)$, so $\Phi(P)$ is $C$-past. If $Q$ is $C$-past, $I^-(\Psi(Q))=\Psi(Q)$ and $\Phi\Psi(Q)=D_p^-(I^-(Q))=C^-(Q)=Q$. Both maps are monotone and preserve unions, so they preserve indecomposability. Under future reflectivity $C=I$ (Prop.\ \ref{frcomf}), so $\Phi(P)=C^-(P)=I^-(P)=P$. Conversely, if $\Phi$ is the identity, $\Phi(P):=D_p^-(P)$ applied to $P=I^-(x)$ gives $I^-(x)=D_p^-(I^-(x))=C^-(x)$, namely $C=I$ which is future reflectivity by Prop.\ \ref{frcomf}.
\qed \end{proof}

\section{Completions and spacetime boundaries} \label{cmpl}

In the regular theory in a globally hyperbolic spacetime $(M,g)$ every upper-bounded future-directed timelike curve has an endpoint.
When the curve is inextendible, the missing suprema are naturally interpreted as ``points at infinity'' or ``singular points'' that ought to be attached to $M$. Constructing such ideal points from the causal structure alone, without any auxiliary conformal embedding, is the strategy of the causal boundary of Geroch, Kronheimer and Penrose \cite{geroch72}, later refined by several authors \cite{budic74,harris98,marolf03,flores06b,flores11}.

Our study suggests a different strategy to approach spacetime completions.

In Section \ref{qcft} we have seen that several distinguished causal relations on $(M,g)$ turn $M$ into a continuous, or bicontinuous, poset.
The result has been generalized to topological KP-spaces in Section \ref{rjgh}, a sufficiently general framework that comprises the various definitions of Lorentzian length spaces that have been proposed in the literature.

We recalled that in a globally hyperbolic poset every upper-bounded directed set  has a supremum, Thm.\ \ref{oftr}. Consequently, directed sets are analogous to timelike curves and suprema are analogous to endpoints. The first idea could be to use ideals to build a completion. In fact, a few completion schemes have been developed in domain theory using variations of this strategy. We shall explore the most important ones below. The interpretation of spacetime as a continuous posets will allow us to apply them verbatim.

\subsection{Completion schemes}

Let us consider the  most pertinent completion schemes for the continuous poset approach to spacetimes. An approach, not related to continuous posets, but valid under low regularity was proposed in \cite{burgos23,ake26}. It would be interesting to clarify if it has connections with those developed in this section.

\subsubsection{Ideal completion}
The most basic completion of a poset $(X,\le)$ is obtained by adjoining, formally, the missing suprema of all directed subsets.

\begin{definition}[Ideal completion, \cite{abramsky94}, Def.\ 2.2.21] \label{idealcompl}
Let $(X,\le)$ be a poset. The set $\hat X$ of all ideals of $X$ (recall: non-empty directed lower sets), ordered by inclusion, is a directed complete poset (dcpo), called the {\em ideal completion} of $X$. The map $j:X\to \hat X$, $j(x)=\,\downarrow x$, is an order embedding, and every ideal is the directed supremum, in $\hat X$, of the principal ideals $\downarrow x$, $x\in X$, it contains.
\end{definition}

The ideal completion is characterized by a universal property: it is the free dcpo generated by $X$, in the sense that every monotone map from $X$ to a dcpo extends uniquely to a Scott-continuous map from $\hat X$ \cite[Prop.\ 2.2.24]{abramsky94}.

\subsubsection{Directed completion}


Working in the synthetic setting of Lorentzian pre-length spaces, Gigli \cite[Sec.\ 2]{gigli25} and Gigli et al.\ \cite{gigli26} introduce a {\em directed completion}, characterized by a universal property with respect to suprema-{\em preserving}, i.e.\ Scott continuous, maps. As they acknowledge, it is essentially the same completion studied by Markowsky \cite{markowsky76}, and Zhao and Fan \cite{zhao10}. It also appears in Keimel and Lawson \cite[Thm.\ 7.4]{keimel09}.
This construction is close in spirit to, but not literally an instance of, the ideal completion of Def.\ \ref{idealcompl}.


Observe that with the inclusion order on $\mathcal{P}(X)$,   the Scott continuity of the map $x\mapsto I^-(x)$ reads precisely as the property (SC)
\begin{quote}
for every directed $S\subset X$ admitting a supremum,
$I^-(S)=I^-(\sup S)$.
\end{quote}
They prove that, under (SC) together with approximability and past-distinction, the directed completion coincides with the future causal completion of Geroch, Kronheimer and Penrose, and that (SC) holds automatically for smooth globally hyperbolic spacetimes (via forward completeness of the metric, itself implied by compactness of causal diamonds). They illustrate the result with an explicit description of the boundary of the maximally extended Schwarzschild spacetime. Proposition \ref{scgigli} below shows that (SC) also follows from some hypotheses established to get the continuous poset property, by a route independent of theirs. Their identification of the completion is recovered, under weaker assumptions, in Thm.\ \ref{ipmatch} and the remark following it.

The following is a reformulation of our Thms.\  \ref{sckp} and \ref{reflsc} which weaken considerably the assumptions needed to get (SC).

\begin{proposition}[Deriving (SC)]
\label{scgigli}
Let $X$ be a lower topological KP-space with open chronological pasts $I^-(x)$, satisfying lower topological approximation ($\star$), whose topology is separable and locally compact, and in which any two chronologically related points are connected by an isocausal curve. Assume one of
\begin{itemize}
\item[-] strong causality, ($\star^{\textrm{op}}$) and pointed local past-reflectivity ($\diamond$);
\item[-] past-distinction and past reflectivity.
\end{itemize}
Then $(X,D_p)$ is a continuous poset satisfying the supremum-compatibility condition (SC) and $I\subset \, \ll\,=C=I\circ D_p$.
\end{proposition}
%
%

No form of future reflectivity is required here: we used that hypothesis in Sec.\ \ref{rjgh} only to upgrade $I\subset\,\ll$ to the full identity $\ll\,=I$, a logically independent question from (SC). Continuity of $(X,D_p)$ and the inclusion $I\subset\,\ll$ follow from Thm.\ \ref{contkp}, past-distinction being a consequence of strong causality and ($\star$) by Prop.\ \ref{scdist} in the first alternative. Future reflectivity becomes necessary only for the identification $\ll\,=I$ (Cor.\ \ref{iiskkkp}(iii)), that is, for $I$ to be a fixed point of the future-reflective hull (Thm.\ \ref{fixedkp}).


In the context of KS-Lorentzian length spaces  ($\diamond$) can be derived from local causal closure as we did in Thm.\ \ref{kcct}.
For a smooth spacetime $(M,g)$  $\diamond$ follows from strong causality and closure of the causal relation in globally hyperbolic neighborhoods.



The following corollary improves the causality conditions from the main result in \cite{gigli26} which is used in the last step once (SC) is proved.
\begin{corollary} \label{kcvt}
Let $(X,J,I,\rho,d)$ be a KS-Lorentzian pre-length space (in particular a KS-LLS) which is locally causally closed, in which any two chronologically related points are connected by an isocausal curve (in particular causally path connected), whose topology is separable and locally compact, and which satisfies
\begin{itemize}
\item[-]  two-sided approximation: $I^\pm(x)$ are non-empty and $x\in \overline{I^-(x)}\cap \overline{I^+(x)}$ for every $x\in X$, i.e.\ ($\star\star$),
\end{itemize}
together with one of
\begin{itemize}
\item[-] strong causality,
\item[-] past-distinction and past reflectivity.
\end{itemize}
Then $(X,D_p)$ satisfies {\em (SC)}, and the directed completion \cite{gigli25,gigli26} of the poset $(X,D_p)$ --- or of $(X,J)$ if the causal pasts of points are closed --- coincides with the GKP future completion.
\end{corollary}

\begin{proof}
The proof of (SC) goes as in the proof of Thm.\ \ref{kcct}.
%
%
%

Two-sided approximation is the approximability required in \cite[Thm.\ 3.15]{gigli26}, and past-distinction holds so that theorem applies and identifies the directed completion of $(X,D_p)$ with the GKP future completion. Finally, if $J^-(x)$ is closed for every $x$ then $D_p=J$ by Prop.\ \ref{ppcb}, and the statement holds for the poset
$(X,J)$.
\qed \end{proof}

\subsubsection{Round-ideal completion}

For our purposes the plain ideal completion is too coarse: applied to a continuous poset it typically destroys continuity, because a generic ideal need not be ``approximated from within'' in the sense of the way-below relation. The refinement that repairs this, and which is  the one relevant for our  approach that regards spacetimes as continuous posets, is due to Lawson.

We recall that a {\em round-ideal} is an ideal $A$ such that for every $x\in A$ we can find $y\in A$ with $x\ll y$.

\begin{definition}[Round-ideal completion, \cite{lawson97}] \label{roundidealdef}
Let $(X,\le)$ be a continuous poset with way-below relation $\ll$.
The set $\mathrm{RI}(X)$ of round ideals of $X$, ordered by inclusion, is the {\em round ideal completion} of $X$.
\end{definition}

By the interpolation property established in Thm.\ \ref{oovg}, every primitive ideal $\twoheaddownarrow x$, $x\in X$, is automatically round, so the embedding $j:X\to \mathrm{RI}(X)$, $j(x)=\, \twoheaddownarrow x$, is again available.

The modern way to arrive at this completion is via the notion of abstract basis. Let us illustrate this path as it shall provide several details on the properties of the completion.

Let $(X, \prec)$ be a set with a binary relation. For a subset $F \subset X$ and an element $y \in X$, we denote $F \prec y$ whenever $x \prec y$ holds for every $x \in F$. The relation $\prec$ is said to be {\em fully transitive} if it is transitive and obeys the {\em  interpolation property}:
\begin{equation} \label{mkcv}
  \forall \text{ finite } F \subset X, \quad F \prec z \implies \exists\, y \prec z \text{ such that } F \prec y.
\end{equation}
When $F = \emptyset$, this condition is understood as: for every $z \in X$, there is some $y \in X$ with $y \prec z$.

An {\em abstract basis} is defined as a pair $(B, \prec)$, where $B$ is a set and $\prec$ is a fully transitive relation on $B$.

\begin{proposition} \label{zzsd}
If $B$ is a basis of  a continuous poset $(X,\le)$ the pair $(B,\ll)$ is an abstract basis. In particular, $(X,\ll)$ is an abstract basis.
\end{proposition}

\begin{proof}
Let $F$ be a finite subset of $B$ and $F\ll z$, by   Theorem \ref{nxre} there is $w\in X$ such that $F\ll w$, $w\ll z$. But $\sup (B\cap \twoheaddownarrow z)=z$, thus by definition of $\ll$, there is $y\in B\,\cap \twoheaddownarrow z$, $w \le y$, thus $y\in B$, $F\ll y\ll z$.
\qed \end{proof}

Abstract bases need not come from a basis of the poset in the sense of the definition above: density is often enough. Let $(M,g)$ be globally hyperbolic and let $C\subset M$ be a countable dense subset. Then $(C,\ll)$, with $\ll\,=I$ restricted to $C$, is an abstract basis. Indeed, $I$ is transitive, and if $F\subset C$ is finite and $F\ll z$, $z\in C$, then $I^-(z)\cap\bigcap_{f\in F}I^+(f)$ is a non-empty open set (it is a neighbourhood of a point $y'$ with $F\ll y'\ll z$, which exists by interpolation), hence it contains some $y\in C$, so that $F\ll y\ll z$. This is the setting of Martin and Panangaden's reconstruction theorem \cite[Thm.\ 7.2]{martin06}: they complete the abstract basis $(C,\ll)$ in the sense of \cite[Def.\ 9.2]{martin06}, that is, they take its round-ideal completion (Definition below), within the domain of intervals of $M$, and recover $M$, with its manifold topology, as the set of maximal elements of the resulting domain. As

\begin{definition}
Suppose that $(B, \prec)$ is an abstract basis. A non-empty subset $I$ of $B$ is a \emph{round ideal} if
\begin{itemize}
\item[(i)] $y \in I$ and $x \prec y$ imply that $x \in I$,
\item[(ii)] $x, y \in I$ implies there exists $z \in I$ with $x \prec z$ and $y \prec z$.
\end{itemize}
The \emph{round ideal completion} is the partially ordered set $({RI}(B), \subset)$ consisting of all round ideals ordered by set inclusion.
\end{definition}
Observe that if $\mathcal{U}=\{U_\alpha, \alpha \in A\}$ is a family of round ideals which is directed with respect to inclusion, then its union $\bigcup_\alpha U_\alpha$ is again a round ideal: it is clearly non-empty and $\bigcup_\alpha U_\alpha$ is again a round-ideal: it is clearly non-empty and $U_\gamma\supset U_\alpha\cup U_\beta$ in $\mathcal{U}$, and (ii) applied in $U_\gamma$ provides $z$. Since it is an upper bound of $\mathcal{U}$ contained in every other upper bound, $\sup \mathcal{U}=\bigcup_\alpha U_\alpha$. As a consequence, directed suprema always exist in $({RI}(B),\subset)$, which is therefore a dcpo. For non-directed families the union need not satisfy (ii), and the supremum need not exist; e.g.\ for $(M,I)$ Minkowski spacetime and $p,q$ spacelike separated, $I^-(p)\cup I^-(q)$ is not a round ideal.


For a poset $(X,\le)$, not necessarily continuous, an abstract basis is provided by the pair $(X,\le)$ itself for which choice the round ideals coincide with the standard ideals and the round ideal completion becomes coincident with the ideal completion. Continuity of the poset allows one to make the choice $(X,\ll)$ for abstract basis, and it is this choice that will be generically referred to as round ideal completion of   $(X,\le)$ (in this case $j(y)={\twoheaddownarrow}y$ in the following result). The round ideal in this case is a subset $I$ of $X$ such that (i) $y \in I$ and $x \ll y$ imply that $x \in I$, and (ii) $x, y \in I$ implies there exists $z \in I$ with $x,y \ll z$.

The following is proved in \cite[Sec.\ 2.2]{abramsky94} \cite[Example 3.5]{lawson97} \cite{keimel09}

\begin{theorem} \label{roundidealprop}
Let $(B, \prec)$ be an abstract basis. Then the following hold:
\begin{enumerate}
  \item[\textup{(i)}] the round ideal completion $({RI}(B), \subset)$ is a continuous domain;
  \item[\textup{(ii)}] $I \ll J$ in ${RI}(B) \Leftrightarrow \exists\, y \in J$ such that $I \prec y$;
  \item[\textup{(iii)}] the mapping $j \colon B \to {RI}(B)$ defined by $j(y) = {\Downarrow}y := \{x : x \prec y\}$ has the property that $x \prec y$ implies $j(x) \ll j(y)$;
  \item[\textup{(iv)}] the set $j(B)$ is a basis for ${RI}(B)$.
\end{enumerate}
Additionally, if $(B,\prec)=(X,\ll)$ where $(X,\le)$ is a continuous poset, the properties are improved as follows
\begin{enumerate}
  \item[\textup{(i')}] the round ideal completion $({RI}(X), \subset)$ is a continuous domain;
  \item[\textup{(ii')}] $I \ll J$ in ${RI}(X) \Leftrightarrow \exists\, y \in J$ such that $I \ll y$ $\Leftrightarrow  \exists\, y \in J$ such that $I \subset \, \downarrow y$;
  \item[\textup{(iii')}] the mapping $j \colon X \to {RI}(X)$ defined by $j(y) = \,\twoheaddownarrow y$ is injective, strictly order preserving,\footnote{Defining $x<y$ if $x\le y$ and $x\ne y$, this means $x< y \Rightarrow j(x)\subsetneq j(y)$.} and has the property that $x \ll y$ if and only if $j(x) \ll j(y)$;
  \item[\textup{(iv')}] the set $j(X)$ is a basis for ${RI}(X)$;
  \item[\textup{(v')}] each round ideal of $(X,\le)$ is an ideal;
  \item[\textup{(vi')}] when $X$ and  ${RI}(X)$ are endowed with the respective Scott topologies the map $j$ is a topological embedding with Scott-dense image, a sobrification and a $D$-completion;
    \item[\textup{(vii')}]   the map $j$  preserves existing suprema: if $D$ is a directed subset that has supremum then $j(D)$ is also directed, has also supremum and
        \begin{equation} \label{ccnd}
        j(\sup D)=\sup j(D).
  \end{equation}
\end{enumerate}
\end{theorem}
We shall not enter into the notions of sobrification and $D$-completion.

Perhaps it is worth noting how to prove Eq.\ (\ref{ccnd}) as it is especially simple

\begin{proof}[Proof of point (vii')]
We want to prove $j(\sup D)=\cup_{d\in D} j(d)$. $\supset$: If $a\in \,\twoheaddownarrow d$ for some $d\in D$, then $a\ll d \le \sup D$, so $a\ll \sup D$ giving $a\in \, \twoheaddownarrow \sup D$. $\subset$: If $a\in \, \twoheaddownarrow \sup D$ by Thm.\  \ref{tncv} there is $d\in D$ such that $a\in   \, \twoheaddownarrow d$.
\qed \end{proof}

Some comments, particularly on the analogies with the Geroch, Kronheimer and Penrose completion \cite{geroch72}, are in order.
\begin{itemize}
\item[$(\alpha$)] In the GKP construction for abstract KP-causal spaces the authors pass to a property of {\em past semi-fullness} which is precisely the interpolation property of Eq.\ (\ref{mkcv}) for $I$. However, while in GKP this property is imposed, as $I$ and $J$ are unrelated, here it follows from the very definition of way-below relation and the continuity of the poset (Prop.\ \ref{zzsd}). This makes manifest the advantages of the domain theory approach to completions.
\item[$(\beta$)] Property (ii') determines the way-below relation in ${RI}(B)$, the partial order being just the inclusion. Nicely, point (ii') tells us that this is precisely the criterion adopted by GKP in their future completion \cite{geroch72} \cite[p.\ 623]{penrose79}.
\item[$(\gamma$)] The completion property (i') seems  more satisfactory than the property called future completeness by GKP in \cite[Thm.\ 3.7]{geroch72}.
\item[$(\delta$)] Property (iii') tells us that the way-below relation in the completed space, once restricted to the original space gives back the original way-below relation.
\item[$(\epsilon$)] The property  (iv') tells us that the completed space is still continuous as it admits a basis.
\end{itemize}
In fact, we shall see in a moment that these are not only analogies as the two completions coincide.

Let $(M,g)$ be past-distinguishing, satisfying (SC) and future-reflecting, so that, by Thm.\ \ref{kssw}, $(M,D_p)$ is a continuous poset with way-below relation $I$. Theorem \ref{roundidealprop} on the round ideal completion applies to it. The next result identifies the round ideal completion of $(M,D_p)$ with the classical construction of Geroch, Kronheimer and Penrose, as it is also immediate from Theorem \ref{ipeqv}. We prefer to state it in the more general framework of lower topological KP-spaces which goes beyond the originally intended application of the GKP construction. It is really a corollary of Thm.\  \ref{ipeqv}, Cor.\ \ref{ipeqvC} and Lemma \ref{pastiso}, and future reflectivity is not needed.

\begin{theorem} \label{ipmatch}
Let $X$ be a past-distinguishing lower topological KP-space satisfying ($\star$) and {\em (SC)}. Then $(X,D_p)$ is a continuous poset, and $P\mapsto D_p^-(P)$ is an order isomorphism from ${IP}(X)$ onto ${RI}(X,D_p)$ which maps $I^-(x)$ to $j(x)=\,\twoheaddownarrow x$; its inverse is $A\mapsto I^-(A)$. Consequently ${RI}(X,D_p)$ is order-isomorphic to ${IP}(X)$, and the points of ${RI}(X,D_p)\setminus j(X)$ correspond to the terminal indecomposable past sets (TIPs) of $X$, that is, the future causal boundary of $X$.

If $X$ is future reflecting the isomorphism is the identity: a subset of $X$ is a round-ideal of  $(X,D_p)$ if and only if it is an indecomposable past set (${IP}$), and, if $X$ is chronologically separable, if and only if it is of the form $I^-(C)$ where $C$ is a non-empty countable $I$-chain ($I^-(\gamma)$ for some, possibly future-inextendible, timelike curve $\gamma$ in the regular case).
\end{theorem}

Of course, the advantage of the round ideal completion approach to the spacetime completion/boundary is that one gets for free all the properties listed in Thm.\ \ref{roundidealprop}.

\begin{proof}
By Thm.\ \ref{fixedkp} the poset is continuous with $\ll\,=C$, and by Cor.\ \ref{ipeqvC} the round-ideals are the $C$-indecomposable past sets, the primitive ones being the sets $C^-(x)=\,\twoheaddownarrow x$. Lemma \ref{pastiso} gives the isomorphism. Under future reflectivity $C=I$ and Thm.\ \ref{ipeqv} identifies the round ideals.

Every non-empty countable timelike chain can be converted, in the regular case $(M,g)$, to a timelike curve passing through the points (conversely countable points can be extracted from the curve in such  a way that $I^-(C)=I^-(\gamma)$).
\qed \end{proof}

Without future reflectivity the round ideal attached to an indecomposable past set $P$ is not $P$ itself but its $D_p$-down-closure $D_p^-(P)$, from which $P$ is recovered as $I^-(D_p^-(P))$. The completion is nevertheless the future causal completion of $I$, with the same partial order and the same embedding of $X$. For a continuous poset the round ideal completion is also its universal dcpo-completion \cite[Prop.\ 9.3]{keimel09}, that is, the directed completion recalled above, so the round-ideal description of Thm.\ \ref{ipmatch} is compatible and actually recovers that of \cite[Thm.\ 3.15]{gigli26} with neither upper approximation, nor separability, nor future reflectivity. Future reflectivity is automatic when their result is applied to the causal relation $\le$ of a KS-Lorentzian pre-length space with closed causal pasts, since then $D_p=\,\le$ and $I\circ\le\,\subset I$ (Prop.\ \ref{ppcb}); in this case the round ideals are the indecomposable past sets themselves.

\begin{remark}[Two-sided completions]
Dually, ${RI}(M,D_f)$ realizes the indecomposable future sets ${IF}(M)$ and its added points are the terminal indecomposable future sets (TIFs), i.e.\ the past causal boundary. The full (two-sided) GKP causal boundary, however, is {\em not} simply the union, or a naive gluing along $M$, of these two one-sided completions: a boundary point of the complete construction is represented by a pair $(P,F)\in {IP}(M)\times {IF}(M)$ subject to a compatibility condition, a point of view
carefully investigated by several authors \cite{szabados88,marolf03,flores06b,flores11}.
One should not expect ${RI}(M,D_p)$ and ${RI}(M,D_f)$ to combine correctly unless some compatibility is built into the completion from the start. To our knowledge, a bicompletion for bicontinuous posets playing, with respect to the interval/biScott topology, the role that the round ideal completion plays for the Scott topology of a continuous poset has not been isolated in the domain-theory literature.
\end{remark}

\section{Way-below relation as a foundation of continuous posets}  \label{wbfound}

Throughout this work the order $\le$ has been the primitive datum and the way-below relation $\ll$ has been derived from it. In the spacetime application, however, the opposite is the case: what is physically given is the chronological relation $I$, while the order $D_p$ is reconstructed from it through $D_p=\{(x,y): I^-(x)\subset I^-(y)\}$, and it was found in Sec.\ \ref{qcft} that this is the only choice of $\le$ compatible with the identification $\ll\,=I$. It is therefore natural to reverse the logical order and to ask which conditions a relation $R$ on a set $X$ must satisfy in order that the construction
\begin{equation} \label{wbord}
x\le y \quad :\Longleftrightarrow \quad R^{-}(x)\subset R^{-}(y),
\end{equation}
where $R^{-}(x)=\{a\in X: (a,x)\in R\}$ and $R^{+}(x)=\{b\in X: (x,b)\in R\}$, returns a continuous poset whose way-below relation is precisely $R$. This section answers the question. The interest of the result is that it isolates, in purely relational terms, the whole content of the theory: a set endowed with a relation satisfying the conditions below is, to all effects, a continuous poset, and no order needs to be postulated.

Observe first that \eqref{wbord} always defines a preorder, so that only antisymmetry is at issue, and that one absorption law is automatic.

\begin{proposition} \label{wbauto}
For $\le$ defined by \eqref{wbord} one has $\le\circ\, R\,\subset R$, that is, $(x,z)\in R$ and $z\le y$ imply $(x,y)\in R$.
\end{proposition}

\begin{proof}
If $(x,z)\in R$ then $x\in R^{-}(z)\subset R^{-}(y)$.
\qed \end{proof}

In the spacetime case $R=I$ and $\le\,=D_p$ this is the inclusion $D_p\circ I\subset I$, which indeed holds in every spacetime, while the opposite absorption law $I\circ D_p\subset I$ is equivalent to future reflectivity by Prop.\ \ref{oovt}. The latter reappears below as condition (fR).

\begin{theorem} \label{wbaxioms}
Let $X$ be a set and $R$ a binary relation on $X$, and let $\le$ be defined by
\eqref{wbord}.
\begin{enumerate}
\item The binary relation $\le$ is a partial order, $(X,\le)$ is a continuous poset, and its way-below relation is $R$, if and only if the following conditions hold:
\begin{itemize}
\item[(T)] $R$ is transitive, that is
$R\circ R\subset R$;
\item[(D)] for every $x\in X$ the set $R^{-}(x)$ is non-empty and $R$-directed, namely $(a,x), (b,x)\in R$  imply that there is $c$ with $(a,c), (b,c)\in R$, and $(c,x)\in R$;
\item[(Bd)] the map $x\mapsto R^{-}(x)$ is injective;
\item[(fR)] $R\,\, \circ \le\,\, \subset R$, that is, $x\le z$ and $(z,y)\in R$ imply $(x,y)\in R$;
\item[(SC)] if $S\subset X$ is $\le$-directed and admits a supremum, then
$R^{-}(\sup S)=R^-(S)$.
\end{itemize}
\item  If only (T),(D),(Bd),(SC) hold, then $\le$ is still a partial order and $(X,\le)$ is still a continuous poset, but of its way-below relation we can only say $R\subset\, \ll$. In this case $\ll\,=R$ if and only if (fR) holds.
\end{enumerate}
\end{theorem}
The chosen letters have the following meaning: Transitivity, Directedness, Backward distinction, forward Reflectivity, and  Scott Continuity/Supremum Compatibility  condition. We have denoted, as usual,  $R^-(S):=\bigcup_{s\in S} R^{-}(s)$. Note that from (T) and (D) we have  $R=R\circ R$, just take $a=b$ in (D).
\begin{proof}
{\em Necessity.} Suppose $(X,\le)$ is a continuous poset with $\ll\,=R$. Condition (T) holds since $\ll$ is transitive.
For (D), the set
$\twoheaddownarrow x=R^{-}(x)$ is non-empty and directed, so given $a,b\ll x$ there is $c\ll x$ with $a,b\le c$; interpolating $c\ll d\ll x$ and using $a,b\le c\ll d$ we get $(a,d), (b,d)\in R$, and $(d,x) \in R$. Condition (Bd) is antisymmetry, since $x\le y$ and $y\le x$ mean $R^{-}(x)=R^{-}(y)$. Condition (fR) is the implication $x\le z\ll y\Rightarrow x\ll y$. For (SC), denoting with $u$ the supremum, the inclusion $\supset$ follows from
$s\le u$; conversely if $a\ll u$, interpolate $a\ll d\ll u$ and use $\sup S=u$ to find $s\in S$ with $d\le s$, whence $a\ll d\le s$ gives $(a,s)\in R$.

{\em Sufficiency.} Assume (T)--(SC). By (Bd) the preorder $\le$ is a partial order.

We first note that $(a,x)\in R$ implies $a\le x$. Indeed, if $(b,a)\in R$ then $(b,x)\in R$ by transitivity, so $R^{-}(a)\subset R^{-}(x)$ which proves $a\le x$ by the arbitrariness of $b$.

In particular, $x$ is an upper bound of $R^{-}(x)$.
If $y$ is another upper bound and $(a,x)\in R$, then by interpolation there is $b$ with $(a,b), (b,x)\in R$; from $b\in R^{-}(x)$ we get $b\le y$, that is $R^{-}(b)\subset R^{-}(y)$, and $a\in R^{-}(b)$ gives $a\in R^{-}(y)$. Thus $R^{-}(x)\subset R^{-}(y)$, that is $x\le y$, and we conclude
\[
\sup R^{-}(x)=x .
\]
Moreover $R^{-}(x)$ is $\le$-directed: given $a,b\in R^{-}(x)$, condition (D) provides $c$ with $(a,c),(b,c)\in R$ and $(c,x)\in R$, and $(a,c)\in R$ implies $a\le c$, similarly for $b$, while $c\in R^{-}(x)$.

Next we show $R\subset\, \ll$. Let $(a,u)\in R$ and let $S$ be $\le$-directed with supremum $w$ and $u\le w$. Then $a\in R^{-}(u)\subset R^{-}(w)$, so by (SC) there is $s\in S$ with $a\in R^{-}(s)$, hence $a\le s$. This is precisely $a\ll u$.

Consequently $R^{-}(x)$ is a directed subset of $\twoheaddownarrow x$ with supremum $x$, so $(X,\le)$ is a continuous poset, Prop.\ \ref{wwgf}.

Finally $\ll\,\subset R$. Let $a\ll x$. Applying the definition of $\ll$ to the directed set $R^{-}(x)$, whose supremum is $x$, we obtain $c\in R^{-}(x)$ with $a\le c$; then $a\le c$ and $(c,x)\in R$ gives $(a,x)\in R$ by (fR).

For point 2 it suffices to observe that (fR) has been used only in this last paragraph, so that (T), (D), (Bd) and (SC) already yield the partial order, the continuity of $(X,\le)$ and the inclusion $R\subset\,\ll$. Under these four conditions, if (fR) holds then $\ll\,=R$ by the last paragraph, while conversely, if $\ll\,=R$, then (fR) holds because in any poset $x\le z\ll y$ implies $x\ll y$, that is $\ll\circ \le\,\subset\, \ll$.
\qed \end{proof}

The proof shows more: without (fR) the way-below relation is still determined by $R$, and it is a fixed point of a natural operation on relations.

\begin{corollary}[The way-below relation as a fixed point] \label{wbfixed}
Let $R$ satisfy (T), (D), (Bd), (SC), let $\le$ be defined by \eqref{wbord}, and let $\hat R:=R\,\,\circ \le$, namely $(x,y)\in \hat R$ iff there is $z$ with $x\le z$ and $(z,y)\in R$. Then
\begin{itemize}
\item[(a)] $\ll\,=\hat R$, and $\hat R$ is the smallest relation containing $R$ such that $\hat R\,\,\circ \le\,\subset \hat R$;
\item[(b)] $\hat R$ satisfies (T), (D), (Bd), (fR) and (SC), and the order that it induces through \eqref{wbord} is $\le$;
\item[(c)] $\hat{\hat R}=\hat R$, and $\hat R=R$ if and only if (fR) holds.
\end{itemize}
Thus $R\mapsto \hat R$ maps the relations satisfying (T), (D), (Bd), (SC) onto the way-below relations of continuous posets, which are precisely its fixed points.
\end{corollary}

\begin{proof}
(a) If $a\ll x$, the last paragraph of the proof of Thm.\ \ref{wbaxioms} provides $c$ with $a\le c$ and $(c,x)\in R$, so $(a,x)\in \hat R$. Conversely, if $a\le c$ and $(c,x)\in R$, then $c\ll x$ by point 2 of Thm.\ \ref{wbaxioms}, and $a\ll x$. If $R'\supset R$ and $R'\,\circ\le\,\subset R'$ then $\hat R=R\,\circ \le\,\subset R'$, while $\hat R\,\circ\le\,=R\,\circ\le\circ\le\,=\hat R$.
(b) By (a) $\hat R$ is the way-below relation of the continuous poset $(X,\le)$, and by Prop.\ \ref{cnzz} $x\le y$ iff $\hat R^-(x)\subset \hat R^-(y)$. Thus the order induced by $\hat R$ is $\le$, and the necessity part of Thm.\ \ref{wbaxioms} gives the five conditions.
(c) By (b) the order induced by $\hat R$ is $\le$, so $\hat{\hat R}=\hat R\,\circ \le\,=\hat R$. Since $\le$ is reflexive $R\subset \hat R$, so $\hat R=R$ iff $R\,\circ\le\,\subset R$, which is (fR).
\qed \end{proof}

In the spacetime case $R=I$, $\le\,=D_p$, the relation $\hat I=I\circ D_p$ is the common-future relation $C$, and (a) is Thm.\ \ref{contreg}. Condition (fR) is therefore the fixed-point equation $R=R\,\,\circ\le$, while the continuity of the induced order is guaranteed by the remaining four conditions.

Some comments on the meaning of the five conditions are in order.

Conditions (T) and (D) together say exactly that $(X,R)$ is an {\em abstract basis} in the sense of Sec.\ \ref{cmpl}, that is, that $R$ is fully transitive. Indeed, taking $F=\emptyset$ in \eqref{mkcv} gives the non-emptiness of $R^{-}(x)$,
and taking $F$ a doubleton gives the $R$-directedness of $R^{-}(x)$; conversely full transitivity follows from (T) and (D) by induction on the cardinality of $F$. Under (T) and (D) each $R^{-}(x)$ is a round ideal
of $(X,R)$, so that \eqref{wbord} is nothing but the pullback along the canonical map
\[
j: X\to \mathrm{RI}(X), \qquad j(x)=R^{-}(x),
\]
of the inclusion order of the round ideal completion. Condition (Bd) states that $j$ is injective, and the first part of the proof of sufficiency shows that it is then an order embedding. 

Condition (SC) states that $j$ preserves the existing directed suprema: if a directed set $S$ has a supremum in $X$, then the supremum of $j(S)$ in $\mathrm{RI}(X)$, namely $R^-(S)$, belongs to $j(X)$ and equals $j(\sup S)$. The converse is automatic: if $R^-(S)=R^-(q)$ then $q=\sup S$. Since a continuous poset need not be a dcpo, and spacetimes never are, preservation of existing directed suprema is the correct weakening of the requirement that $j$ be onto.

On a spacetime $(M,g)$ with  $R=I$  (T) and (D) hold automatically, because $I$ is transitive, open and non-empty in the past of each point, while (Bd) is past-distinction and (fR) is future reflectivity by Prop.\ \ref{oovt}. A past-distinguishing but non-future reflecting spacetime therefore satisfies (T), (D) and (Bd) but not (fR).
By Thm.  \ref{kssw} these conditions, namely past-distinction, future reflectivity and (SC), with (T) and (D) automatic, are sufficient and necessary for $(M,D_p)$ to be a continuous poset with way-below relation $I$, in agreement with Thm. \ref{wbaxioms}. (SC) is not implied by the others, as Example  \ref{suex} shows, although by Thm.\ \ref{mmwr} it follows from any of several causality conditions, none of which is necessary.


A similar  reading applies to the topological KP-spaces of Sec.\ \ref{rjgh}.

\begin{remark}[Comparison with the approach of Gigli et al. \cite{gigli26}] \label{gigliaxioms}
The axioms of Thm.\ \ref{wbaxioms} locate precisely the framework adopted in
\cite{gigli26} and recalled in Sec.\ \ref{cmpl}. Those authors take as primitive a Lorentzian pre-length space $(X,d,I,\le,\ell)$, where we have written $I$ for the chronological relation that they denote by $\ll$, and introduce, in their Def.\ 3.6, the {\em past chronological inclusion relation} $x\preceq_{I^-} y$ if $I^-(x)\subset I^-(y)$, which is the construction \eqref{wbord} for $R=I$ which we also denote $D_p$. Their standing hypotheses are then
exactly (T), (D), (Bd) and (SC): transitivity of $I$ is part of their Def.\ 2.6; (D) follows from approximability together with their Lemmas 3.3 and 3.5; (Bd) is past-distinction, cf.\ their Remark 3.7; and their supremum-compatibility condition is  (SC).

By point 2 of our Thm.\ \ref{wbaxioms}, their $(X,\preceq_{I^-})$ is a continuous poset and  $I\subset\,\ll$, with $I^{-}(x)$ directed and $\sup I^{-}(x)=x$, so that $I$ approximates the order from below.\footnote{Let us stress that $I$ is not, in general, an auxiliary relation in the sense of \cite[Def.\ I-1.11]{gierz03}. The absorption condition required there, $u\le x$, $(x,y)\in I$, $y\le v$ imply $(u,v)\in I$, splits into two halves: the half $(x,y)\in I$ and $y\le v$ imply $(x,v)\in I$ is Prop.\ \ref{wbauto} and holds automatically, while the half $u\le x$ and $(x,y)\in I$ imply $(u,y)\in I$ is precisely (fR).} What is missing there is  (fR), that is future reflectivity. Without it the way-below relation of $\preceq_{I^-}$ can be strictly larger than $I$, so that the chronological relation, although approximating, is not identified with the way-below relation. By Cor.\ \ref{wbfixed} it is the hull $I\circ\preceq_{I^-}$. This does not affect their completion: by Thm.\ \ref{ipmatch} and Lemma \ref{pastiso} the round ideals of $\ll$ correspond bijectively to the indecomposable past sets of $I$.
\end{remark}

\section{Conclusions}

We have carried out the  four-fold program announced in the Introduction.

First, we gave a complete, self-contained treatment of the fragment of continuous poset theory needed for causality theory, filling in proofs that \cite{martin06} could only summarize.

Second, we settled the question of which causal relation on a regular spacetime should play the role of the order $\le$. The answer is unambiguous: continuity of $(M,\le)$ with way-below relation $I$ forces $\le\,=D_p$ (Thm.\ \ref{kssw}), a choice available precisely under past-distinction, (SC) and future-reflectivity; the same analysis identifies bicontinuous posets with causally continuous spacetimes (Thm.\ \ref{bxxd}), and globally hyperbolic posets with globally hyperbolic spacetimes, this time with $\le\,=J$.
The Lawson, interval (biScott), manifold and Alexandrov topologies are all shown to coincide for bicontinuous posets (Thm.\ \ref{dzxd} and Thm.\ \ref{kxqp}).

Third, we recovered at the abstract level the top three rungs of the causal ladder. Causal simplicity is invisible to the poset $(M,D_p)$, since the order is there closed by construction (as causal continuity implies $D_p=D=\bar J$), and we showed that it is restored by the chain realizability condition of Def.\ \ref{def:chainrel}: any two related points are to be joined by a chain compact and connected in the interval topology. The relevant results were established for Nachbin  semiclosed ordered spaces, where they involve no approximation relation at all (Thm.\ \ref{thm:ghimpliescs}), and they showed that, for compact intervals, chain realizability, density of the order and connectedness of the intervals coincide (Thm.\ \ref{thm:densityconnected}). It is very interesting that connectedness and compactness of the intervals $[a,b]$, jointly with a second countability assumption, imply the existence of connecting chains that are actually curves, and that this statement can be interpreted as ``global hyperbolicity implies causal simplicity''.


In the end we got a ladder of abstract causality properties that returns the classical one when applied to a spacetime $(M,g)$.
It has five rungs, all distinct (Table \ref{tab:ladder}), and includes a level which has no counterpart in the poset literature: ``stable causality and future reflectivity'' corresponds to continuity together with Hausdorffness of the Alexandrov-like topology $\tau$ (Thm.\ \ref{prop:spacesta}). Its interest is that it is exactly the rung at
which the abstract analogue of a time function becomes available, Thm.\ \ref{levthf}. This fully solves the problem raised by Martin and Panangaden of the determination of the causality levels in the abstract poset theory \cite[Conclusions]{martin06}.

 Fourth, the topological KP-space framework of Sec.\ \ref{rjgh} allows one to transport the whole correspondence to low-regularity Lorentzian geometry, uniformly for the two main synthetic proposals in the literature -- the Lorentzian length spaces of Kunzinger and S\"amann and the BMS-Lorentzian metric/length spaces -- under hypotheses  that are often automatic or already built into the respective definitions.

 We showed that this order-theoretic viewpoint has several advantages: it gives access, essentially for free, to the completion machinery of domain theory. The round-ideal completion of Lawson and Keimel--Lawson, applied to $(M,D_p)$, reproduces exactly the future causal boundary of Geroch, Kronheimer and Penrose (Thm.\ \ref{ipmatch}), with the terminal indecomposable past sets recovered as the added ideal points and, as a bonus, with the chronological relation itself surviving as the restriction of the completion's way-below relation to the original spacetime (Thm.\ \ref{roundidealprop}). The very recent, independently motivated directed completion of Gigli et al.\ \cite{gigli26} arrives to an analogous equivalence.  In their setting we show that strong causality is sufficient to get (SC); there is no need to assume global hyperbolicity or forward completeness.

Lastly, Sec.\ \ref{wbfound} closed the circle by taking the approximation relation, rather than the order, as the primitive datum. Thm.\ \ref{wbaxioms} isolated the five conditions (T), (D), (Bd), (fR) and (SC) that a relation $R$ on a set must satisfy in order that the induced order be a continuous poset with $\ll\,=R$, and showed that (fR) is exactly what separates the identification $\ll\,=R$ from the weaker inclusion $R\subset\,\ll$; without (fR) the way-below relation is the hull $R\,\,\circ\le$, a fixed point of the construction (Cor.\ \ref{wbfixed}). Two of the five are past-distinction and future-reflectivity, so the causality conditions of Sec.\ \ref{qcft} are not an accident of Lorentzian geometry but the relational content of continuity itself.

Taken together, these results change the status of the domain-theoretic reading of spacetime. It is not one interpretation among others, to be tried out on this or that causal relation: once the way-below relation is identified with $I$, the order is forced, being $D_p$, and the causality conditions that make the identification possible are forced as well. Once this rigidity is recognized, causality theory and continuous poset theory become two descriptions of the same object, and results can be transported in both directions. We used the transport from order theory to geometry in the completions, where the round-ideal completion reproduces the causal boundary without any additional construction; and in the reverse direction to obtain, from the geometric picture of causal curves, the abstract chain realizability condition, and from it a purely order-theoretic causal ladder valid for the low-regularity models. The framework is not tied to the smooth category, and this is where we expect it to be most useful: what is required of a spacetime is a relation satisfying five conditions, none of which mentions a manifold, a metric or a curve. Whether the fundamental structure will turn out to be a continuous poset or, as the discussion of black hole evaporation suggests, its dual, the order-theoretic language appears to be the right one in which to ask the question.


A number of questions remain open.
A natural one is whether the one-sided completions of this paper can be upgraded to a genuine {\em bicompletion} for bicontinuous posets.
We are not aware of any domain-theoretic construction that, for the interval (biScott) topology, plays the role that the round-ideal completion plays for the Scott topology. We suspect that isolating such a construction would be of independent interest in domain theory, quite apart from its application to causality.

Since the entire development is naturally phrased in terms of the frame of Scott-open sets, it would be interesting to carry it out one level of abstraction further, directly on locales rather than on points, in the spirit of the point-free causal boundary recently proposed by van der Schaaf \cite{vanderschaaf24}; relating that construction to the round-ideal completion of Thm.\ \ref{roundidealprop} is left for future work.

Finally, it would be interesting to incorporate the metrical aspects of the Lorentzian distance—not just causality—into the framework of continuous posets. This can be achieved in a neat and satisfactory manner, and an account will be presented in a forthcoming work.


\subsection*{Declarations}

\subsubsection*{Data availability statement}
Data sharing not applicable to this article as no datasets were generated or analyzed during the current study.

\subsubsection*{Conflicts of Interest}
The authors has no relevant financial or non-financial interests to disclose.



\begin{thebibliography}{10}

\bibitem{abramsky94}
Abramsky, S. and Jung, A.: Domain theory.
\newblock In S.~Abramsky, D.~M. Gabbay, and T.~S.~E. Maibaum, editors,
  \emph{Handbook of Logic in Computer Science}, Oxford: Oxford University
  Press, vol.~3, pages 1--168 (1994)

\bibitem{beran25}
Beran, T., Harvey, J., Napper, L., and Rott, F.: A {T}oponogov globalisation
  result for {L}orentzian length spaces.
\newblock Mathematische Annalen \textbf{392}, 3447--3478 (2025)

\bibitem{braun26}
Braun, M.: Spacetime reconstruction by order and number.
\newblock Class. Quantum Grav. \textbf{43}, 045015 (2026)

\bibitem{budic74}
Budic, R. and Sachs, R.~K.: Causal boundaries for general relativistic
  spacetimes.
\newblock J. Math. Phys. \textbf{15}, 1302--1309 (1974)

\bibitem{burgos23}
Burgos, S., Flores, J., and Herrera, J.: The $c$-completion of {L}orentzian
  metric spaces.
\newblock Class. Quantum Grav. \textbf{40}, 205013 (2023)

\bibitem{burtscher25}
Burtscher, A. and {Garc{\'\i}a--Heveling}, L.: Time functions on {L}orentzian
  length spaces.
\newblock Ann. {H}enri {P}oincar{\'e} \textbf{26}, 1533--1572 (2025)

\bibitem{minguzzi25b}
Bykov, A. and Minguzzi, E.: Global hyperbolicity and manifold topology from the
  {L}orentzian distance.
\newblock Lett. Math. Phys. \textbf{116}, 62 (2026).
\newblock ArXiv:2503.04382

\bibitem{minguzzi24b}
Bykov, A., Minguzzi, E., and Suhr, S.: Lorentzian metric spaces and
  {GH}-convergence: the unbounded case.
\newblock Lett. Math. Phys. \textbf{115}, 63 (2025).
\newblock ArXiv:2412.04311

\bibitem{dieckmann88}
Dieckmann, J.: Volume functions in general relativity.
\newblock Gen. Relativ. Gravit. \textbf{20}, 859--867 (1988)

\bibitem{ebrahimi14}
Ebrahimi, N.: Causally simple spacetimes and domain theory.
\newblock {\,C}ankaya University Journal of Science and Engineering
  \textbf{11}, 1--6 (2014)

\bibitem{ebrahimi15}
Ebrahimi, N.: Domain theory in general relativity.
\newblock J. Dyn. Sys. Geom. Theor. \textbf{13}, 1--41 (2015)

\bibitem{eilenberg41}
Eilenberg, S.: Ordered topological spaces.
\newblock American {J}ournal of {M}athematics \textbf{63}, 39--45 (1941)

\bibitem{engelking89}
Engelking, R.: \emph{General Topology}.
\newblock Berlin: Helderman {V}erlag (1989)

\bibitem{finster21}
Finster, F., Much, A., and Papadopoulos, K.: \emph{On Global Hyperbolicity of
  Spacetimes: Some Recent Advances and Open Problems}, Cham, Switzerland:
  Springer, pages 281--295 (2021).
\newblock Mathematical Analysis in Interdisciplinary Research, I. N. Parasidis,
  E. Providas and T. M. Rassias eds.

\bibitem{flores06b}
Flores, J.~L.: The causal boundary of spacetimes revisited.
\newblock Commun. Math. Phys. \textbf{276}, 611--643 (2007)

\bibitem{flores11}
Flores, J.~L., Herrera, J., and S{\'a}nchez, M.: On the final definition of the
  causal boundary and its relation with the conformal boundary.
\newblock Adv. Theor. Math. Phys. \textbf{15}, 991--1057 (2011)

\bibitem{heveling21}
Garc{\'\i}a-Heveling, L.: Causality theory of spacetimes with continuous
  {L}orentzian metrics revisited.
\newblock Class. Quantum Grav. \textbf{38}, 145028 (2021)

\bibitem{geroch70}
Geroch, R.: Domain of dependence.
\newblock J. Math. Phys. \textbf{11}, 437--449 (1970)

\bibitem{geroch72}
Geroch, R., Kronheimer, E.~H., and Penrose, R.: Ideal points in spacetime.
\newblock Proc. {R}oy. {S}oc. {L}ond. {A} \textbf{327}, 545--567 (1972)

\bibitem{gierz03}
Gierz, G., Hofmann, K.~H., Keimel, K., Lawson, J.~D., Mislove, M.~W., and
  Scott, D.~S.: \emph{Continuous lattices and domains}.
\newblock Cambridge {U}niversity {P}ress (2003)

\bibitem{gigli25}
Gigli, N.: Hyperbolic {B}anach spaces {I} (2025).
\newblock {arXiv:}2503.10467

\bibitem{gigli26}
Gigli, N., Ohanyan, A., Picerni, M., Xu, Z.-F., and Zanardini, M.: Ideal
  points, directed completion and the case of the maximally extended
  {S}chwarzschild spacetime (2026).
\newblock ArXiv:2608.04718

\bibitem{grant18}
Grant, J. D.~E., Kunzinger, M., and S{\"a}mann, C.: Inextendibility of
  spacetimes and {L}orentzian length spaces.
\newblock Ann. Glob. Anal. Geom. \textbf{55}, 133--147 (2019)

\bibitem{harris98}
Harris, S.~G.: Universality of the future chronological boundary.
\newblock J. Math. Phys. \textbf{39}, 5427--5445 (1998)

\bibitem{ake26}
Hau, L.~A., Burgos, S., and Solis, D.~A.: The causal structure of the
  $c$-completion of warped spacetimes.
\newblock Gen. Relativ. Gravit. \textbf{58}, 13 (2026)

\bibitem{ake20}
Hau, L.~A., Pacheco, A. J.~C., and Solis, D.~A.: On the causal hierarchy of
  {L}orentzian length spaces.
\newblock Class. Quantum Grav. \textbf{37}, 215013 (2020)

\bibitem{hawking68}
Hawking, S.~W.: The existence of cosmic time functions.
\newblock Proc. {R}oy. {S}oc. {L}ondon, series {A} \textbf{308}, 433--435
  (1968)

\bibitem{hawking73}
Hawking, S.~W. and Ellis, G. F.~R.: \emph{The Large Scale Structure of
  Space-Time}.
\newblock Cambridge: Cambridge {U}niversity {P}ress (1973)

\bibitem{hawking74}
Hawking, S.~W. and Sachs, R.~K.: Causally continuous spacetimes.
\newblock Commun. Math. Phys. \textbf{35}, 287--296 (1974)

\bibitem{keimel09b}
Keimel, K.: Bicontinuous domains and some old problems in {D}omain theory.
\newblock Electronic Notes in Theoretical Computer Science \textbf{257}, 35--54
  (2009)

\bibitem{keimel09}
Keimel, K. and Lawson, J.~D.: D-completions and the d-topology.
\newblock Annals of Pure and Applied Logic \textbf{159}, 292--306 (2009)

\bibitem{kodama79}
Kodama, H.: Inevitability of a naked singularity associated with the black hole
  evaporation.
\newblock Prog. Theor. Phys. \textbf{62}, 1434--1435 (1979)

\bibitem{kronheimer67}
Kronheimer, E.~H. and Penrose, R.: On the structure of causal spaces.
\newblock Proc. {C}amb. {P}hil. {S}oc. \textbf{63}, 482--501 (1967)

\bibitem{kunzinger18}
Kunzinger, M. and S{\"a}mann, C.: Lorentzian length spaces.
\newblock Ann. Global Anal. Geom. \textbf{54}, 399--447 (2018)

\bibitem{lawson91}
Lawson, J.~D.: \emph{Order and strongly sober compactifications}, Oxford:
  Clarendon Press, vol. Topology and Category Theory in Computer Science, pages
  171--206 (1991)

\bibitem{lawson97}
Lawson, J.~D.: The round ideal completion via sobrification.
\newblock Topology Proceedings \textbf{22 Summer}, 261--274 (1997)

\bibitem{lesourd19}
Lesourd, M.: Causal structure of evaporating black holes.
\newblock Class. Quantum Grav. \textbf{36}, 025007 (2019)

\bibitem{levin83}
Levin, V.~L.: A continuous utility theorem for closed preorders on a
  $\sigma$-compact metrizable space.
\newblock Soviet Math. Dokl. \textbf{28}, 715--718 (1983)

\bibitem{malament77}
Malament, D.~B.: Causal theories of time and the conventionality of
  simultaneity.
\newblock No\^us \textbf{11}, 293--300 (1977)

\bibitem{malament12}
Malament, D.~B.: \emph{Topics in the foundations of general relativity and
  {N}ewtonian gravitation theory}.
\newblock Chicago Lectures in Physics. University of Chicago Press, Chicago, IL
  (2012)

\bibitem{markowsky76}
Markowsky, G.: Chain-complete posets and directed sets with applications.
\newblock Algebra {U}niv. \textbf{6}, 53--68 (1976)

\bibitem{marolf03}
Marolf, D. and Ross, S.~F.: A new recipe for causal completions.
\newblock Class. Quantum Grav. \textbf{20}, 4085--4117 (2003)

\bibitem{martin06}
Martin, K. and Panangaden, P.: A domain of spacetime intervals in general
  relativity.
\newblock Commun. Math. Phys. \textbf{267}, 563--586 (2006)

\bibitem{mazibuko23}
Mazibuko, L., Baboolal, D., and Goswami, R.: Causal structure of spacetime and
  {S}cott topology.
\newblock Afrika Matematika \textbf{34}, 94 (2023)

\bibitem{minguzzi07b}
Minguzzi, E.: The causal ladder and the strength of {$K$}-causality. {I}.
\newblock Class. Quantum Grav. \textbf{25}, 015009 (2008).
\newblock {arXiv}:0708.2070

\bibitem{minguzzi07}
Minguzzi, E.: The causal ladder and the strength of {$K$}-causality. {II}.
\newblock Class. Quantum Grav. \textbf{25}, 015010 (2008).
\newblock {arXiv}:gr-qc/0703128

\bibitem{minguzzi07e}
Minguzzi, E.: Weak distinction and the optimal definition of causal continuity.
\newblock Class. Quantum Grav. \textbf{25}, 075015 (2008).
\newblock {arXiv}:0712.0338

\bibitem{minguzzi08e}
Minguzzi, E.: Characterization of some causality conditions through the
  continuity of the {L}orentzian distance.
\newblock J. Geom. Phys. \textbf{59}, 827--833 (2009).
\newblock {arXiv}:0810.1879

\bibitem{minguzzi09c}
Minguzzi, E.: Time functions as utilities.
\newblock Commun. Math. Phys. \textbf{298}, 855--868 (2010).
\newblock {arXiv}:0909.0890

\bibitem{minguzzi11c}
Minguzzi, E.: Topological conditions for the representation of preorders by
  continuous utilities.
\newblock Appl. {G}en. Topol. \textbf{13}, 81--89 (2012).
\newblock {arXiv}:1108.5123

\bibitem{minguzzi13d}
Minguzzi, E.: Convex neighborhoods for {L}ipschitz connections and sprays.
\newblock Monatsh. Math. \textbf{177}, 569--625 (2015).
\newblock {arXiv}:1308.6675

\bibitem{minguzzi18b}
Minguzzi, E.: Lorentzian causality theory.
\newblock Living Rev. Relativ. \textbf{22}, 3 (2019)

\bibitem{minguzzi20}
Minguzzi, E.: A gravitational collapse singularity theorem consistent with
  black hole evaporation.
\newblock Lett. Math. Phys. \textbf{110}, 2383--2396 (2020)

\bibitem{minguzzi26c}
Minguzzi, E.: Global hyperbolicity meets order completeness (2026).
\newblock {arXiv:}2608.03476

\bibitem{minguzzi26}
Minguzzi, E.: Results on lorentzian metric spaces.
\newblock Gen. Relativ. Gravit. \textbf{58}, 9 (2026).
\newblock {arXiv:}2510.24423

\bibitem{minguzzi06c}
Minguzzi, E. and S\'anchez, M.: \emph{The causal hierarchy of spacetimes},
  Zurich: Eur. Math. Soc. Publ. House, vol. H. Baum, D. Alekseevsky (eds.),
  Recent developments in pseudo-Riemannian geometry of \emph{{ESI} Lect. Math.
  {P}hys.}, pages 299--358 (2008).
\newblock ArXiv:gr-qc/0609119

\bibitem{minguzzi22}
Minguzzi, E. and Suhr, S.: Lorentzian metric spaces and their
  {G}romov-{H}ausdorff convergence.
\newblock Lett. Math. Phys. \textbf{114}, 73 (2024).
\newblock {arXiv:}2209.14384

\bibitem{munkres00}
Munkres, J.~R.: \emph{Topology (2nd ed.)}.
\newblock Upper Saddle river, New Jersey: Prentice-Hall, Inc. (2000)

\bibitem{nachbin65}
Nachbin, L.: \emph{Topology and order}.
\newblock Princeton: D.\ {V}an {N}ostrand {C}ompany, {I}nc. (1965)

\bibitem{ohanyan26}
Ohanyan, A. and Candal, M.~S.: Timelike {R}icci curvature lower bounds via
  optimal transport for {O}rlicz-type {L}orentzian costs (2026).
\newblock ArXiv:2604.22538

\bibitem{peleg70}
Peleg, B.: Utility functions for partially ordered topological spaces.
\newblock Econometrica \textbf{38}, 93--96 (1970)

\bibitem{penrose72}
Penrose, R.: \emph{Techniques of Differential Topology in Relativity}.
\newblock Cbms-{N}sf Regional Conference Series in Applied Mathematics.
  Philadelphia: {SIAM} (1972)

\bibitem{penrose79}
Penrose, R.: \emph{Singularities and time-asymmetry}, Cambridge: Cambridge
  {U}niversity {P}ress, vol. S. W. Hawking and W. Israel (ed.) General
  relativity: {A}n {E}instein centenary survey, pages 581--638 (1979)

\bibitem{sharifzadeh19}
Sharifzadeh, M. and Abad, M. B.~S.: Globally hyperbolic spacetimes as posets.
\newblock Math Phys Anal Geom \textbf{22}, 25 (2019)

\bibitem{sorkin96}
Sorkin, R.~D. and Woolgar, E.: A causal order for spacetimes with {$C^0$}
  {L}orentzian metrics: proof of compactness of the space of causal curves.
\newblock Class. Quantum Grav. \textbf{13}, 1971--1993 (1996)

\bibitem{surya19}
Surya, S.: The causal set approach to quantum gravity.
\newblock Living Reviews in Relativity \textbf{22}, 5 (2019)

\bibitem{szabados88}
Szabados, L.~B.: Causal boundary for strongly causal spacetimes.
\newblock Class. Quantum Grav. \textbf{5}, 121--134 (1988)

\bibitem{vanderschaaf24}
van~der Schaaf, N.: \emph{Towards Point-Free Spacetimes}.
\newblock Ph.D. thesis, University of Edinburgh (2024).
\newblock ArXiv:2406.15406

\bibitem{wald84b}
Wald, R.~M.: \emph{Black holes, singularities and predictability}, Boca Raton,
  FL: CRC Press, vol. Quantum Theory of Gravity. Essays in Honor of the 60th
  Birthday of {B}ryce {S.} {D}e{W}itt, S. M. Christensen (ed.), pages 160--168
  (1984)

\bibitem{wallace45}
Wallace, A.~D.: A fixed-point theorem.
\newblock Bull. Amer. Math. Soc.  (1945)

\bibitem{wallace62}
Wallace, A.~D.: Relations on topological spaces.
\newblock In \emph{General Topology and its Relations to Modern Analysis and
  Algebra (Prague, 1961)}. Academia (1962), pages 356--360

\bibitem{ward54}
{Ward, Jr.}, L.~E.: Partially ordered topological spaces.
\newblock Proc. Am. Math. Soc. \textbf{5}, 144--161 (1954)

\bibitem{willard70}
Willard, S.: \emph{General topology}.
\newblock Reading: {Addison-Wesley} {P}ublishing {C}ompany (1970)

\bibitem{williamson87}
Williamson, R. and Janos, L.: Constructing metrics with the {H}eine-{B}orel
  property.
\newblock Proc. Am. Math. Soc. \textbf{100}, 567--573 (1987)

\bibitem{zhao10}
Zhao, D. and Fan, T.: Dcpo-completion of posets.
\newblock Theoretical {C}omputer {S}cience \textbf{411}, 2167--2173 (2010)

\end{thebibliography}

\end{document}